%% file: main.tex
\documentclass[11pt]{article}

\usepackage[utf8]{inputenc}
\usepackage[T1]{fontenc}
\usepackage[a4paper,margin=1in]{geometry}
\usepackage{amsmath,amssymb,amsthm,mathtools}
\usepackage{microtype}
\usepackage{enumitem}
\usepackage{hyperref}
\usepackage{booktabs}
\usepackage{cleveref,aliascnt}
\usepackage{xspace}
\usepackage{thm-restate,color,xcolor}
\usepackage{amsmath,amssymb,bbm,amsthm}
\usepackage{fullpage}
\usepackage{float}
\usepackage{boxedminipage}
\usepackage{amsmath}
\usepackage{mathtools}
\usepackage{framed}
\usepackage[framemethod=tikz]{mdframed}
\usepackage{titlesec}
\usepackage{lipsum}%
\usepackage{cleveref,aliascnt}
\usepackage{tikz}
\usepackage{multirow}
\usepackage{comment}

\usepackage{wrapfig}
\usepackage[ruled,vlined,linesnumbered]{algorithm2e}
\crefname{algocf}{algorithm}{algorithms}
\Crefname{algocf}{Algorithm}{Algorithms}

\usepackage{mathtools}

\newtheorem{theorem}{Theorem}[section]
\newtheorem{lemma}[theorem]{Lemma}
\newtheorem{question}[theorem]{Question}
\newtheorem{proposition}[theorem]{Proposition}
\newtheorem{claim}[theorem]{Claim}
\newtheorem{definition}[theorem]{Definition}
\newtheorem{remark}[theorem]{Remark}
\newtheorem{corollary}[theorem]{Corollary}
\newtheorem{observation}[theorem]{Observation}

\newcommand{\V}{V}

\newcommand{\G}{G}

\newcommand{\dist}{\mathrm{dist}}
\newcommand{\poly}{\mathrm{poly}}

\newcommand{\Gg}{\overline{\mathcal{K}}}
\newcommand{\Ggh}{\mathcal{K}}
\newcommand{\Ll}{\mathcal{R}}

\newcommand{\sco}{\mathsf{sc}}

\newcommand{\lsco}{\mathsf{lsc}}

\newcommand{\pp}{\mathcal{P}}

\newcommand{\prepp}{\widehat{\pp}}

\newcommand{\RP}{$\mathsf{RPC}$\xspace}

\title{New Fault-Tolerant $(\alpha,\beta)$-Spanners}
\author{%
  Merav Parter\thanks{%
    Email: \texttt{merav.parter@weizmann.ac.il}.
    Supported by ERC grant 949083 under the European Union's
    Horizon 2020 research and innovation programme.%
  }
  \quad
  Elizaveta Popova\thanks{%
    Email: \texttt{elizaveta.popova@weizmann.ac.il}.
    Research supported in part by ISF grant 2364/25
    and Minerva grant 715025.%
  }
  \quad
  Elad Tzalik\thanks{%
    Email: \texttt{elad.tzalik@weizmann.ac.il}.
    Supported by the Adams Fellowship Program of the
    Israel Academy of Sciences and Humanities.%
  }\\[0.5em]
  Weizmann Institute of Science, Israel
}
\date{}

\begin{document}
\maketitle

\begin{abstract}
We present new constructions of vertex fault-tolerant (VFT) $(\alpha,\beta)$-spanners with substantially improved size bounds. Given a bound $f$ on the number of vertex faults, an $f$-VFT $(\alpha,\beta)$-spanner is a subgraph of $G$ that preserves an $(\alpha,\beta)$-spanner of $G- F$ for every fault set $F \subseteq V$ with $|F| \le f$. Despite extensive study, fault-tolerant spanners are well understood primarily in the purely multiplicative regime, i.e., when $\beta = 0$. In contrast, existing constructions of VFT $(\alpha,\beta)$-spanners rely on a black-box sampling technique that incurs an exponential blow-up, either in $\beta/\alpha$ or in $f$, in the size of the spanner compared to their non-fault-tolerant counterparts.

In this work, we develop two independent techniques for constructing $f$-VFT $(\alpha,\beta)$-spanners that improve the dependence on $f$, reducing it to a \emph{polynomial} (rather than exponential) overhead. In particular, for every unweighted graph $G=(V,E)$ on $n$ vertices, stretch parameter $k$, and fault parameter $f$, we show:

\begin{itemize}
    \item $f$-VFT $(O_{\epsilon}(k^{\epsilon}), O_{\epsilon}(k))$-spanners with $\tilde{O}(f^{exp(1/\epsilon)}\cdot n^{1+1/k})$ edges for any $\epsilon \in (0,1)$. For any fixed constant $\epsilon$, this introduces a $\poly(f)$ overhead compared to the non-faulty constructions of [Ben-Levy and Parter, SODA 2020] and [Chechik and Lifshitz, SODA 2026].

   \item For any $\delta \in (0,1)$, an $f$-VFT $\left(O_{\delta}(\sqrt{k}),O_{\delta}(k) \right)$-spanner of size $\tilde{O}(f^{1+\delta} \cdot n^{1+\frac{1}{k}})$. Hence, it incurs a size overhead of only $f^{1+\delta}$
    compared to the above mentioned non-faulty analogs.
    
    \item $f$-VFT $(k,k-1)$ spanner of size $O(k^3 \cdot f  \cdot n^{1+\frac{1}{k}})$, nearly matching the lower bound $\Omega(f^{1-\frac{1}{k}}n^{1+\frac{1}{k}})$ known for the less constrained notion of $f$-VFT $(2k-1,0)$ spanners by [Bodwin et al., SODA 2018].

\end{itemize}
 The first result is obtained via a new adaptation of the Baswana--Sen framework to the VFT $(\alpha,\beta)$ setting. The other two bounds follow from a highly involved analysis of a natural extension of the fault-tolerant greedy spanner by [Bodwin and Patel, PODC 2019] to the VFT $(\alpha,\beta)$-setting. 

Our main technical contribution shows that the dynamics of the greedy algorithm naturally induce robust clusters, in which each clustered vertex is connected to the set of centers via a sufficiently well-spread collection of paths. This structure tolerates not only $f$ adversarial faults, but also up to $n - O\!\left(\frac{n}{f}\right)$ random faults, a property we leverage to obtain improved size bounds. We believe that this strengthened notion of fault tolerance may be of independent interest and find further applications.
\end{abstract}

\newpage
\tableofcontents
\newpage

\input{intro}

\input{technical-overview}

\input{VFT-exp}

\input{FT-Clustering-sec}

\input{Sqrtk-Sec}

\input{Section2-2k}

\bibliographystyle{alpha}
\bibliography{refs}
\appendix

\input{Sampling-lemma}

\input{Equivalence-d-to-r-alpha-beta}

\end{document}

%% file: intro.tex
\section{Introduction}\label{sec:intro}

We present new constructions of fault-tolerant $(\alpha,\beta)$-spanners with significantly improved size bounds. 
An $(\alpha,\beta)$-spanner of an unweighted $n$-vertex graph $G=(V,E)$ is a subgraph $H$ satisfying
\[
    \dist_H(u,v) \le \alpha \cdot \dist_G(u,v) + \beta
\]
for every pair of vertices $u,v \in V$~\cite{ElkinP01}. 
Since the seminal works of Peleg and Sch{\"a}ffer~\cite{PelegS89} and Alth{\"o}fer et al.~\cite{AlthoferDDJS93}, understanding the trade-off between stretch and size has been a central question in graph algorithms. 
For purely multiplicative $(2k-1,0)$-spanners, the classical bound of $O(n^{1+1/k})$ edges is conditionally optimal under the Erd\H{o}s girth conjecture. 
By contrast, in other regimes, such as $(1+\epsilon,\beta)$-spanners, the optimal size--stretch trade-off remained elusive for a long time~\cite{ElkinP01,ElkinP04,BaswanaKMP05,AbboudBP17,ElkinN20}. 
A breakthrough result of Tan and Zhang~\cite{TanZ23} gave an almost optimal construction for every constant $k$: for any $k \ge 2$, any constant $\delta>0$, and any $\epsilon>0$, they construct a $(1+\epsilon,\beta)$-spanner with
\[
    \beta =
    \frac{(k-1)^{k-1}\cdot C_{\delta,k}^k}
         {k^k \cdot \epsilon^{k-1}}
\]
and $O\!\left(n^{1+(1+\delta)/(2^{k+1}-1)}\right)$ edges, for a sufficiently large constant $C_{\delta,k}$. 
The picture for non-constant stretch parameter $k$, however, remains far from complete.

In this work, we focus on a notable class of $(O_{\epsilon}(k^{\epsilon}),O_{\epsilon}(k))$-spanners, for $\epsilon \in (0,1)$, introduced by Ben-Levy and Parter~\cite{Ben-LevyP20}. 
Their construction has size $O_{\epsilon}(n^{1+1/k})$ and admits the following useful interpretation: \emph{distant} pairs, namely pairs at distance $\Omega_{\epsilon}(k^{1-\epsilon})$ in $G$, incur only multiplicative stretch $O_{\epsilon}(k^{\epsilon})$, rather than the $(2k-1)$ stretch guaranteed by classical multiplicative spanners. 
For every fixed constant $\epsilon$, this guarantee is optimal up to constant factors by unconditional girth-based lower bounds. 
More recently, Chechik and Lifshitz~\cite{ChechikL26} obtained a stronger distance-sensitive guarantee, showing that pairs at distance $d$ satisfy
\[
    \dist_H(u,v) \le (2k-1) + O(d \log d).
\]
In particular, their result implies the existence of $(O(\log k),O(k))$-spanners with $O(n^{1+1/k}+kn)$ edges.

We study $(\alpha,\beta)$-spanners in settings where the underlying graph $G$ may experience failures, focusing on the more challenging case of vertex failures. Our goal is to design fault-tolerant spanners that preserve their metric guarantees, namely, the $(\alpha,\beta)$ stretch, even after the failure of a bounded number of vertices. 
\begin{definition}
For a given unweighted graph $G=(V,E)$ and parameters $\alpha,\beta$ and $f$, a subgraph $H \subseteq G$ is an $f$-\emph{vertex fault-tolerant (VFT)} $(\alpha,\beta)$-\emph{spanner} of $G$ if 
\[
\dist_{H - F}(u,v) \le \alpha \cdot \dist_{G - F}(u,v) + \beta, \mbox{~for all~} u,v \in V, F \subseteq V-\{u,v\}, |F|\leq f~.
\]
\end{definition}

Fault-tolerant spanners in general graphs have been studied extensively in the purely multiplicative regime, i.e., when $\beta = 0$. A substantial body of work~\cite{ChechikLPR10,DinitzK11,BodwinDPW18,BodwinP19,BodwinDR22,Parter22}, culminating in~\cite{BodwinDPW18,BodwinP19}, shows that the optimal size of $f$-VFT $(2k-1,0)$-spanners is $O(f^{1-1/k} n^{1+1/k})$ edges. More generally, this regime is by now well understood: tight bounds have been obtained for a range of fault models: edge faults~\cite{BodwinDR22}, bounded-degree faults~\cite{BodwinHP24}, color faults~\cite{PetruschkaST24_ITCS}, and emulators~\cite{BodwinDN23}, and are matched, up to lower-order terms, by several different constructions (see, e.g.,~\cite{DinitzR20,BodwinDR21,Parter22,PPST25}).

In contrast, our understanding of vertex fault-tolerant $(\alpha,\beta)$-spanners remains quite limited. To date, essentially a single general technique is known by Braunschvig et al. ~\cite{BraunschvigCPS15}: sample $\gamma$ subgraphs $G_1,\ldots,G_\gamma$ by including each $G$-edge independently with probability $p$ in each $G_i$, compute a (fault-free) $(\alpha,\beta)$-spanner $H_i$ for each $G_i$, and output their union $H=\bigcup_i H_i$ as the $f$-VFT spanner. This method, known as \emph{FT-sampling}~\cite{weimann2013replacement,DinitzK11}, inherently requires $\gamma$ to grow exponentially in either $f$ or $\beta/\alpha$.

Consequently, the following question has been open:

\begin{question}\label{q:alpha-beta}
    Is it possible to provide a construction of $f$-VFT $(\alpha,\beta)$-spanners that improve over the size overhead of $a^b$ where $a=\max\{\beta/\alpha, f\}$ and $b=\min\{\beta/\alpha, f\}$?
\end{question}

We answer this question in the affirmative by substantially extending two distinct approaches that have been central to the study of FT multiplicative spanners to the $(\alpha,\beta)$ setting. This extension requires significantly more delicate and involved arguments. At a high level, $(\alpha,\beta)$-spanners are inherently more challenging, as their stretch guarantee does not reduce to a local edge condition as in the purely multiplicative case. Instead, one must ensure the stretch bound over entire paths, whose relevant length scale is $O(\beta/\alpha)$; this difficulty is further amplified in the presence of faults. 

Adopting a Baswana--Sen type perspective, we circumvent the exponential-in-$\min\{f,\beta/\alpha\}$ barrier inherent in the current FT-sampling approach and obtain near-optimal constructions for $(O_{\varepsilon}(k^{\epsilon}),O_{\varepsilon}(k))$-spanners for any constant $\epsilon \in (0,1)$ and $f=\tilde{O}(1)$. In particular, we show:

\begin{theorem}[Informal]\label{thm:alpha-beta-polyf}
    Any $n$-vertex graph $G=(V,E)$ admits an $f$-VFT $(O(k^{\epsilon}),O(k))$-spanner, for any fixed constant $\epsilon \in (0,1)$, with $\poly(f) \cdot n^{1+1/k}$ edges. 
 \end{theorem}
This Baswana–Sen perspective also yields a significantly simpler alternative to the non-faulty construction of \cite{Ben-LevyP20}, albeit at the cost of slightly weaker size bounds.

Our second set of results is obtained via a natural adaptation of the classical greedy spanner algorithm of \cite{PopovaT26}: for every fault set $F \in V^{\le f}$, we add short paths in $G - F$ whenever the desired stretch guarantee is violated. While the correctness of this procedure is immediate, its size analysis is considerably more subtle. We develop a new analysis of FT-greedy constructions that applies even when $\beta/\alpha$ is unbounded.  Our analysis combines two techniques that have largely evolved independently in the literature: the blocking-set method of~\cite{BodwinP19} and the parks framework of~\cite{PPST25}. By coupling the dynamics of the greedy process with a suitable clustering scheme, we show that strong fault-tolerant clusters emerge naturally over the course of the algorithm. Concretely, we prove that the FT-greedy dynamics generates well-spread collections of paths from all vertices to a small set of cluster centers. These clusters are significantly more robust than those in~\cite{Parter22}: they tolerate not only $f$ adversarial faults but also up to $n - O\!\left(\frac{n}{f}\right)$ random faults, while the FT clusters of \cite{Parter22} guarantees that after any failure of at most $f$ vertices, a clustered vertex has at least \emph{one} non-faulty path to a center. Perhaps surprisingly, the additional robustness of the greedy clusters is the key to achieving improved size bounds. Overall, our analysis yields significantly sparser spanners with an almost linear dependence on the number of faults:

 \begin{theorem}\label{thm:sqrtk-greedy}
       Any $n$-vertex graph $G=(V,E)$ admits an $f$-VFT $(O(\sqrt{k}),O(k))$-spanner, for any  fixed constant $\delta >0$, with $\tilde{O}(f^{1+\delta} \cdot n^{1+1/k})$ edges.
\end{theorem}

The above theorem gives a spanner with only $f^{1+\delta}$ size overhead, for any fixed $\delta >0$, over the non-faulty construction of \cite{Ben-LevyP20}, substantially improving upon current size bounds and techniques. On the other hand, we highlight that as the above approach is based on a greedy algorithm which scans all $\sqrt{k}$-long paths in a graph, hence it currently has an inefficient running time of $\tilde{O}(fn^{\sqrt{k}})$, when analyzing the modified greedy algorithm of \cite{DinitzR20}, which is quasi-polynomial in $n$, as $k=O(\log(n))$. Improving the running time of our algorithm, and more generally FT-$(\alpha,\beta)$ spanners with $\alpha=o(k)$ remains an important open problem. It would also be very interesting to develop fault-tolerant analogues of the recent spanners of Chechik and Lifshitz~\cite{ChechikL26}, which achieve nearly optimal guarantees for all distances $d \leq k$, as well as of the $(1+\epsilon,\beta)$ spanners of Tan and Zhang~\cite{TanZ23}. We hope that the techniques developed in this work will help pave the way toward fault-tolerant constructions matching the state-of-the-art stretch guarantees known in the non-faulty setting.

We also study the setting of small locality, concretely, VFT $(k,k-1)$-spanners. We highlight that $(k,k-1)$-spanners are natural as they can be formed by promising stretch $k$ for vertices of distance $2$ in the graph, which is best possible, and were initially studied in the seminal work of Baswana et al. \cite{BaswanaKMP05}, and more recently with two new greedy constructions \cite{PopovaT26, Tzalik2026GreedyCompletion}. Our result in the $(k,k-1)$ is a VFT-spanner whose size is $f$-times larger than its non-faulty counterpart.

 \begin{theorem}\label{thm:hybrid-greedy}
       Any $n$-vertex graph $G=(V,E)$ admits an $f$-VFT $(k,k-1)$ spanner with $O(f\cdot k^3 \cdot n^{1+1/k})$ edges.
\end{theorem}

We highlight that this result nearly matches the $\Omega(f^{1-1/k}n^{1+1/k})$ lower bound for $f$-VFT $(2k-1,0)$-spanners by \cite{BodwinDPW18}.

%% file: technical-overview.tex
\section{Technical Overview}

\input{technical-BS}

\input{technical-overview-greedy}

\input{technical-overview-path-systems}

\input{technical-overview-sqrt-vft}

%% file: technical-BS.tex
To illustrate our key ideas and to give a taste of both of the techniques, it is instructive to focus on the construction of $f$-VFT $(\sqrt{k},O(k))$ spanners. In \Cref{sec:classical}, we first show the size bound obtained by the fault-tolerant sampling technique. In \Cref{sec:BS-VFT}, we present our Baswana-Sen based approach. We note that obtaining the general VFT $(O_{\varepsilon}(k^{\epsilon}),O_{\varepsilon}(k))$ spanner for $\epsilon<1/2$ calls for new ideas similarly to the \cite{Ben-LevyP20}, while the setting of $(O(\sqrt{k}),O(k))$  admits a simpler solution.  In Sec. \ref{sec:greedy-tech-overview} we present the high-level approach of the greedy VFT $(\alpha,\beta)$ spanners. While the greedy approach provides tighter size bounds, the benefit of the Baswana-Sen approach is that it can be extended to provide an $f$-VFT analogue for the general class of $(O_{\varepsilon}(k^{\epsilon}),O_{\epsilon}(k))$-spanners. As usual, one may assume $k=O(\log n)$ since the dominant dependence on $k$ in the size is captured by the term $n^{1+1/k}$.

\subsection{The Classical Approach to VFT $(\alpha,\beta)$ Spanners}\label{sec:classical}

The approach shown in \cite{BraunschvigCPS15} is based on applying the fault-free construction on a collection of $G$ subgraphs, later on formalized as \emph{Replacement Path Covering} \cite{KarthikP21}, and taking the union of these solutions as the output spanner.  For $s,t \in V$ and $F \subseteq V$, let $P(s,t,F)$ be some $s$-$t$ shortest path in $G$.
\begin{definition}[Replacement Path Covering (\RP)]\label{def:RP}
A subgraph $G' \subseteq G$ \emph{covers} a replacement path $P(s,t,F)$ if 
$P(s,t,F) \subseteq G' \mbox{~~and~~} F \cap V(G') =\emptyset.$ 

A collection of subgraphs of $G$, say $\mathcal{G}_{f,L}$, is an $(L,f)$-\emph{\RP} if for every $s,t \in V$ and every $F\subseteq V$ such that $|F|\leq f$, we have that each $P(s,t,F)$ replacement path with at most $L$ edges is covered by some subgraph $G'$ in $\mathcal{G}_{f,L}$. 
\end{definition}

$(L,f)$-\emph{\RP} can be obtained easily by using the fault-tolerant sampling technique which yields the following: 
\begin{claim}\cite{weimann2013replacement,DinitzK11,KarthikP21}\label{cl:RPC-size}
   There is a randomized algorithm for computing $\mathcal{G}_{f,L}$ with $$|\mathcal{G}_{f,L}|=\min\{(c f \log n)^{L+2}, (c (L+1) \log n)^{f+1}\},$$
   for some constant $c \geq 1$, such that the desired properties hold w.h.p. 
\end{claim}

\Cref{cl:RPC-size} can be used to provide $f$-VFT $(\alpha,\beta)$-spanners in a black box manner from a given construction of non-faulty $(\alpha,\beta)$-spanners, by exploiting the following locality of $(\alpha,\beta)$-spanners: 

\begin{observation}\label{obs:local-alpha-beta}
    Let $H$ be a subgraph satisfying that $\dist_{H - F}(u,v)\leq \alpha \cdot \dist_{G - F}(u,v)+\beta$, for some $\beta>0$ and every $u,v\in V$ and $F\in (V\setminus \{u,v\})^{\leq f}$ such that $\dist_{G - F}(u,v)\leq \lceil \beta/\alpha \rceil$, then $H$ is a $f$-VFT $(2\alpha,\beta)$-spanner.
\end{observation}
\begin{proof}
    Let $u,v,F$ be such that $\dist_{G - F}(u,v)=d$ for $d\geq \lceil \beta/\alpha \rceil$ and let $P$ be a $u$-$v$ shortest path in $G - F$. For simplicity denote $\gamma:=\lceil \beta/\alpha \rceil$. Writing $d=\gamma \cdot t+r$ for $0 \le r<\gamma$, we can partition $P$ into $t$ segments of length $\gamma$ and an extra segment of length $r$. By the properties of $H$, applied to each segment: 
    $\dist_{H - F}(u,v)\leq (\alpha \cdot\gamma+\beta) \cdot t + (\alpha \cdot r + \beta) = \alpha d+t \cdot \beta +\beta \leq 2\alpha \cdot d + \beta.$
    Since $\beta \cdot t \leq \alpha \cdot d$ by definition.
\end{proof}

Combining \Cref{obs:local-alpha-beta}, \Cref{cl:RPC-size} with $L=O(\beta/\alpha)$ and the fault-free construction of \cite{Ben-LevyP20}, yields $f$-VFT $(O(\sqrt{k}),O(k))$ spanner with $\tilde{O}(\min\{f^{O(\sqrt{k})},k^{O(f)}\}\cdot n^{1+1/k})$ edges. Our main challenge is to eliminate the exponential dependence in either $\sqrt{k}$ or $f$ in the size bound.

\subsection{A Baswana-Sen Type VFT $(\alpha,\beta)$ Spanner}\label{sec:BS-VFT}

We briefly recall the Baswana--Sen algorithm~\cite{BaswanaS07} for constructing (fault-free) $(2k-1,0)$-spanners, presented here in a slightly more general form tailored to our extensions. A \emph{clustering} $\mathcal{C}=\{C_1,\ldots,C_\ell\}$ is a collection of vertex-disjoint subsets of vertices, and we denote the set of clustered vertices by $V(\mathcal{C})=\bigcup_{C \in \mathcal{C}} C$. Each cluster $C_i$ has a designated \emph{center} $s_i$. We say that $\mathcal{C}$ has radius $r$ if, for every cluster $C_i$, the center $s_i$ spans all vertices of $C_i$ via a breadth-first search tree of depth at most $r$ rooted at $s_i$.
The algorithm proceeds iteratively, starting from the trivial clustering $\mathcal{C}_0=\{\{v\}\}_{v \in V(G)}$. For $k-1$ iterations, each step $i \ge 1$ takes as input an $(i-1)$-radius clustering $\mathcal{C}_{i-1}$ and produces an $i$-radius clustering $\mathcal{C}_i$, along with a set of edges $H_i$. The edges in $H_i$ ensure the stretch guarantee for vertices that become unclustered in this step, namely those in $V(\mathcal{C}_{i-1}) \setminus V(\mathcal{C}_i)$. A key property is that, in expectation, $|\mathcal{C}_i| = n^{1-i/k}$ for every $i \in \{0,\ldots,k\}$, and $\mathcal{C}_k$ is empty. The final spanner $H=\bigcup_i H_i$ guarantees a stretch of $2i-1$ for every edge $(u,v)$ provided that at least one of its endpoint is $i$-unclustered (i.e., not in $V(\mathcal{C}_i)$). Since $V(\mathcal{C}_k)=\emptyset$, all vertices are $k$-unclustered which yields the desired $2k-1$ stretch for all edges.

\paragraph{The Non-Faulty $(O(\sqrt{k}), O(k))$ Algorithm of \cite{Ben-LevyP20}}
The algorithm has two phases. The first performs $\sqrt{k}$ (out of the $k-1$ many) iterations of the Baswana--Sen algorithm. This yields a clustering $\mathcal{C}=\mathcal{C}_{\sqrt{k}}$ with $O(n^{1-1/\sqrt{k}})$ clusters in expectation, together with a subgraph $H=\bigcup_{i \le \sqrt{k}} H_i$ that guarantees a stretch of $(2\sqrt{k}-1)$ for every $G$-edge that is incident to an unclustered vertex $v \in V \setminus V(\mathcal{C})$.

The second phase constructs a \emph{virtual} graph $\widehat{G}$ on the $n^{1-1/\sqrt{k}}$ cluster centers of $\mathcal{C}$. In this graph, two centers $s,s'$ are connected by an edge if $\dist_G(s,s') \le 3\sqrt{k}$. The number of vertices in $\widehat{G}$ enables one to compute a $(2\sqrt{k}-1)$ multiplicative spanner $\widehat{H}$ with only a linear number of edges (in $n$). In the final step, each edge $(s,s') \in \widehat{H}$ is replaced by a shortest $s$--$s'$ path in $G$, adding $O(\sqrt{k}n)$ edges in total. Let $H'$ denote this resulting subgraph. The final spanner is $H \cup H'$, and the size bound follows immediately.

We now sketch the stretch argument. Consider a pair $u,v$ at distance $d \le \sqrt{k}$, and let $P$ be a shortest $u$--$v$ path. The nontrivial case is when $P$ contains at least two clustered vertices, since otherwise the guarantee follows from the first phase. Let $u'$ and $v'$ be the closest clustered vertices to $u$ and $v$ along $P$, respectively. By construction, the centers of $u'$ and $v'$ are adjacent in $\widehat{G}$, and hence are connected by a path of length $O(\sqrt{k})$ in $\widehat{H}$. This implies the existence of a path of length $O(k)$ between them in $H'$, completing the argument.

\paragraph{The Virtual Graph Challenge.}  The main challenge in extending the fault-free construction to the fault-tolerant setting arises in the second step, where a spanner is built over a virtual cluster graph. The first step, running $\sqrt{k}$ phases of the Baswana--Sen algorithm, can be replaced by its VFT analogue from~\cite{Parter22}, which already handles all unclustered vertices, in the sense that for every $G$-edge $(u,v)$ incident to an unclustered vertex it holds that $\dist_{H- F}(u,v)\leq 2\sqrt{k}-1$ for every $F \in (V \setminus \{u,v\})^{\leq f}$ where $H$ is the output spanner of running $\sqrt{k}$ VFT-BS steps. 

The key challenge is to handle the edges between clustered vertices. The fault-tolerant clusters of~\cite{Parter22} ensure that every clustered vertex remains connected to at least one of its centers after the failure of any set of at most $f$ vertices. It is therefore sufficient to provide a good stretch between center pairs, as in the non-faulty construction. 
The major barrier is that one can no longer rely on a virtual cluster graph in which each edge $(z,z')$ corresponds to \emph{some} short $z$-$z'$ path in $G$, as this path may be destroyed by faults. Consequently, the standard approach of replacing each virtual edge by a short path in $G$ breaks down. In the worst case, one must instead associate each such edge with a collection of up to $\min\{\sqrt{k}^f,\,f^{\sqrt{k}}\}$ paths to certify connectivity under $f$ faults, leading to a prohibitive blow-up\footnote{Moreover, the bounds $\min\{\sqrt{k}^{\Theta(f)},\,f^{\Theta(\sqrt{k})}\}$ are tight for encoding distances $\leq \sqrt{k}$ under faults in subgraphs. Concretely, for any $k,f$ there is a graph $G$ and two vertices $u,v$ such that any subgraph $H$ of $G$ satisfying $\dist_{G-F}(u,v)\leq \sqrt{k} \implies \dist_{H-F}(u,v)=\dist_{G-F}(u,v)$, has $|E(H)| \geq \min\{\sqrt{k}^{\Theta(f)},\,f^{\Theta(\sqrt{k})}\}$.}.

Our two approaches address this challenge in different ways. The Baswana--Sen based construction oversamples cluster centers, enabling the virtual graph step to be implemented across the collection of $(3\sqrt{k},f)$-\RP subgraphs that collectively ensure fault tolerance. This leads to a non-tight polynomial dependence on $f$. Our greedy based approach optimizes this dependence to be almost optimal. This improvement is made possible as the greedy algorithm implicitly constructs stronger FT-clusters than those of~\cite{Parter22}, where each cluster is a path collection to centers that is $\Omega(\sqrt{k}f)$-spread. Our analysis then shows that a \emph{single} greedy path plays the role of \emph{many} virtual edges in the $(3\sqrt{k},f)$-\RP subgraphs, allowing the same edges of $G$ to serve many center pairs, yielding the desired bound.

We note that while the BS-based approach provides a $\poly(f)$ size overhead for $f$-VFT $(O(k^{\epsilon}), O(k))$ spanners for any constant $\epsilon$, the greedy argument becomes much harder to achieve when $\epsilon< 1/2$.

\paragraph{$f$-VFT $(\sqrt{k},O(k))$-Spanner with $\poly(f)n^{1+1/k}$ Edges.}

The algorithm mimics the two phases of the non-faulty algorithm by \cite{Ben-LevyP20}. In the first phase, we apply $\sqrt{k}$ steps of the VFT variant of the Baswana-Sen algorithm by Parter \cite{Parter22} with one key modification: Instead of setting the sampling probability to $p=(f/n)^{1/k}$, we set it to $q=1/((cf\log n)^4 n^{1/k})$ where $c$ is the constant from \Cref{cl:RPC-size}. At the end of this phase, we have the following:
\begin{itemize}
    \item A collection of $\sqrt{k}$-radius clusters $\mathcal{C}$ centered at $S$ centers, where
    $$|\mathcal{C}|=|S|=n\cdot q^{\sqrt{k}}=O\left(\frac{n^{1-1/\sqrt{k}}}{(cf\log n)^{4\sqrt{k}}}\right)~.$$
    \item A subset $V^* \subseteq V(\mathcal{C})$ of $f$-VFT clustered vertices, in the following sense: For $u \in V^*$ and every $F \in (V \setminus \{u\})^{\leq f}$, $u$ has a fault-free path to one of the centers in $S$ in the trees of $\mathcal{C}$.
    \item A subgraph $H'$ of size $\tilde{O}(fn/q)$ that contains all the tree edges of the clusters in $\mathcal{C}$ and a subset of the edges incident to $V \setminus V^*$, such that: For every $G$-edge $(u,v)\in (V \setminus V^*)\times V$ and $F \in (V\setminus \{u,v\})^{\leq f}$, it holds that $\dist_{H - F}(u,v)\leq 2\sqrt{k}-1$.
\end{itemize}

In the second phase, the algorithm applies the second phase of the non-faulty algorithm of \cite{Ben-LevyP20} in each subgraph $G_j$ in a $(L=3\sqrt{k},f)$-\RP family $\mathcal{G}_{f,3\sqrt{k}}$. Specifically, for every $G_j \in \mathcal{G}_{f,3\sqrt{k}}$, compute a virtual graph $\hat{G}_j$ where:
$$\hat{G}_j=(S, \{(s,s')\in S \times S ~\mid~ \dist_{G_j}(s,s')\leq 3\sqrt{k}\})~.$$
Let $\hat{H}_j$ be a $\sqrt{k}$-spanner on $\hat{G}_j$, and let $H_j$ be a $G_j$-subgraph obtained by replacing each edge $(s,s')$ in $\hat{G}_j$ by the corresponding $s$-$s'$ shortest path in $G_j$, i.e.,
$H_j=\bigcup_{(s,s')\in \hat{H}_j}SP(s,s',G_j)~.$
The final spanner is given by $H=H' \cup \bigcup H_j$.

\noindent \textbf{Size:} By the $f$-VFT spanner of \cite{Parter22}, we have that: $|H'|=\tilde{O}(f^5n^{1+1/k})$. By the bound on $S$ and \Cref{cl:RPC-size}, we have: 
$$|\bigcup H_j|=|\mathcal{G}_{f,3\sqrt{k}}|\cdot 3\sqrt{k} \cdot |S|^{1+1/\sqrt{k}}=\tilde{O}(n)~.$$

\noindent \textbf{Stretch:} Fix $F \in V^{\leq f}$ and let $u,v \in V$ such that $\dist_{G - F}(u,v)\leq \sqrt{k}$. Let $P=SP(u,v, G - F)$. 
By the output of the first phase, it is sufficient to consider the case where $P$ contains at least two clustered vertices in $V^*$. Let $u',v'$ be the closest clustered vertex to $u$ and $v$ on $P$, respectively. Let $s_1$ be a center in $S$ such that the tree path in $T(s_1)$ from $u'$ to $s_1$ is non-faulty. Define $s_2$ w.r.t. $v'$. We then have that $\dist_{G - F}(s_1,s_2)\leq 3\sqrt{k}$, and thus by the properties of $\mathcal{G}_{f,3\sqrt{k}}$, there exists $G_j \in \mathcal{G}_{f,3\sqrt{k}}$ with $F \cap V(G_j)=\emptyset$, $\dist_{G_j}(s_1,s_2)\leq 3\sqrt{k}$ and consequently, $(s_1,s_2)\in \hat{G}_j$. We conclude that $\dist_{H_j}(s_1,s_2)=O(k)$ and therefore $\dist_{(H' \cup H_j)- F}(u,v)=O(k)$. The claim follows. 

\paragraph{VFT $(O_{\epsilon}(k^{\epsilon}),O_{\epsilon}(k))$ Spanners.} The case of $\epsilon = 1/2$ admits a relatively simple two-step solution. As observed in~\cite{Ben-LevyP20}, handling the regime $\epsilon < 1/2$ requires significantly more delicate constructions even in the non-faulty setting. Their algorithm employs a three-stage approach, introducing an intermediate \emph{superclustering} phase in which cluster centers are grouped into super-clusters. A key feature of this step is that it adds $\Omega(n)$ edges, even when the number of clusters is small (sub-linear in $n$). Concretely, during each superclustering step of ~\cite{Ben-LevyP20}, \emph{every} vertex connects to its nearest super-cluster via a short path. In the presence of faults, such paths may fail, and applying this procedure independently across the subgraphs of an $(L,f)$-replacement path collection leads again to an exponential blow-up in the number of edges.

To circumvent this issue, we design a new \emph{non-faulty} $(O_{\varepsilon}(k^{\epsilon}), O_{\epsilon}(k))$-spanner construction based on $(1/\epsilon)$ iterations of the Baswana--Sen algorithm applied to a virtual cluster graph. In this formulation, the number of edges added at each step depends only on the number of clusters, rather than the total number of vertices. While this provides slightly weaker bounds than \cite{Ben-LevyP20}, we find this formulation to be easier, and more naturally extendable to the faulty setting. In the fault-tolerant adaptation, each clustering step oversamples centers to ensure the following invariant: at the beginning of each step, we are given an $r$-radius clustering with 
\[
N = \frac{n^{1-1/x}}{(f \log n)^{3r}}
\]
clusters, for some integers $x$ and $r$. This choice of parameters allows us to construct a $(2x-1)$-spanner on the corresponding cluster graph within each subgraph of an $(L=3r, f)$-replacement path collection $\mathcal{G}$, while ensuring that the total number of edges across all subgraphs remains $N^{1+1/k}\cdot |\mathcal{G}|=\tilde{O}(n)$ (by \Cref{cl:RPC-size}). Specifically, for each subgraph $G_j \in \mathcal{G}$, we define a virtual graph on the cluster centers by connecting pairs at distance at most $3r$ in $G_j$. We then compute a $(2x-1)$-spanner $H'_j$ of this virtual graph, and obtain a subgraph $H_j$ of $G_j$ by replacing each virtual edge with a corresponding shortest path in $G_j$.

The choice $L = 3r$ is crucial: for the stretch analysis, we must guarantee $k^{\epsilon}$ stretch for clustered vertices that are sufficiently close in $G$, namely at distance $x$. This implies that their respective centers are within distance at most $3r$. This factor of $3$ turns out to be essential, and in fact drives the super-exponential dependence on $\epsilon$ in the final size bound.

%% file: technical-overview-greedy.tex
\subsection{Greedy VFT $(\alpha,\beta)$-Spanners}\label{sec:greedy-tech-overview}

Our constructions build on the greedy algorithm for fault-tolerant $d\to r$ spanners, introduced in~\cite{PopovaT26}. A $d\to r$ spanner of $G$ is a subgraph $H$ such that for every pair $x,y\in V(G)$ with $\dist_G(x,y)=d$, we have $\dist_H(x,y)\le r$. As was explained in \Cref{sec:BS-VFT}, $(\alpha,\beta)$-spanners for unweighted graphs can be obtained as a union of $O(\beta/\alpha)$ many $d\to r$ spanners. In particular, in our setting, a VFT $2\to 2k$ spanner of size $\tilde{O}(fn^{1+1/k})$ yields a $(k,k-1)$-spanner of the same size, while a VFT $\sqrt{k}\to O(k)$ spanner yields a $(O(\sqrt{k}),O(k))$-spanner with an additional $O(\sqrt{k})$ factor in the size.

We use the following greedy VFT $d\to r$ spanner algorithm on a graph $G$: iterate over all length-$d$ paths $P$ in $G$ with endpoints $x,y$, and add $P$ to the spanner $H$ whenever there exists a fault set $F$, disjoint from $P$, of size at most $f$ such that $\dist_{H-F}(x,y)>r$. As is typical for greedy algorithms, correctness is immediate from the update rule, and the main challenge is to bound the number of paths added.

Let $(P_1,F_1),(P_2,F_2),\ldots,(P_m,F_m)$ be the sequence of paths added by the VFT greedy algorithm, with $F_i$ being the set of faults causing $P_i$ to be added. For each $i\in[m]$, let $H_i$ denote the state of the spanner just before $P_i$ was added; thus $H_1=(V,\emptyset)$ and $H_i=\bigcup_{j=1}^{i-1}P_j$ for $i\ge 2$. Our goal is to bound $m$, which immediately yields an upper bound of $d\cdot m$ on the number of edges in the output spanner. For both the $2\to 2k$ and $\sqrt{k}\to O(k)$ constructions, our strategy is to partition the added paths into two classes,
\[
\{P_1,\ldots,P_m\}=\mathcal{L}\overset{\cdot}{\cup}\mathcal{G},
\]
and bound $|\mathcal{L}|$ and $|\mathcal{G}|$ separately. We denote these classes by $\mathcal{L}$ and $\mathcal{G}$ because paths in $\mathcal{L}$ increase the connectivity recorded at \emph{local} length scale at most $\lceil r/2d\rceil$, whereas bounding $\mathcal{G}$ requires understanding connectivity at a scale close to $r$, hence more \emph{global}.

To define the partition into $\mathcal{L}$ and $\mathcal{G}$, we maintain throughout the algorithm a collection of oriented paths $\mathcal{P}$ contained in the current spanner $H$. We use $\pp_v$ to denote the paths in $\pp$ that start at $v$, and $\pp_{v,t}$ to denote the paths in $\pp$ that start at $v$ and end at $t$. Let $\pp_v^{\ell}$ denote the paths in $\pp$ of length $\ell$ that start at $v$. Given a fault set $F$ and a path collection $\pp$, let $\pp-F = \{P \in \pp \mid V(P) \cap F = \emptyset\}$ be the non-faulty paths in $\pp$ under $F$. For a set $J$ and a path collection $\pp$, $\pp[J]=\{ P \in \pp \mid J\subseteq V(P) \}$, which will be used for analyzing path collections. The partition into $\mathcal{L}$ and $\mathcal{G}$ will depend on two parameters, the first is $D$ which will be approximately the average degree bound we would like to achieve in the output spanner. The second parameter is $s$, an integer smaller than $r/d$.

\paragraph{Deciding if a path is in $\mathcal L$ or $\mathcal G$}
To describe how $\pp$ evolves, we introduce the notion of \emph{critical} edges. We say that an edge $e=\{u,v\}\in P_i$ is $s$-critical if $\dist_{H_i-F_i}(u,v)>s$. By the stretch requirement, every path $P_i$ added by the greedy algorithm contains at least one $s$-critical edge, assuming $s\leq r/d$. The path collection $\pp$ is initially $\{(v)\}_{v\in V}$, and will consist of paths of length $\leq \lceil s/2\rceil$. We update $\pp$ according to the following rule, which also determines whether $P_i$ belongs to $\mathcal{L}$ or to $\mathcal{G}$. Informally, a path is placed in $\mathcal{G}$ if every $s$-critical edge already lies in a saturated part of the path system, and it is placed in $\mathcal{L}$ if some critical edge allows us to grow the path system at some non-saturated level. Setting $N_{\ell}=D^{\ell}$, we define $\pp,\mathcal{L},\mathcal{G}$ as follows:

\begin{itemize}
    \item If for every $s$-critical edge $e=(u,v)\in P_i$ and every $\ell \leq \lceil s/2\rceil$ we have $|\pp_v^\ell|\ge N_\ell$, then we place $P_i$ in $\mathcal{G}$ and leave $\pp$ unchanged.
    
    \item Otherwise, we show that there exist an $s$-critical edge $e=(u,v)\in P_i$ and an index $0\le \ell \leq \lceil s/2\rceil-1$ such that $|\pp_u^\ell|\ge N_\ell$ but $|\pp_v^{\ell+1}|<N_{\ell+1}$. In this case, we place $P_i$ in $\mathcal{L}$ and update
\[
\pp \gets \pp \cup (v,u)\circ (\pp_u^\ell-F_i).
\]
That is, we prepend the edge $(v,u)$ to the surviving length-$\ell$ paths rooted at $u$, thereby creating new length-$(\ell+1)$ paths rooted at $v$.
\end{itemize}

The paths in $\mathcal{L}$ increase the \emph{local} connectivity captured by $\pp$, since each such path forces $\pp$ to grow. We think of $\pp$ as tracking connectivity up to length $\lceil s/2\rceil \leq r/2d$, which may be much smaller than $r$, which is the scale relevant for the actual replacement paths in a $d\to r$ spanner. Thus, $\pp$ captures only local connectivity. In contrast, paths in $\mathcal{G}$ need not increase this local connectivity, and controlling them requires a different argument that captures the \emph{global} structure formed by adding them. Bounding $\mathcal{G}$ will require understanding how well things connect at length scales up to $r$.

This high-level approach, namely bounding $\mathcal{L}$ and $\mathcal{G}$ separately, applies to both the $\sqrt{k}\to O(k)$ and the $2\to 2k$ VFT spanners we construct. The argument for $\mathcal{L}$ is more general and applies for arbitrary $d,r$, whereas the argument for $\mathcal{G}$ is tailored to each construction and uses additional ideas.

In the technical overview, we focus only on the analysis of the $\sqrt{k}\to O(k)$ VFT spanners, since this is the more challenging regime, in which the stretch guarantee concerns longer paths. In this case, it is helpful to view the argument in light of the algorithm of Ben-Levy and Parter described in \Cref{sec:BS-VFT}: the paths in $\mathcal{L}$ are analogous to the first phase of their algorithm, while the paths in $\mathcal{G}$ play a role analogous to its second phase.

%% file: technical-overview-path-systems.tex
\subsubsection{Bounding $\mathcal L$}\label{sec:technical-over-path-systems}

It is not a priori clear why the path system $\pp$ should help bound either $|\mathcal L|$ or $|\mathcal G|$, since the update rule only tracks the sizes of $\pp_v^\ell$, but not its structural properties. Our main result is that this rule nevertheless forces $\pp$ to satisfy strong fault-tolerant properties. We show that for every $v,t$, the collection $\pp^\ell_{v,t}$ serves as a small certificate for length-$\ell$ connections from $v$ to $t$, while the collections $\pp_v^\ell$ are well-distributed, in the sense that no small set of vertices captures too many of their paths. These two features imply the desired bound on $\mathcal L$, namely $|\mathcal L|=\Tilde{O}(Dn)$, and will also play a central role in the analysis of $\mathcal G$.

\begin{lemma}[\Cref{lem:pp-invariants}, informal]\label{lem:local-invariants-intro}
    Let $\pp$ be the state of the path system before $P_i$ is considered. Assume $D>10s^2f$. Then the following properties hold for $\pp$:
    \begin{enumerate}
        \item All paths in $\pp$ are contained in $H_i$.
        \item (Bounded size) $|\pp_v^\ell| = O(N_\ell)$.
        \item (Weak collective FT) For any $u\neq v$ and every $\ell$, we have $|\pp_v^\ell[u]|\leq N_\ell/2f$.
        \item (Strong collective FT) For any $W \subseteq V\setminus\{v\}$ and every $\ell$, we have $|\pp_v^\ell[W]|= O(N_\ell/D^{|W|})$.
        \item (Pairwise a park) For every pair $v,t\in V$ and every $\ell$, the collection $\pp_{v,t}^\ell$ forms a park\footnote{See \Cref{def:park} for a full definition, though knowing it is not necessary for reading the technical overview.}.
    \end{enumerate}
\end{lemma}

The above lemma is the main technical input in our analysis of $\pp$ and $\mathcal{L}$. A main difficulty is that all these properties of $\pp$ \emph{depend on one another}, so they form a web of dependencies that must be proved \emph{simultaneously}.

Roughly speaking, the lemma shows that although $\pp$ is updated by examining only the sizes of $\pp_v$, it nevertheless acquires strong structural and fault-tolerant properties due to the choice of the update rule. We now briefly explain how these properties are used.

The first property says that the paths in $\pp$ reflect structural properties of the current spanner $H_i$. This will be crucial later in the global analysis, namely in bounding the number of paths in $\mathcal G$.

The second and third properties already imply the desired bound on $\mathcal L$. Consider a step in which a path $P_i\in\mathcal L$ is processed. By definition, there exist a critical edge $e=(u,v)$ and an index $0\leq \ell\leq \lceil s/2 \rceil-1$ such that the update rule adds the paths
\[
(v,u)\circ(\pp_u^\ell-F_i)
\]
to $\pp$. We charge the pair $(v,\ell+1)$ to this step. The key point is that each such charge increases $\pp_v^{\ell+1}$ by many new paths. Indeed, by the third property,
\[
|\pp_u^\ell-F_i|
\ge
|\pp_u^\ell|-\sum_{w\in F_i}|\pp_u^\ell[w]|
\ge
N_\ell-f\cdot \frac{N_\ell}{2f}
=
\Omega(N_\ell).
\]
Therefore, each charge to $(v,\ell+1)$ contributes $\Omega(N_\ell)$ paths to $\pp_v^{\ell+1}$. On the other hand, by the bounded-size property we always have $|\pp_v^{\ell+1}|=O(N_{\ell+1})$. Hence each pair $(v,\ell+1)$ can be charged at most $O(N_{\ell+1}/N_\ell)=O(D)$ times. Summing over all vertices $v$ and all $\ell\in \{0,1,\ldots,\lceil s/2 \rceil\}$ gives
\[
|\mathcal L| = O(nsD),
\]
which is the desired bound.

The remaining properties are needed for the later analysis of $\mathcal G$. The fourth property strengthens the weak collective FT property by showing that $\pp_v^\ell$ remains well-distributed at multiple scales. In particular, even if $|\pp_v^\ell[u]|$ is as large as $N_\ell/2f$, the subcollection $\pp_v^\ell[u]$ still satisfies a similar bound; for example,
\[
|\pp_v^\ell[\{u,w\}]| = O(N_\ell/D^2) = O(N_\ell/f^2).
\]
Thus, the fault tolerance of $\pp_v^\ell$ persists recursively across scales. This phenomenon is closely related to the notion of spreadness, which is a key ingredient in the proof of \Cref{thm:sqrtk-greedy}.

The fifth property is perhaps the most important. It says that when paths are added to $\pp$, they encode genuinely new local connectivity relations that were not previously captured by the path system. To formalize this, we use the framework of parks, introduced in~\cite{PPST25}, which quantifies the extent to which a path collection is resilient to $f$ faults.

We also include a short proof, assuming \Cref{lem:local-invariants-intro}, that the output of the multiplicative greedy VFT-$(2k-1)$ spanner has size $O(k^{3}f^{1-1/k}n^{1+1/k})$, which is only an $O(k^{3})$ factor larger than the bound of \cite{BodwinP19}, after plugging the appropriate values of $s,D$ into \Cref{lem:local-invariants-intro}. The main reasons we include the proof are: (1) to demonstrate the utility of \Cref{lem:local-invariants-intro} and how it can be used to bound greedy algorithms, and (2) after setting $D,s$ appropriately, one can use the ``pairwise a park'' condition to show that $\mathcal{G}=\emptyset$.

On the other hand in our greedy VFT spanners, it is impossible to choose $D,s$ such that (1) $|\mathcal{L}|=O(nDs)$ is smaller than the desired size bound of the spanner, and (2) $\mathcal{G}$ is empty. Hence we must find a new argument showing that $\mathcal{G}$ is not too large.

%% file: technical-overview-sqrt-vft.tex
\subsubsection{Bounding $\mathcal{G}$}\label{sec:sqrtk-technical-overview-greedy}

To obtain the global analysis, we still sample a set of centers as in \Cref{sec:BS-VFT},
but require a fundamentally different analysis. The key difference is that
replacement path covering treats each pair of centers separately, whereas we show
that every \emph{single} path in $\mathcal{G}$ \emph{simultaneously} improves the
connectivity of \emph{many} pairs of centers in a fault-resilient way.
To capture this effect, we introduce the Robust Color List Fault-Tolerant (CLFT)
process, which is a virtual process that can be coupled with how the connectivity between centers improve with each added path in $\mathcal{G}$. 
For simplicity, we assume
that $k$ is a perfect square and fix $0<\delta<1$.

Our starting point for the analysis is extracting from the path system $\pp$ a subcollection of paths to a small set of centers, with similar properties as $\pp$. We emphasize that the existence of such path collections \emph{to centers} relies on properties 4,5 of the path collection $\pp$ described in \Cref{lem:local-invariants-intro}. To describe the properties of the path collection, we recall that a hypergraph $\mathcal{H}=(U,E)$ is \emph{$q$-spread} if for every $J\subseteq U,|\mathcal{H}[J]| \ \le\ q^{-|J|}\,|\mathcal{H}|$, where $\mathcal{H}[J]=\{e \in E \mid J \subseteq e\}$.
We have:

\begin{proposition}[\Cref{prop:sqrt-to-ok-vft}, \Cref{obs:virtual-fault-proc} informal]\label{prop:sqrt-to-ok-vft-intro}
    Let $(P_1,F_1), \ldots, (P_m,F_m)$ be the paths added by the VFT greedy $\sqrt{k} \to 12k/\delta$ spanner with their corresponding faults. Then there exist $S \subseteq V$, $|S|=\tilde{O}\left(f^{- \sqrt{k}} n^{1 - \frac{1}{\sqrt{k}}}\right)$, such that:
    
    For any $\left(x_0,\ldots,x_{ \sqrt{k}} \right)=P_i  \in \mathcal G$, there exist $0\leq \ell<r \leq \sqrt{k}$ such that: 
        \begin{itemize}
            \item There is a path collection $\prepp_{x_{\ell}} \subseteq \pp_{x_\ell}$ of paths starting at $x_{\ell}$ and ending in $S$ of length $\sqrt{k}/\delta$ in $H_i$, such that the collection $\{V(P)\setminus\{x_{\ell}\} \mid P \in \prepp_{x_{\ell}} \}$ is $\Theta_\delta(\sqrt{k}f)$-spread. A similar path collection $\prepp_{x_r} \subseteq \pp_{x_r}$ exists for $x_r$.

            \item For any $s_{\ell}\in S$, such that some path in $\prepp_{x_{\ell}}-F_i$ ends in $s_{\ell}$ and any $s_{r}\in S$, such that some path in $\prepp_{x_{r}}-F_i$ ends in $s_{r}$,
            we have the following: 
            \begin{enumerate}
                \item(Big distance before adding $P_i$) $\dist_{H_i-F_i}(s_{\ell},s_r) \geq 8 \frac{k}{\delta}$.
                \item(Small distance after adding $P_i$)  $\dist_{H_{i+1}-F_i}(s_{\ell},s_r) \leq 3 \frac{\sqrt{k}}{\delta}$.
            \end{enumerate}
        \end{itemize}
\end{proposition}

The proposition states that along every path $P_i \in \mathcal{G}$ there are two vertices $x_\ell,x_r$ and path collections \emph{ending in centers}, $\prepp_{x_{\ell}},\prepp_{x_r}$, which should be thought of as sampled analogues of $\pp$ from \Cref{sec:technical-over-path-systems}. Before discussing the spread properties of $\prepp_{x_{\ell}}$ and $\prepp_{x_r}$, we turn to the second property, which says that the centers $s_{\ell},s_r$ experience an improvement in distance, with respect to the fault set $F_i$, by a factor exceeding $2\sqrt{k}$. This is strikingly similar to the condition for adding an edge in the multiplicative greedy spanner. Indeed, we soon define a virtual process that encodes how centers connect during the algorithm.

\paragraph{Spreadness, Robust Sunflowers and Massive Failures}

Spreadness has proved to be a powerful tool behind recent breakthroughs in mathematics and computer science \cite{LovettSunflowers, ParkPham}.  The celebrated spread lemma \cite{LovettSunflowers, ParkPham}, says that if $\mathcal{H}$ is $h$-uniform, and $W \sim Bin(U,p)$, for $p \approx log(h)/q$, then $\Pr[\exists \text{ hyperedge } e\in \mathcal{H} \text{ and }  e\subseteq W]= \Omega(1)$. A closely related concept is that of a robust sunflower:

\begin{definition}[Robust sunflower, \cite{Rossman14}]\label{def:robust-sun}
    Let $\mathcal H$ be a hypergraph over a universe $U$. Setting $K = \bigcap_{S \in \mathcal H} S$, we say that $\mathcal{H}$ is a $(p, \varepsilon)$-robust sunflower with kernel $K$ if 

    $$\Pr_{W \sim Bin(U, p)}[\exists S \in \mathcal{H},\ S \subseteq K \cup W]\geq 1 - \varepsilon.$$
\end{definition}

For example, the spread lemma implies that given an $h$-uniform and $q$-spread hypergraph $\mathcal{H}$, then $\{e \cup K \mid e\in \mathcal{H}\}$ is a $\left(O(\log(h)/q),\Omega(1) \right)$-robust sunflower with kernel $K$. The key combinatorial engine behind the global analysis is that the spread guarantees of \Cref{prop:sqrt-to-ok-vft-intro} imply that the set of newly added paths between centers which are formed by the greedy step are highly structured, and in fact form a robust sunflower:

\begin{lemma}\label{lem:robust-sunflower-intro}
    For $P_i \in \mathcal{G}$, let $x_{\ell},x_r$ be as in \Cref{prop:sqrt-to-ok-vft-intro}, and let $P'_i$ be the sub-path of $P_i$ from $x_{\ell}$ to $x_r$. Then the hypergraph $\left(V, (\prepp_{x_\ell}-F_i)^{rev} \circ P'_i \circ (\prepp_{x_r}-F_i)\right)$ is an $(1/2f,1/2)$-robust sunflower with kernel $V(P'_i)$.
\end{lemma}

Intuitively, the lemma above says that at each step the pairs of centers considered by \Cref{prop:sqrt-to-ok-vft-intro} form a highly fault tolerant structure. Indeed, being a $(1/2f,1/2)$-robust sunflower means that, assuming $V(P'_i)$ is non-faulty, even if we sample a random and \emph{massive} set of failures $F \sim Bin(V\setminus V(P'_i),1-1/2f)$, then with constant probability some pair of centers considered in \Cref{prop:sqrt-to-ok-vft-intro} still connect to each other. We refer to this as a massive failure regime, since each vertex fails independently with probability $1 - \frac{1}{2f}$, which is very close to $1$, yet collectively some path between centers still manages to survive with constant probability. 

In other words, we show our notion of FT-clusters is robust to $f$ adversarial failures, as well as massive random failures. In contrast, Parter's FT clustering, used in \Cref{sec:BS-VFT}, can handle $f$ adversarial faults, but cannot handle massive random failures, which helps explain the weaker bound obtained in \Cref{sec:BS-VFT}. We find the use of structures robust to massive random failures as a way to analyze \emph{adversarial} faults a main novelty of the paper.

\subsubsection{The Robust CLFT Process}
Our final step is to devise a virtual process that encodes the distance--fault relations between centers and how they evolve throughout the greedy algorithm. We let $\mathbf{V}_1=(S,\emptyset)$ be a virtual graph. For a path $P_i \in \mathcal{G}$, consider a \emph{multigraph} $B_i$, whose vertex set is $S$, and for every path $Q_{\ell} \in \prepp_{x_\ell}-F_i$ ending in $s_{\ell}$, and $Q_{r} \in \prepp_{x_r}-F_i$ ending in $s_{r}$, there is a virtual edge $e$ whose endpoints are $s_\ell$ and $s_r$\footnote{It may be the case that there exist $Q_{\ell}'\neq Q_\ell$ and $Q_{r}' \neq Q_r$ ending in the same centers as $Q_\ell$ and $Q_r$, respectively. In such a case, $Q_{\ell}'$ and $Q_{r}'$ give rise to an edge $e'$ which is \emph{different} from $e$, but has \emph{the same} endpoints as $e$, so $B_i$ may contain parallel edges.}. Every such edge $e$ is assigned the list $L(e) = V(Q_{\ell}) \cup V(P'_i) \cup V(Q_{r})$ of vertices lying on the path that $e$ represents in $H_{i+1}$. We then set $\mathbf{V}_{i+1}=\mathbf{V}_i \cup B_i$ whenever $P_i \in \mathcal{G}$; otherwise $\mathbf{V}$ is unchanged.

\Cref{prop:sqrt-to-ok-vft-intro}, and \Cref{lem:robust-sunflower-intro} translate to the following properties:

\begin{enumerate}\label{enum:clft-intro}
    \item Every $e \in E(B_i)$ represents a path of length $\leq \frac{3\sqrt{k}}{\delta}$, and gets a list $L(e)$ of all \emph{vertices} on the path $e$ represents. $L(e) \cap F_i = \emptyset$ by construction.

    \item For every $(s,s') = e \in E(B_i)$, $\dist_{H_i-F_i}(s,s') > 2\sqrt{k} \cdot  (3 \frac{\sqrt{k}}{\delta})$. Passing to $\mathbf{V}$, this implies $\dist_{\mathbf{V}_i-F_i}(s,s') > 2\sqrt{k}$.

    \item The collection $\{L(e)\}_{e\in E(B_i)}$ is a $(1/2f,1/2)$-robust sunflower with kernel $V(P'_i)$. In particular their kernel has size $\leq \sqrt{k}+1$.
\end{enumerate}

To conclude the global analysis, we need to analyze the above virtual process, which we do in full generality. We will work in the color list fault model \cite{PetruschkaST24_ITCS}. Let $\mathcal{C}$ be an abstract set of ``colors'' that are possibly failing, and every edge has a list $L(e)$ of colors that must remain operational for $e$ to function. Concretely, the fault model is that some $F \subseteq \mathcal{C}$ of size $\leq f$ colors may fail, causing an edge $e$ to fail whenever $L(e) \cap F \neq \emptyset$. In our context, one should think of $\mathcal{C}$ as the vertex set of our big graph $G$, and an edge $e=(s,s')$ represents a relatively short path between $s,s'$, whose vertex set is $L(e)$. 

We consider a process, depending on two integers $t,w$. Initially $\mathbf{V}$ is $(S,\emptyset)$ where $S$ is a set of vertices.  In every round,  we are given a list colored graph $B_i=(S,E_i)$, a fault set $F_i$, and each edge $e \in B_i$ has a list $L(e) \subseteq \mathcal{C}$ such that:

\begin{enumerate}
    \item For every $e \in E(B_i)$ we have $L(e) \cap F_i = \emptyset$.

    \item For every $(s,s')=e\in E(B_i)$, $\dist_{\mathbf{V}_i-F_i}(s,s') >t$, where $\mathbf{V}_i = \bigcup_{j<i} B_j$.

    \item The collection $\{L(e)\}_{e\in E(B_i)}$ is a $(1/2f,1/2)$-robust sunflower with kernel of size $\leq w$.
\end{enumerate}

Observe that all the properties are the exact counterparts of the properties of the new-center connections formed by an addition of a path in $\mathcal{G}$. Bounds on the number of rounds of this process, translate to bounds of $\mathcal{G}$, which is how we will conclude the proof.
We call the above process the robust CLFT (color list fault tolerant) process, as it is a robust description of the edges that are added by multiplicative greedy CLFT spanner (with stretch $t$) of \cite{PetruschkaST24_ITCS}, but instead of adding one edge with a list of size $\leq w$ each time, we add many edges that are \emph{much more} prone to faults (as their list is of possibly unbounded size). See \Cref{fig:to_clft_intro} for an illustration of how $B_i$ is built from the path collections guaranteed in \Cref{prop:sqrt-to-ok-vft-intro}.

\begin{figure}[H]
    \centering
    \includegraphics[width=0.9\linewidth]{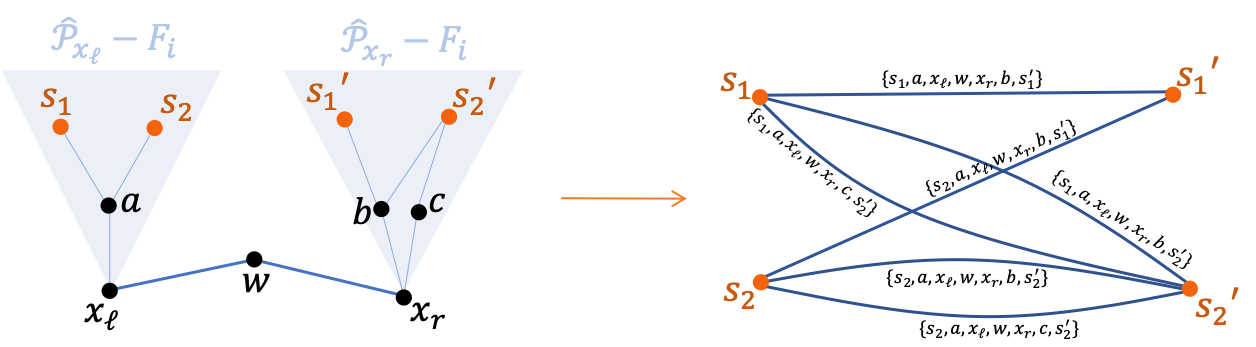}
    \caption{For path collections $\prepp_{x_\ell}$ and $\prepp_{x_r}$ of \Cref{prop:sqrt-to-ok-vft-intro}, we consider the non-faulty paths in each of them. Edges of  $B_i$ are in bijection to $\left(\prepp_{x_{\ell}}-F_i\right) \times \left(\prepp_{x_{r}}-F_i\right)$ and every edge in $B_i$ is assigned with the set of vertices along the path it represents. In the figure, assuming $\prepp_{x_{\ell}}-F_i=\{(x_\ell,a,s_1),(x_\ell,a,s_2)\}$ and  $\prepp_{x_{r}}-F_i=\{(x_r,b,s'_1),(x_r,b,s_2'), (x_r,c,s_2')\}$, we construct $B_i$ as appears on the right. }
    \label{fig:to_clft_intro}
\end{figure}

This introduces two contradicting forces: (1) $E(B_i)$ may be very big, hence in every step many edges are added in parallel to $\mathbf{V}$, and (2) the list of every edge $e$ is of possibly unbounded size, making it extremely unreliable. We show that when $\{L(e)\}_{e\in E(B_i)}$ is a robust sunflower, these two forces balance each other out. Concretely we prove the following bound, which is sharp up to a multiplicative factor $w$ and generalizes the work of \cite{PetruschkaST24_ITCS}: 

\begin{theorem}\label{thm:robust-clft-process-intro}
    The number of rounds in the robust CLFT process is at most:
    \[
    O\!\left(w f^{w}  |S|^{1+1/\lceil t/2 \rceil }\right).
    \]
\end{theorem}

The proof of \Cref{thm:robust-clft-process-intro} shows that size bounds derived from the blocking-set method of \cite{BodwinP19} carry over to the robust setting, where a single edge is replaced by a collection of edges forming a robust sunflower, with only negligible loss in the parameters. We believe this may be of independent interest.

\paragraph{Concluding the proof}

To bound the number of paths in $\mathcal{G}$, we apply the bounds from the robust CLFT process, but w.r.t. the actual $w,t$ from \Cref{enum:clft-intro}. In that case we have $w\leq \sqrt{k}+1$, and $t=2\sqrt{k}$ and the bound $|S|=\tilde{O}\!\left(f^{- \sqrt{k}} n^{1 - \frac{1}{\sqrt{k}}}\right)$ from \Cref{prop:sqrt-to-ok-vft-intro}. Applying \Cref{thm:robust-clft-process-intro}, we conclude:

\[|\mathcal{G}|=O\left((\sqrt{k}+1)\cdot f^{\sqrt{k}+1} \cdot |S|^{1+1/\sqrt{k}}\right)= \tilde{O}\left(f^{\sqrt{k}+1} \cdot (f^{- \sqrt{k}} n^{1 - \frac{1}{\sqrt{k}}})^{1+\frac{1}{\sqrt{k}}}\right)= \tilde{O}(n).\]

\paragraph{Comparison to Prior Work of \cite{PPST25} and \cite{PopovaT26}.}
The two works most closely related to ours are \cite{PPST25} and \cite{PopovaT26}. In \cite{PPST25}, Parter et al.\ developed a Baswana--Sen-type framework for fault-tolerant clustering, designed for color fault-tolerant spanners. Their algorithm maintains path collections satisfying guarantees similar to those of $\prepp$ in \Cref{prop:sqrt-to-ok-vft-intro}. There are, however, several important differences.

First, in \cite{PPST25} the collection $\prepp$ is constructed explicitly, as part of a rather intricate keep/safe/postpone paradigm. In contrast, in our setting $\prepp$ emerges implicitly along the execution of the greedy algorithm via a simple update rule.
Second, the bounds obtained by their method incur an additional $\exp(k^2)$ factor in the spanner size, whereas our approach incurs only a $\poly(k)$ overhead. As a result, our bounds apply throughout the full range of $k$ and yield stronger size guarantees. Third, \cite{PPST25} uses the notion of parks, closely related to spreadness (see \Cref{sec:parks}), as a structure for quantifying connectivity under varying fault sets. In that work, spreadness is used to handle collections of paths under up to $f$ adversarial faults. In our setting, spreadness plays a new central role: it induces sunflower-type connections between centers, which we exploit in the size analysis rather than in the stretch analysis.

The second closely related work is the recent paper of \cite{PopovaT26}, which studies greedy $d\to r$ spanners and their fault-tolerant variants. The differences are both in the results and in the techniques. At the level of results, \cite{PopovaT26} derives bounds for fault-tolerant spanners through a reduction to the non-faulty greedy setting. Because this approach is fundamentally black-box, it does not overcome the exponential-dependence barrier. For instance, for EFT $d\to r$ spanners, their framework yields an $f^d$ overhead relative to the size of the non-faulty counterpart (\cite{PopovaT26}, Theorem 45), which is essentially the same dependence as in the folklore bound of \Cref{cl:RPC-size}. In contrast, our results go beyond this fault-tolerant sampling barrier.

At the level of techniques, their fault-tolerant analysis is based on reducing the execution of the FT-greedy algorithm to that of the standard greedy algorithm. Our approach is substantially different. It develops new tools tailored to the fault-tolerant setting: evolving robust path systems through greedy update rules, modeling the interactions between centers via the robust CLFT process, and providing a tight analysis of this process.

\subsection{Organization of the Paper}

In \Cref{sec:BS}, we construct our VFT $(O_{\varepsilon}(k^{\varepsilon}),O_{\varepsilon}(k))$-spanners. The section begins with \Cref{sec:BS-non-faulty}, where we present a new \emph{non-faulty} $(O_{\varepsilon}(k^{\varepsilon}),O_{\varepsilon}(k))$-spanner construction that serves as the backbone of the VFT construction developed later in \Cref{sec:BS-VFT-sec}.

In \Cref{sec:greedyFTclustering}, we formally introduce the $\mathcal{L},\mathcal{G}$ decomposition of the paths added by the greedy algorithm, define the path system $\pp$, and establish its main properties. The proof is somewhat technical and may be skipped on a first reading. In \Cref{sec:sqrtk-greedy}, we complete the size analysis of the paths in $\mathcal{G}$ and obtain an $(O_{\delta}(\sqrt{k}),O_{\delta}(k))$-spanner of size $\tilde{O}(f^{1+\delta}n^{1+1/k})$. This section may be read independently of \Cref{sec:greedyFTclustering}, assuming \Cref{prop:vft-d-to-r-structure}, whose proof is given there. Finally, in \Cref{sec:vft2to2k}, we construct VFT $(k,k-1)$-spanners via a different analysis of the paths in $\mathcal G$ added by the greedy algorithm. This section relies heavily on \Cref{sec:greedyFTclustering} and is independent of \Cref{sec:sqrtk-greedy}.

%% file: VFT-exp.tex
\section{VFT $(k^{\epsilon},O_{\epsilon}(k))$-Spanners}\label{sec:BS}

Our VFT $(\alpha,\beta)$ spanner algorithms are based on a new and arguably simpler algorithm for computing $(k^{\epsilon},O(k))$ spanners, in the non-faulty setting. The desired VFT spanners will be then obtained by combining this algorithm with the VFT multiplicative spanners of \cite{Parter22} and the hit and miss graph families of \cite{KarthikP21}. 

\subsection{A New Algorithm for $(O_\epsilon(k^{\epsilon}),O_{\epsilon}(k))$-Spanners (Without Faults)}\label{sec:BS-non-faulty}
We start by presenting an algorithm for $(O_{\varepsilon}(k^{\epsilon}),O_{\epsilon}(k))$-spanners which is arguably simpler than the algorithm of \cite{Ben-LevyP20} at the cost of obtaining sub-optimal bounds in terms of the dependency on $\epsilon$. The benefit of this algorithm is that it can be combined with the hit and miss families to provide VFT spanners.

\textbf{Baswana-Sen Step.} The input is given by an $n$-vertex graph $G$, a collection of $r$-radius clusters $\mathcal{C}$ with centers $S=\{s_1,\ldots, s_\ell\} \subseteq V$ and a sampling probability $p$. (In the Baswana-Sen algorithm $p=n^{-1/k}$ where $k$ is the stretch parameter, for our purposes $p$ is given as input to the algorithm.) The clusters in $\mathcal{C}$ are vertex disjoint, and denote $C(s_j)$ to be the cluster of the center $s_j \in S$. Each cluster $C(s_j)$ induces a BFS tree of depth $r$ rooted at $s_j$ and spanning all vertices in $C(s_j)$. 
The output of the step is an $(r+1)$-radius clustering $\mathcal{C}'$ with $O(p \cdot \ell)$ clusters $S' \subseteq S$, w.h.p., and a subset of edges $H'$, defined as follows:
\begin{itemize}
\item $H'$ initially includes all the $r$-radius spanning trees of each cluster $C(s_j) \in \mathcal{C}$.

\item Let $S' \subset S$ be a random sample of centers, obtained by adding each $s_j \in S$ into $S'$ independently with probability $p$. Initialize the clusters $\mathcal{C}'=\{C'(s_j)=C(s_j) ~\mid~ s_j \in S'\}$.
\end{itemize}
In the following two steps process only vertices whose original cluster was not sampled.
\begin{itemize}
    \item For every vertex $v\in V(\mathcal{C})$ that is incident to some sampled cluster $C(s_j)\in \mathcal{C}$ for $s_j \in S$,
        pick some $(v,u) \in E$ such that $u\in C(s_j)$ and add $(v,u)$ to $H$ and $v$ to $C'(s_j)$. Each such vertex joins only one arbitrary incident cluster.
    
    \item For every vertex $v \in V(\mathcal{C})$ with no incident sampled clusters add one edge $(v,u_j)$ to each incident cluster $C_j$ in $\mathcal{C}$.
   \end{itemize}
\begin{claim}\label{cl:BSstep}
    (i) $|\mathcal{C}'|=|S'|=O(|S|\cdot p)$, in expectation. (ii) each $s_j \in S'$ is a center of a cluster in $\mathcal{C}'$ spanning its cluster members by a $(r+1)$-radius tree. (iii) for every vertex $u \in V(\mathcal{C})\setminus V(\mathcal{C}')$ and every $G$-edge $(u,v)$ with $v \in V(\mathcal{C})$, it holds that $\dist_{H'}(u,v)\leq 2r+1$. (iv) W.h.p., $|H'|=\tilde{O}(n/p)$.
\end{claim}
\begin{proof}
   (i) and (ii) are immediate by the construction. To see (iii), consider such an edge $(u,v)$. Let $C_v$ be the cluster of $v$ in $\mathcal{C}$. Since $u \in V(\mathcal{C})\setminus V(\mathcal{C}')$, we have that $u$ is incident to the cluster $C_v$. Hence, the algorithm adds to $H$ an edge from $u$ to some vertex $w \in C_v$. Since $H'$ contains an $r$-radius tree spanning all vertices in $C_v$, it holds that $\dist_{H'}(u,v)\leq 1+\dist_{H'}(w,v)\leq 2r+1$. We now bound the size of $H'$. Since the clusters in $\mathcal{C}$ are vertex disjoint, adding their BFS trees adds at most $n-1$ edges to $H'$. We next bound the edge added due to unclustered vertices, namely, in $V(\mathcal{C})\setminus V(\mathcal{C}')$.
   Observe that each vertex in $V(\mathcal{C})$ that is incident to at least $\Omega(\log n/p)$ distinct clusters in $\mathcal{C}$ is w.h.p. in $V(\mathcal{C'})$. Therefore, w.h.p., a vertex $u$ in $V(\mathcal{C})\setminus V(\mathcal{C}')$ is incident to $O(\log n/p)$ clusters in $\mathcal{C}$, and since the algorithm adds at most one $u$-edge to each incident cluster, the total number of edges added is $O(n\log n/p)$. 
\end{proof}

By repetitive applications of \Cref{cl:BSstep}, we have:
\begin{lemma}\label{lem:prop-NFT-BSsteps}
    Given an $N$-vertex graph $G$ with initial clustering $\mathcal{C}_0=\{\{v\} ~\mid~ v \in V(G)\}$ and probability of $p$. Having $x$ Baswana-Sen steps yields a clustering $\mathcal{C}'$ and a subgraph $H'$ with the following properties:
    \begin{enumerate}
    \item The expected number of clusters is $\mathbb E[|\mathcal{C}'|]= N \cdot p^{x}$.
    \item The radius of the clusters is at most $x$.
    \item $|H'|=\tilde{O}(x \cdot N/p)$, w.h.p., and also $\mathbb{E}[|H'|] = O(x \cdot N/p)$.
     \item For every $(u,v)\in E$ such that $u \notin V(\mathcal{C}')$, $\dist_{H'}(u,v)\leq 2x-1$. 
    \end{enumerate}
\end{lemma}

The second part of item 3 holds because of the following. On each step, conditional on the input clustering, a vertex incident to $t$ clusters adds at most one edge when it sees a sampled cluster, and at most $t$ edges otherwise. Hence its expected contribution is at most $1 + t(1-p)^t = O(1/p)$. Summing over at most $N$ vertices and over $x$ steps gives $O(x \cdot N/p)$.

\paragraph{The $(O(k^{\epsilon}),O_{\epsilon}(k))$-Spanner Algorithm.} 

The algorithm has $\ell=\lceil 1/\epsilon \rceil$ phases. Set $r_1=1$ and for $i  \in \{2,\ldots, \ell\}$, let:
$$r_i=3\lceil k^{\epsilon} \rceil\cdot \sum_{j=1}^{i-1}r_j~.$$

The input to phase $i\geq 1$ is given by an $N_{i}$-vertex graph $\hat{G}_{i}=(V_{i}, E_{i})$ where each $\hat{G}_{i}$-edge $(u,v)$ corresponds to a $u$-$v$-shortest path in $G$ of length at most $r_{i}$. Initially, $\hat{G}_1=G$, $V_1=V$ and $r_1=1$, $N_1=|V_1|=n$. Let $\mathcal{C}_0=\{ \{v\} ~\mid~ v \in V\}$.

Define $$\bar N_1 = n, \,\,\, \bar N_i = n^{1 - 1/k^{1 - (i-1)\epsilon}}\,\, (2\leq i \leq \ell), \,\,\, p_i = \bar N_i/n^{1 + 1/k}.$$

\smallskip
\noindent \textbf{Phase $i \in \{1,\ldots, \ell-1\}$:} The algorithm applies $\lceil k^{\epsilon}\rceil$ Baswana-Sen steps on the $N_i$-vertex graph $\hat{G}_i$ with sampling probability of $p_i$ and the trivial initial (singleton) clustering $\mathcal{C}_{i,0}=\{\{v\}, v \in V(\hat{G}_i)\}$. Let $\hat{H}_i$ and $\mathcal{C}_i$ be the output $\hat{G}_i$-subgraph and clustering, respectively, at the end of these $\lceil k^{\epsilon}\rceil$ steps. Denote the cluster centers of $\mathcal{C}_i$ by $S_i\subseteq V_i$.

The output of phase $i$ is a $G$-subgraph $H_i$ and a virtual graph $\hat{G}_{i+1}$ that will be given as input to phase $i+1$. These graphs are defined as follows: Let $H_i$ be a $G$-subgraph obtained by taking a $u$-$v$ shortest path in $G$, denoted by $SP(u,v,G)$ for every $(u,v)\in \hat{H}_i$:
$$H_i=\bigcup_{(u,v)\in \hat{H}_i}SP(u,v,G)~.$$

Define a virtual graph $\hat{G}_{i+1}$ with vertex set $V_{i+1}=S_i$ obtained by connecting nearby vertex pairs in $V_{i+1}$ at distance at most $r_{i+1}$. Formally,  
$\hat{G}_{i+1}=(V_{i+1}=S_i, \hat{E}_{i+1})$, where 
$$\hat{E}_{i+1}=\{(u,v) \in V_{i+1}\times V_{i+1} ~\mid~ \dist_G(u,v)\leq r_{i+1}\}.$$ 
This concludes the description of Phase $i$. 

\noindent\textbf{Phase $i = \ell$:} The algorithm does exactly the same as in previous phases, except that it applies only $\lceil k^{\epsilon'}\rceil$ Baswana-Sen steps on $\hat{G}_{\ell}$, where $\epsilon' = \epsilon - (\epsilon\lceil \frac{1}{\epsilon}\rceil - 1)$. (This modification is needed to handle the case where $1/\epsilon$ is not integral.)

\noindent The output spanner is given by $H=\bigcup_{i=1}^\ell H_i$.

\smallskip
\noindent \textbf{Size Analysis.} By \Cref{lem:prop-NFT-BSsteps}(3), taking expectation conditional on the input to the phase $i$, we have that $\mathbb E[|\hat{H}_i|]=O(\lceil k^\epsilon \rceil \mathbb E [N_i]/p_i)$ for every $i$. We shall see below in \Cref{cl:properties} that $\mathbb E[N_i] \leq \bar N_i$, and thus $\mathbb E [N_i]/p_i \leq n^{1 + 1/k}$. This gives $\mathbb E[|\hat{H}_i|]=O(\lceil k^\epsilon \rceil n^{1+1/k})$. Since each edge of $\hat{H}_i$ corresponds to a $G$-path of length at most $r_i$ and as $i \leq \ell$, we have that 
$\mathbb E[|H_i|]=O(r_i \mathbb E[|\hat{H}_i|])= O(r_i \lceil k^\epsilon \rceil n^{1+1/k})$. Thus, for a suitable absolute constant $C > 0$, $$\mathbb E[|H|] \leq B := C \lceil k^\epsilon \rceil n^{1+1/k}\sum_{i = 1}^\ell r_i.$$
Note that $B = \tilde{O}(3^{1/\epsilon} \cdot n^{1 + 1/k})$, as desired. 

\smallskip
\noindent\textbf{Stretch Analysis.} We start by proving the following:

\begin{claim}\label{cl:properties}
   $N_1=n$ and for every $i \in \{2,\ldots, \ell\}$, it holds that $\mathbb E[N_i] \leq \bar{N_i}$. Hence, $\mathbb E[N_{\ell}]\leq n^{1-1/k^{\epsilon}}$. Also, $N_{\ell+1} = 0$ w.h.p. 
\end{claim}
\begin{proof}
    The claims are shown by induction on $i$. For $i=1$, $N_1=n$. Assume that it holds for $i$ and consider $i+1$. By \Cref{lem:prop-NFT-BSsteps}(1) and combining with the induction assumption, for $i\in \{1,..,\ell-1\}$ we have:
    $$\mathbb{E}[N_{i+1}] = p_i^{\lceil k^\epsilon\rceil}\mathbb{E}[N_i] \leq p_i^{k^\epsilon}\bar N_i \leq \bar N_{i+1}~.$$
    Similarly,
    $$\mathbb E[N_{\ell+1}]\leq \mathbb E[N_\ell] \cdot p_\ell^{\lceil k^{\epsilon'}\rceil}\leq \bar N_\ell \cdot p_\ell^{k^{\epsilon'}}=\frac{\bar N_\ell^{k^{\epsilon'}+1}}{n^{(1+1/k)k^{\epsilon'}}} = o(1)~,$$
    which implies $N_{\ell+1} = 0$ w.h.p. by Markov inequality.
\end{proof}

By Markov, $\Pr[|H| \leq 4B] \geq 3/4$. Together with \Cref{cl:properties}, this gives an execution with $|H| \leq 4B$ and $N_{\ell + 1} = 0$. Fix such an execution for the stretch analysis.
\smallskip

\noindent \textbf{Clustering.} Let $\mathcal{C}_0=\{\{v\} ~\mid~ v \in V\}$ be the collection of $n$ singleton clusters, and recall that $\mathcal{C}_i$ is the output clustering of applying $\lceil k^{\epsilon}\rceil$ BS steps on the graph $\hat{G}_i$ in phase $i$ (or $\lceil k^{\epsilon'}\rceil$ BS steps if $i = \ell$). We have $V(\mathcal{C}_i)\subseteq V_i$. We first define the notion of $i$-clustered vertices for the vertices in $V_i$: a vertex $u \in V_i$ is $i$-clustered if $u$ belongs to the clustering $\mathcal{C}_i$ (i.e., $u \in V(\mathcal{C}_i)$) and let $c_i(u) \in V_{i+1}$ be the cluster center of $u$ in $\mathcal{C}_i$. We next extend this definition for any vertex $v \in V$ in an inductive manner. For $i=0$, $c_0(v)=v$ and for $i \geq 1$, $v$ is $i$-clustered if it is $(i-1)$-clustered and its $(i-1)^{th}$ center $c_{i-1}(v)$ is $i$-clustered. (Since $c_{i-1}(v)$ is a cluster center in $\mathcal{C}_{i-1}$, $c_{i-1}(v)$ is in $V_i$ and hence this is well-defined.) Otherwise, $v$ is $i$-unclustered.  For an $i$-clustered vertex, let $c_i(v)=c_i(c_{i-1}(v))$. 

By \cref{cl:properties}, we have:
\begin{observation}\label{obs:final}
    Any vertex $u$ is $0$-clustered and it is $i$-unclustered for some $i \in \{1,\ldots,\ell\}$.
\end{observation}

\begin{observation}\label{obs:rad-NFT-step}
    Let $i \in \{1,\ldots, \ell-1\}$. For an $i$-clustered vertex $v \in V$ it holds that $\dist_{H}(v,c_i(v))\leq r_{i+1}/3$. For $i=0$, $\dist_{H}(v,c_0(v))=0$.
\end{observation}
\begin{proof}
The proof is shown by induction on $i$. For $i=0$, we have that $c_0(v)=v$ and the claim holds vacuously. Assume that it holds for $i-1$ and consider an $i$-clustered vertex $v$.
Let $z=c_{i-1}(v)$, by definition we have that $c_i(z)=c_i(v)$. By \Cref{lem:prop-NFT-BSsteps}(2) and plugging $x=\lceil k^{\epsilon} \rceil$, we have $\dist_{\hat{H}_i}(z,c_i(v))\leq \lceil k^{\epsilon}\rceil$. As each edge in $\hat{H}_i$ corresponds to a path of length at most $r_i$ in $G$, we have $\dist_{H_i}(z,c_i(z))\leq \lceil k^{\epsilon}\rceil \cdot r_i$. By combining with the induction assumption for $i-1$:
\begin{eqnarray*}
    \dist_{H}(v,c_{i}(v))&\leq &\dist_{H}(v,z)+\dist_{H}(z,c_{i}(v))\leq r_{i}/3+ \lceil k^{\epsilon}\rceil \cdot r_{i}=\lceil k^{\epsilon}\rceil \cdot \sum_{j=1}^{i} r_j=r_{i+1}/3~.
\end{eqnarray*}
   
\end{proof}

Let $r'_1=1$ and $r'_i=r_i/3$ for every $i \in \{2,\ldots, \ell\}$.

\begin{claim}\label{cl:key-stretch-NFT-basic}
Let $v$ be an $(i-1)$-clustered and $i$-unclustered vertex. For any $(i-1)$-clustered vertex $u \in B_G(v,r'_i)$, it holds that
    $\dist_{H}(v,u)\leq 2\lceil k^{\epsilon}\rceil\cdot r_i$ (and hence also, $\dist_{H}(v,u)\leq 6\lceil k^{\epsilon}\rceil\cdot r'_i$).
If $i = \ell$ then under the same conditions we actually have $\dist_{H}(v,u)\leq 2\lceil k^{\epsilon'}\rceil\cdot r_\ell$.
\end{claim}
\begin{proof}
    Let $a=c_{i-1}(v)$ and $b=c_{i-1}(u)$. We show first that $(a,b)\in \hat{G}_{i}$. For $i=1$, this is immediate as $\hat{G}_1=G$. For $i\geq 2$, by \Cref{obs:rad-NFT-step}:
\begin{eqnarray*}
\dist_G(a,b) &\leq& \dist_G(a,v)+\dist_G(v,u)+\dist_G(u,b)
\leq  2r_{i}/3  + r_{i}/3\leq r_{i}~.
\end{eqnarray*}

Since $v$ is $i$-unclustered and $(i-1)$-clustered, it holds that $a$ is not in $V(\mathcal{C}_{i})$. By \Cref{lem:prop-NFT-BSsteps}(4) we have that $\dist_{\hat{H}_{i}}(a,b)\leq 2\lceil k^{\epsilon}\rceil-1$ ($\dist_{\hat{H}_{i}}(a,b)\leq 2\lceil k^{\epsilon'}\rceil-1$ if $i = \ell$) and since each edge in $\hat{H}_{i}$ corresponds to a $G$-path of length at most $r_{i}$, we have that $\dist_{H_{i}}(a,b)\leq (2\lceil k^{\epsilon}\rceil -1)\cdot r_{i}$ ($\dist_{H_{i}}(a,b)\leq (2\lceil k^{\epsilon'}\rceil -1)\cdot r_{i}$ if $i = \ell$).
By \Cref{obs:rad-NFT-step} we have:
\begin{eqnarray*}
\dist_H(v,u) &\leq& \dist_H(v,a)+\dist_H(a,b)+\dist_H(b,u)
\leq  2r_{i}/3+ (2\lceil k^{\epsilon}\rceil-1)r_{i} \leq 2\lceil k^{\epsilon}\rceil \cdot r_{i}~.
\end{eqnarray*}

Similarly, $\dist_H(v,u) \leq 2\lceil k^{\epsilon'}\rceil \cdot r_{\ell}$ for $i = \ell$.
\end{proof}

\begin{claim}\label{cl:key-stretch-NFT}
Let $v$ be an $(i-1)$-clustered and $i$-unclustered vertex and let $u \in B_G(v,r'_i)$. Then, $\dist_H(v,u)\leq 6i\cdot \lceil k^{\epsilon}\rceil \cdot r'_i$. 
\end{claim}

\begin{proof}
The claim is proven by induction. For $i=1$ the claim is immediate by \Cref{lem:prop-NFT-BSsteps}(4). Assume that it holds for $i-1$ and consider an $(i-1)$-clustered and $i$-unclustered vertex $v$ and $u \in B_G(v,r'_i)$. 
Let $P$ be some $v$-$u$ shortest path in $G$. Assume first that $u$ is $(i-1)$-clustered. In this case, by \Cref{cl:key-stretch-NFT-basic}, $\dist_H(v,u)\leq 2\lceil k^{\epsilon}\rceil \cdot r_i=6\lceil k^{\epsilon}\rceil\cdot r'_i$. 

It remains to consider the case where $u$ is $(i-1)$-unclustered.
For a $(q-1)$-clustered and $q$-unclustered vertex $z$, let $r(z)=r'_q$. We traverse $P$ from $u$ to $v$ by taking the following steps until encountering a vertex on $P$ that is $(i-1)$-clustered (since $v$ is $(i-1)$-clustered, such vertex exists). 
Set $j=0$ and $u_0=u$. As long as $\dist_{P}(u_{j},v)\geq r(u_j)$ and $u_j$ is $(i-1)$-unclustered, define $u_{j+1}$ be the vertex at distance $r(u_j)$ from $u_j$ on $P$, and let $j=j+1$. Let $u_{j^*}$ be the last defined vertex. Since $\dist_P(u_j,u_{j+1})=r(u_j)\leq r'_{i-1}$ for every $j \in \{0,\ldots, j^*-1\}$, by the induction assumption we have:
$$\dist_H(u_{j^*},u)\leq 6(i-1)\cdot \lceil k^{\epsilon}\rceil\cdot \dist_G(u_{j^*},u)~.$$

We next bound $\dist_H(v,u_{j^*})$. By the stopping criteria, there are two cases to consider. Assume first that $u_{j^*}$ is $(i-1)$-clustered. Then, by \Cref{cl:key-stretch-NFT-basic}, $\dist_{H}(v,u_{j^*})\leq 2\lceil k^{\epsilon}\rceil\cdot r_i$. Otherwise, we have that $u_{j^*}$ is $(q-1)$-clustered and $q$-unclustered for $q\leq i-1$. Since $v$ is also $(q-1)$-clustered,  $\dist_G(v,u_{j^*})\leq r'_q$ and $r'_q\leq r'_{i-1}$, by \Cref{cl:key-stretch-NFT-basic}, we get that: 
$$\dist_H(v,u_{j^*})\leq 6\cdot \lceil k^{\epsilon}\rceil \cdot r'_{i-1}~.$$
Altogether: $\dist_H(v,u)\leq \dist_H(v,u_{j^*}) + \dist_H(u_{j^*},u) \leq 6i \cdot \lceil k^{\epsilon}\rceil \cdot r'_i$.
\end{proof}

\begin{claim}\label{cl:final-stretch-NFT}
For all $u,v \in V(G)$ we have $\dist_H(u,v) \leq 6\ell\cdot \lceil k^{\epsilon}\rceil\cdot \dist_G(u, v) + O(4^{1/\epsilon}\cdot k/\epsilon)$. 
\end{claim}
\begin{proof}
    Let $P$ be a $u$-$v$ shortest path. Set $j=0$ and $u_0=u$. Let, as before, $r(w)=r'_q$ such that $w$ is $(q-1)$-clustered and $q$-unclustered. 
    As long as $\dist_{P}(u_{j},v)\geq r(u_j)$ define $u_{j+1}$ be the vertex at distance $r(u_j)$ from $u_j$ on $P$, and let $j=j+1$. Let $u_{j^*}$ be the last defined vertex. By \Cref{cl:key-stretch-NFT}, $\dist_H(u, u_{j^*}) \leq 6\ell\cdot \lceil k^{\epsilon}\rceil\cdot \dist_G(u, u_{j^*})$.

    Do the same from the side of $v$: set $j=0$ and $v_0=v$. As long as $\dist_{P}(v_{j},u_{j^*})\geq r(v_j)$ define $v_{j+1}$ be the vertex at distance $r(v_j)$ from $v_j$ on $P$, and let $j=j+1$. Let $v_{j^{**}}$ be the last defined vertex. By \Cref{cl:key-stretch-NFT}, $\dist_H(v, v_{j^{**}}) \leq 6\ell\cdot \lceil k^{\varepsilon}\rceil\cdot \dist_G(v, v_{j^{**}})$.

    Assume first that either $u_{j^*}$ or $v_{j^{**}}$ is $(q-1)$-clustered and $q$-unclustered with $q \leq \ell-1$. Then we can again use \Cref{cl:key-stretch-NFT} to get $\dist_{H}(u_{j^*}, v_{j^{**}})\leq 6(\ell-1)\cdot \lceil k^{\epsilon}\rceil\cdot r'_{\ell-1} =O(1/\epsilon\cdot 4^{1/\epsilon}\cdot k^{(\ell - 1)\epsilon}) = O(4^{1/\epsilon}\cdot k/\epsilon)$ since $\dist_{G}(u_{j^*}, v_{j^{**}})\leq r_q' \leq r_{\ell-1}'$ by construction of $u_{j^*}$ and $v_{j^{**}}$.

    The remaining case is that both $u_{j^*}$ and $v_{j^{**}}$ are $(\ell - 1)$-clustered (and $\ell$-unclustered). Then can apply the special case $i = \ell$ of \Cref{cl:key-stretch-NFT-basic} and get $\dist_{H}(u_{j^*}, v_{j^{**}})\leq 6\cdot \lceil k^{\epsilon'}\rceil\cdot r'_{\ell} = O(4^{1/\epsilon}\cdot k^{(\ell - 1)\epsilon + \epsilon'}) = O(4^{1/\epsilon}\cdot k)$.

    The claim follows by the triangle inequality.  
\end{proof}

\begin{corollary}
    $H$ is a $(\alpha=O(\lceil k^{\epsilon}\rceil /\epsilon), \beta=O(4^{1/\epsilon}\cdot k/\epsilon))$-spanner. For example, for $\epsilon=400/\log k$ this provides stretch $O(\log k)$ for pairs at distance $\Omega(k^{1.01})$.
\end{corollary}

\subsection{$f$-VFT $(O_{\epsilon}(k^{\epsilon}),O_{\epsilon}(k))$-Spanner}\label{sec:BS-VFT-sec}

Recall \Cref{def:RP} and let  $\mathcal{G}_{f,L}$ be an $(L,f)$-\emph{\RP}. The algorithm has $\ell+1$ phases for $\ell=\lceil 1/\epsilon \rceil -1$. In each phase $i \geq 0$, we apply a (fault-tolerant) variant of the Baswana-Sen algorithm on a graph with in expectation at most $N_i$ vertices $V_i$. Initially $N_0=n$, $V_0=V$, and for every $i \in \{1,\ldots, \ell\}$:
\begin{equation}\label{eq:numberi-ni-FT-BS}
    N_i=\frac{n^{1-1/k^{1-i\epsilon}}}{g^{4^{2\ell-i+3}\cdot k^{i\epsilon}}} \mbox{~where~~}g=c  f \cdot \log n~,
\end{equation}
where $c$ is the constant from \Cref{cl:RPC-size}.

\paragraph{Phase 0: $\lceil k^{\epsilon}\rceil$-Steps of $f$-VFT Baswana-Sen}
Given the $n=N_0$-vertex graph $G$, the algorithm applies $\lceil k^{\epsilon} \rceil$ steps of the $f$-VFT Baswana-Sen by \cite{Parter22} with the only modification is that we set sampling probability $p_0$ to be considerably smaller\footnote{In the original algorithm of \cite{Parter22}, $p=(f/n)^{1/k}$.}:
\begin{equation}\label{eq:p0}
    p_0=\frac{1}{g^{4^{2\ell+2}}\cdot n^{1/k}}~.
\end{equation} 
At the end of this phase, we have a subgraph $H_0$ with $\tilde{O}(f\cdot n/p_0)$ edges and an $f$-VFT clustering $\mathcal{C}_0$ of depth $\lceil k^{\epsilon}\rceil$ rooted at centers $S_0 \subseteq V$. The algorithm also outputs a subset $V^*$ of all the vertices that are $f$-VFT clustered in $\mathcal{C}_0$, namely, belonging to $\Omega(k f)$ clusters in $\mathcal{C}_0$.

The following claim is immediate by \cite{Parter22}:
\begin{claim}\label{cl:guarantee-unclustered-VFT}
    For every $G$-edge $(u,v)$ for either $u$ or $v$ not in $V^*$ it holds that $\dist_{H_0 - F}(u,v)\leq 2\lceil k^\epsilon \rceil-1$ for every $F \in (V \setminus \{u,v\})^{\leq f}$. 
\end{claim}

\paragraph{Phase $i \in \{1,\ldots, \ell-1\}$.} 
The input to this phase is the set of centers $S_{i-1}$ of phase $i-1$. Denote $V_i=S_{i-1}$. Recall that $S_0$ is the set of $n\cdot p_0^{\lceil k^\epsilon \rceil}\leq n^{1-1/k^{1-\epsilon}}/g^{4^{2\ell+2}\cdot k^\epsilon}$ clusters, in expectation obtained at the end of Phase $0$. 

The algorithm applies the non-faulty BS procedure on a collection of virtual graphs defined on the vertex set $V_i$ based on their distances in a collection of $G$-subgraphs $\mathcal{G}_{f,r_i}$:

\begin{equation}\label{eq:numberi-FT-BS}
r_0=1 \mbox{~~and~~} r_i=3\lceil k^{\epsilon}\rceil\cdot \sum_{j=0}^{i-1}r_j \mbox{~for~} i \in \{1,\ldots,\ell\}.
\end{equation}

Importantly, the algorithm uses the same random coins in each of the following $|\mathcal{G}_{f,r_i}|$ iterations. 

\begin{itemize}
    \item For every $G_{i,j} \in \mathcal{G}_{f,r_i}$ do:
    \begin{enumerate}
        \item Define a virtual graph $\hat{G}_{i,j}=(V_i,\hat{E}_{i,j})$ where:
$$\hat{E}_{i,j}=\{(u,v) \in V_{i}\times V_{i} ~\mid~ \dist_{G_{i,j}}(u,v)\leq r_{i}\}.$$ 

   \item Apply $\lceil k^{\epsilon}\rceil$ Baswana-Sen steps on the graph $\hat{G}_{i,j}$ with the trivial initial (singleton) clustering and sampling probability of:

   \begin{equation}\label{eq:p-rad-FT-BS}
 p_i=\frac{N_i \cdot g^{R_i}}{n} \mbox{~where~} R_i=4^{i} k^{i\epsilon}~.
\end{equation}

  Let $\hat{H}_{i,j}$ and $\mathcal{C}_{i,j}$ be the output $\hat{G}_{i,j}$-subgraph and clustering, respectively, at the end of these $\lceil k^{\epsilon}\rceil$ steps. Denote the cluster centers of $\mathcal{C}_{i,j}$ by $S_i\subseteq V_i$. Since we use the same random coins for every iteration $j$, the set $S_i$ is the same in each iteration.

  \item Let $H_{i,j}$ be a $G$-subgraph obtained by taking a $u$-$v$ shortest path in $G_{i,j}$ for every $(u,v)\in \hat{H}_{i,j}$:
$$H_{i,j}=\bigcup_{(u,v)\in \hat{H}_{i,j}}SP(u,v,G_{i,j})~.$$
\end{enumerate}
    
\end{itemize}
This concludes the description of Phase $i$.  

\paragraph{Phase $i = \ell$.} This phase does exactly the same as previous phases, with the only exception that it applies only $\lceil k^{\epsilon'}\rceil$ steps of Baswana-Sen to each of the $\hat{G}_{\ell,j}$, where $\epsilon' = \epsilon - (\epsilon\lceil \frac1\epsilon \rceil - 1)$.

The output spanner is given by: 
$$H=H_0 \cup \bigcup_{i=1}^\ell \bigcup_{j=1}^{|\mathcal{G}_{f,r_i}|} H_{i,j}.$$

\noindent \textbf{Size Analysis.} We start by bounding the size of $H_i:= \bigcup_{j}H_{i,j}$ for every $i \in \{0,\ldots, \ell\}$. We make the following claims. 

\begin{claim}\label{cl:ni-size-FT}
    For every $0 \leq i\leq \ell$, $|V_i|\leq N_i$ in expectation and $|V_{\ell+1}| = 0$ w.h.p.
\end{claim}
\begin{proof}
The proof is shown by induction on $i$. For $i=0$, $|V_0|=n$. Assume it holds up to $i$ and consider $i+1$. The set of vertices $V_{i+1}$ corresponds to the output cluster centers of applying $\lceil k^{\epsilon}\rceil$ BS steps in each $V_i$-vertex graph $G_{i,j}\in \mathcal{G}_{f,r_i}$. Since we use the same random coins in each of these applications, using \Cref{lem:prop-NFT-BSsteps}(1), we have that in expectation:
$$|V_{i+1}|=|V_i| \cdot p_i^{\lceil k^{\epsilon}\rceil} \leq|V_i| \cdot p_i^{k^{\epsilon}} \leq N_i \cdot p_i^{k^{\epsilon}}\leq N_{i+1},$$ 
where the second inequality follows by the induction assumption, the claim follows by plugging Eq. (\ref{eq:numberi-ni-FT-BS}) and (\ref{eq:p-rad-FT-BS}).

Similarly,
$$|V_{\ell+1}|=|V_\ell| \cdot p_\ell^{\lceil k^{\epsilon'}\rceil} \leq|V_\ell| \cdot p_\ell^{k^{\epsilon'}} \leq N_\ell \cdot p_\ell^{k^{\epsilon'}} = o(1),$$
and so $|V_{\ell+1}| = 0$ with high probability.
\end{proof}

By the VFT spanner algorithm of \cite{Parter22}, it holds that $|H_0|=\tilde{O}(fn/p_0)$, hence by Eq. (\ref{eq:p0}):
\begin{claim}\label{cl:ni-size-FT-baseH}
   $|H_0|=\tilde{O}((c f \log n)^{4^{2\ell+2}+1}n^{1+1/k})~.$ 
\end{claim}

\begin{claim}\label{cl:ni-size-FT-H}
   Fix $i \in \{1,\ldots, \ell\}$, we have $\mathbb E[|H_i|]=O(4^{i} \cdot k^{i\epsilon+1} \cdot g^2 \cdot n)$. 
\end{claim}
\begin{proof}
   Since $R_i\geq r_i$, for every $i$, by \Cref{cl:RPC-size}, we have that $|\mathcal{G}_{f,r_i}|\leq (c f \log n)^{R_i+2}$. By Eq. (\ref{eq:p-rad-FT-BS}), \Cref{cl:ni-size-FT} and \Cref{lem:prop-NFT-BSsteps}(3), for every $j \in \{1,\ldots, |\mathcal{G}_{f,r_i}|\}$, we have that:
$$\mathbb E[|\hat{H}_{i,j}|]=O\left(k^{\epsilon}\cdot \frac{\mathbb E[|V_i|]}{p_i}\right)\leq O\left(k^{\epsilon}\cdot \frac{N_i}{p_i}\right)=O\left(\frac{k^{\epsilon}\cdot n} {g^{R_i}}\right)~.$$
As each edge in $\hat{H}_{i,j}$ corresponds to a $G_{i,j}$-path of length at most $r_i$, we have that $\mathbb E[|H_{i,j}|]=O(4^{i} \cdot k^{i\epsilon+1} \cdot n / g^{R_i})$. Since $|\mathcal{G}_{f,r_i}|\leq g^{R_i+2}$, by linearity of expectation we have that:
$$\mathbb E[|H_i|]=O(4^{i} \cdot k^{i\epsilon+1} \cdot g^2 \cdot n),$$
and, using linearity of expectation once again, 
$$\mathbb E[|\bigcup_{i = 1}^{\ell} H_i|]=\tilde{O}(4^{1/\epsilon} \cdot g^2 \cdot n)~.$$
\end{proof}

Let $B = \tilde{O}(4^{1/\epsilon} \cdot g^2 \cdot n)$ be an upper bound on $|\bigcup_{i = 1}^{\ell} H_i|$. By Markov's inequality, $|\bigcup_{i = 1}^{\ell} H_i|\leq 4B$ with probability at least $3/4$. With probability $1-o(1)$, all RPC families have their required covering properties, Phase 0 satisfies its fault-tolerance guarantees and the size bound of \Cref{cl:ni-size-FT-baseH}, and $V_{\ell + 1} = \emptyset$. By a union bound, these events and the size event hold simultaneously with positive probability. Fix such an execution for the stretch analysis.

\paragraph{Stretch Analysis.} We start by defining the notion of clustering in the presence of faults. Throughout, fix a subset of at most $f$ vertex faults $F$. A vertex $v$ is $(0,F)$-clustered if $v \in V^*$ and otherwise it is $(0,F)$-unclustered. By \cite{Parter22}, we have:
\begin{claim}[\cite{Parter22}]\label{cl:clustered-VFT-BS}
    For every $v \in V \setminus (V^* \cup F)$, it holds that $\dist_{H  -F}(u,v)\leq 2\lceil k^{\epsilon}\rceil -1$ for any neighbor $u \notin F$ of $v$.
\end{claim}

For every $i \in \{1,\ldots, \ell\}$, we first define the notion of $i$-clustered for vertices in $V_i$ and then extend it to any vertex in $V$. A vertex $v \in V_{i}$ is $(i,F)$-clustered if there exists an index $j \in \{1,\ldots, |\mathcal{G}_{f,r_i}|\}$ such that: 
$$v \in V(\mathcal{C}_{i,j}) \mbox{~~and~~} F \cap V(G_{i,j})=\emptyset,$$
and define $c_{i,F}(v)$\footnote{When $F$ is understandable from  context, we shall sometimes omit it and write $c_{i}(v)$.} to be its cluster center in $\mathcal{C}_{i,j}$.
A vertex $v \in V$ is $(i,F)$-clustered if $v$ is $(i-1,F)$-clustered and in addition, $c_{i-1,F}(v)$ is $(i,F)$-clustered. Since $c_{i-1,F}(v) \in V_i$, this is well-defined. For such a vertex $v$ we define $c_{i, F}(v) = c_{i,F}(c_{i-1, F}(v))$.
For each vertex $v \in V^*\setminus F$ we define $c_{0,F}(v)$ to be a center in $S_0$ such that there is a path of length $\leq \lceil k^\epsilon \rceil$ from $v$ to $c_{0,F}(v)$ that survives faults $F$. Such a path exists by definition of the fault-tolerant clustering. 
For ease of notation throughout this section, let $\mathcal{C}_{-1}=\{\{v\}, v \in V\}$ with $c_{-1,F}(v)=v$ and thus every $v \in V \setminus F$ is $(-1,F)$-clustered. 
Since $V_{\ell+1}=\emptyset$, we have:
\begin{observation}\label{obs:final}
    Any vertex $u$ is $(i,F)$-unclustered for some $i\leq \ell$.
\end{observation}

\begin{observation}\label{obs:rad-FT-step}
    Let $i \in \{0,\ldots, \ell-1\}$. For any $(i,F)$-clustered vertex $v \in V \setminus F$ it holds that $\dist_{H - F}(v,c_{i,F}(v))\leq r_{i+1}/3$.
\end{observation}
\begin{proof}
The proof is shown by induction on $i$. For $i=0$, we have that $\dist_{H - F}(v, c_{0,F}(v))\leq \lceil k^{\epsilon}\rceil \leq r_1/3$ as needed, see Eq. (\ref{eq:numberi-FT-BS}). Assume that the claim holds up to $i-1$ and consider an $(i,F)$-clustered vertex $v$.
Let $z=c_{i-1,F}(v)$. By definition, we have that $c_{i,F}(z)=c_{i,F}(v)$. Let $G_{i,j} \in \mathcal{G}_{f,r_i}$ be such that $z \in V(\mathcal{C}_{i,j})$, $c_{i,F}(z)$ is the center of $z$ in $\mathcal{C}_{i,j}$ and $F \cap V(G_{i,j})=\emptyset$.

By \Cref{lem:prop-NFT-BSsteps}(2) and plugging $x=\lceil k^{\epsilon}\rceil$, we have $\dist_{\hat{H}_{i,j}}(z,c_{i,F}(v))\leq \lceil k^{\epsilon}\rceil$. As each edge in $\hat{H}_{i,j}$ corresponds to a path of length at most $r_i$ in $G_{i,j}$, we have $\dist_{H_{i,j}}(z,c_i(z))\leq \lceil k^{\epsilon}\rceil \cdot r_i$. By combining with the induction assumption 
for $i-1$:
\begin{eqnarray*}
    \dist_{H - F}(v,c_{i,F}(v))\leq \dist_{H - F}(v,z)+\dist_{H - F}(z,c_{i}(v))\leq r_{i}/3+ \lceil k^{\epsilon}\rceil \cdot r_{i}=\lceil k^{\epsilon} \rceil\cdot \sum_{j=0}^{i} r_j=r_{i+1}/3~.
\end{eqnarray*}
   
\end{proof}

Let $r'_0=1$ and $r'_i=r_i/3$ for every $i \in \{1,\ldots, \ell\}$.
\begin{claim}\label{cl:key-stretch-FT-basic}
Let $v$ be an $(i-1,F)$-clustered and $(i,F)$-unclustered vertex. For any $(i-1,F)$-clustered vertex $u \in B_{G- F}(v,r'_i)$, it holds that
    $\dist_{H- F}(v,u)\leq 2\lceil k^{\epsilon}\rceil\cdot r_i$ (and hence, $\dist_{H - F}(v,u)\leq 6\lceil k^{\epsilon}\rceil\cdot r'_i$).
    
Moreover, if $i = \ell$ we actually have $\dist_{H- F}(v,u)\leq 2\lceil k^{\epsilon'}\rceil\cdot r_i$.
\end{claim}
\begin{proof}
For $i=0$, the claim holds by \Cref{cl:clustered-VFT-BS}. From now on consider $i\geq 1$.
Let $a=c_{i-1,F}(v)$ and $b=c_{i-1,F}(u)$. We show first that there exists an index $j \in \{1,\ldots, |\mathcal{G}_{f,r_i}|\}$, that satisfies the following: (i) $V(G_{i,j})\cap F=\emptyset$ and (ii) $\dist_{G_{i,j}}(a,b)\leq r_i$. By \Cref{obs:rad-FT-step}:
\begin{eqnarray*}
\dist_{G - F}(a,b) &\leq& \dist_{G - F}(a,v)+\dist_{G - F}(v,u)+\dist_{G- F}(u,b)
\leq  2r_{i}/3+ r_{i}/3 \leq r_{i}~.
\end{eqnarray*}

Hence by the properties of the graph family $\mathcal{G}_{f,r_i}$, there exists $G_{i,j}\in \mathcal{G}_{f,r_i}$ satisfying (i) and (ii). We therefore conclude that $(a,b)\in \hat{G}_{i,j}$.

Since $v$ is $(i,F)$-unclustered and $(i-1,F)$-clustered, it holds that $a$ is not in $V(\mathcal{C}_{i,j})$. By \Cref{lem:prop-NFT-BSsteps}(4) we have that $\dist_{\hat{H}_{i,j}}(a,b)\leq 2\lceil k^{\epsilon}\rceil-1$ and since each edge in $\hat{H}_{i,j}$ corresponds to a $G_{i,j}$-path of length at most $r_{i}$, we have that $\dist_{H_{i,j}}(a,b)\leq (2\lceil k^{\epsilon}\rceil-1)\cdot r_{i}$. Since $F \cap V(G_{i,j})=\emptyset$, we have that also $\dist_{H - F}(a,b)\leq (2\lceil k^{\epsilon}\rceil -1)\cdot r_{i}$.
By \Cref{obs:rad-FT-step}, we have:
\begin{eqnarray*}
\dist_{H - F}(v,u) \leq \dist_{H - F}(v,a)+\dist_{H - F}(a,b)+\dist_{H - F}(b,u)
\leq  2r_{i}/3+ (2\lceil k^{\epsilon}\rceil-1)r_{i} \leq 2\lceil k^{\epsilon}\rceil\cdot r_{i}.
\end{eqnarray*}

The "moreover" part follows by exactly the same argument with $\epsilon$ replaced by $\epsilon'$, since on the $\ell^\text{th}$ phase the algorithm does only $\lceil k^{\epsilon'}\rceil$ Baswana-Sen steps instead of $\lceil k^{\epsilon}\rceil$.
\end{proof}

\begin{claim}\label{cl:key-stretch-FT}
Let $v$ be an $(i-1,F)$-clustered and $(i,F)$-unclustered vertex and let $u \in B_{G- F}(v,r'_i)$. Then, $\dist_{H- F}(v,u)\leq 6(i+1)\cdot \lceil k^{\epsilon}\rceil \cdot r'_i$.
\end{claim}

\begin{proof}
The claim is proven by induction. For $i=0$ the claim is immediate by \Cref{cl:clustered-VFT-BS}. 
Assume that it holds for $i-1$ and consider an $(i-1,F)$-clustered and $(i,F)$-unclustered vertex $v$ and $u \in B_{G- F}(v,r'_i)$.
Let $P$ be some $v$-$u$ shortest path in $G - F$. Assume first that $u$ is $(i-1,F)$-clustered. In this case, by \Cref{cl:key-stretch-FT-basic}, $\dist_{H - F}(v,u)\leq 6\lceil k^{\epsilon}\rceil\cdot r'_i$.

It remains to consider the case where $u$ is $(i-1,F)$-unclustered.
For a $(q-1,F)$-clustered and $(q,F)$-unclustered vertex $z$, let $r(z)=r'_q$. We traverse $P$ from $u$ to $v$ by taking the following steps until encountering a vertex on $P$ that is $(i-1,F)$-clustered (since $v$ is $(i-1,F)$-clustered, such vertex exists). 
Set $j=0$ and $u_0=u$. As long as $\dist_{P}(u_{j},v)\geq r(u_j)$ and $u_j$ is $(i-1,F)$-unclustered, define $u_{j+1}$ be the vertex at distance $r(u_j)$ from $u_j$ on $P$, and let $j=j+1$. Let $u_{j^*}$ be the last defined vertex. Since $\dist_P(u_j,u_{j+1})=r(u_j)\leq r'_{i-1}$ for every $j \in \{0,\ldots, j^*-1\}$, by the induction assumption we have:
$$\dist_{H - F}(u_{j^*},u)\leq 6\cdot i\cdot \lceil k^{\epsilon}\rceil\cdot \dist_{G - F}(u_{j^*},u)~.$$

We next bound $\dist_{H - F}(v,u_{j^*})$. By the stopping criteria, there are two cases to consider. Assume first that $u_{j^*}$ is $(i-1,F)$-clustered. Then, by \Cref{cl:key-stretch-FT-basic}, $\dist_{H - F}(v,u_{j^*})\leq 6\lceil k^{\epsilon}\rceil \cdot r'_i$. Otherwise, we have that $u_{j^*}$ is $(q-1,F)$-clustered and $(q,F)$-unclustered for $q\leq i-1$. 
Since $v$ is also $(q-1,F)$-clustered,  $\dist_{G-F}(v,u_{j^*})\leq r'_q$ and $r'_q\leq r'_{i-1}$, by \Cref{cl:key-stretch-FT-basic}, we get that:
$$\dist_{H-F}(v,u_{j^*})\leq 6\cdot \lceil k^{\epsilon}\rceil \cdot r'_{i-1}~.$$
Altogether: $\dist_{H - F}(v,u)\leq \dist_{H - F}(v,u_{j^*}) + \dist_{H - F}(u_{j^*},u) \leq 6(i+1) \cdot \lceil k^{\epsilon}\rceil\cdot r'_i$.
\end{proof}

\begin{claim}\label{cl:final-stretch-FT} 
For all $u,v \in V(G)$ and $F \subset V(G)$ of size at most $f$ we have $\dist_{H-F}(u,v) \leq 6(\ell+1)\cdot \lceil k^{\epsilon}\rceil\cdot \dist_{G-F}(u, v) + O(4^{1/\epsilon}\cdot k/\epsilon)$. 
\end{claim}
\begin{proof}
    The proof is identical to the proof of \Cref{cl:final-stretch-NFT}.
    
    Let $P$ be a $u$-$v$ shortest path in $G-F$. Set $j=0$ and $u_0=u$. Let, as before, $r(w)=r'_q$ such that $w$ is $(q-1,F)$-clustered and $(q,F)$-unclustered. 
    As long as $\dist_{P}(u_{j},v)\geq r(u_j)$ define $u_{j+1}$ be the vertex at distance $r(u_j)$ from $u_j$ on $P$, and let $j=j+1$. Let $u_{j^*}$ be the last defined vertex. By \Cref{cl:key-stretch-FT}, $\dist_{H-F}(u, u_{j^*}) \leq 6(\ell+1)\cdot \lceil k^{\epsilon}\rceil\cdot \dist_{G-F}(u, u_{j^*})$.

    Do the same from the side of $v$: set $j=0$ and $v_0=v$. As long as $\dist_{P}(v_{j},u_{j^*})\geq r(v_j)$ define $v_{j+1}$ be the vertex at distance $r(v_j)$ from $v_j$ on $P$, and let $j=j+1$. Let $v_{j^{**}}$ be the last defined vertex. By \Cref{cl:key-stretch-FT}, $\dist_{H-F}(v, v_{j^{**}}) \leq 6(\ell+1)\cdot \lceil k^{\varepsilon}\rceil\cdot \dist_{G-F}(v, v_{j^{**}})$.

    Assume first that either $u_{j^*}$ or $v_{j^{**}}$ is $(q-1,F)$-clustered and $(q,F)$-unclustered with $q \leq \ell-1$. Then we can again use \Cref{cl:key-stretch-FT} to get $\dist_{H-F}(u_{j^*}, v_{j^{**}})\leq 6\ell\cdot \lceil k^{\epsilon}\rceil\cdot r'_{\ell-1} =O(1/\epsilon\cdot 4^{1/\epsilon}\cdot k^{(\ell - 1)\epsilon}) = O(4^{1/\epsilon}\cdot k/\epsilon)$ since $\dist_{G-F}(u_{j^*}, v_{j^{**}})\leq r_q' \leq r_{\ell-1}'$ by construction of $u_{j^*}$ and $v_{j^{**}}$.

    The remaining case is that both $u_{j^*}$ and $v_{j^{**}}$ are $(\ell - 1, F)$-clustered (and $\ell$-unclustered). Then we can apply the special case $i = \ell$ of \Cref{cl:key-stretch-FT-basic} and get $\dist_{H-F}(u_{j^*}, v_{j^{**}})\leq 6\cdot \lceil k^{\epsilon'}\rceil\cdot r'_{\ell} = O(4^{1/\epsilon}\cdot k^{(\ell - 1)\epsilon + \epsilon'}) = O(4^{1/\epsilon}\cdot k)$.

    The claim follows by the triangle inequality.
 
\end{proof}

\begin{corollary}\label{cor:VFT-FT-sampling-BS}
    $H$ is a $f$-VFT $(\alpha=O(k^{\epsilon}/\epsilon), \beta=O(4^{1/\epsilon}\cdot k/\epsilon))$-spanner with $\tilde{O}((c \cdot f \cdot \log n)^{4^{2/\epsilon}}\cdot n^{1+1/k})$ edges.
\end{corollary}

%% file: FT-Clustering-sec.tex
\section{Path Systems Emerging in the FT-Greedy}\label{sec:greedyFTclustering}

We begin with the definition of $d\to r$ spanners.
\begin{definition}[\cite{PopovaT26}]
    A $d \to r$ spanner of a graph $G$ is a subgraph $H$ such that for every two vertices $u,v\in V(G)$ at distance exactly $d$ in $G$, we have $\dist_{H}(u,v) \leq r$.
\end{definition}

For unweighted graphs it is folklore that $d \to r$-spanners and $(\alpha,\beta)$-spanners are equivalent up to negligible losses in size bounds. We include in \Cref{sec:equiv-d-to-r-alpha-beta} a short discussion on this equivalence. By applying the reductions in  \Cref{sec:equiv-d-to-r-alpha-beta} to the VFT setting, we obtain that it suffices to construct a VFT $d\to r$ spanner for $d= \lfloor\sqrt{k}\rfloor$, and $r=\frac{12k}{\delta}$, of size $\tilde{O}_{\delta}(f^{1+\delta}n^{1+1/k})$ to get a $(\frac{24\sqrt{k}}{\delta},\frac{12k}{\delta})$-spanner of size $ \sqrt{k} \cdot \tilde{O}_{\delta}(f^{1+\delta}n^{1+1/k}) = \tilde{O}_{\delta}(f^{1+\delta}n^{1+1/k})$ as one may assume  $k\leq \log(n)$, which will conclude \Cref{thm:sqrtk-greedy} \footnote{In \Cref{sec:sqrtk-greedy}, we assume $0<\delta<1/2$, instead of $0<\delta<1$ as stated in \Cref{thm:sqrtk-greedy}, but this can be done without loss of generality, as the size bound gets \emph{sharper} as $\delta$ decreases, hence the result for say $\delta = 1/3$, applies for all $\delta \in [1/3,1)$ .}. Similarly, proving a bound of $O(k^3fn^{1+1/k})$ for the $2\to 2k$ VFT spanner will imply \Cref{thm:hybrid-greedy}.

\subsection{Greedy VFT $d\to r$ spanner}

There is a natural greedy algorithm to construct such a spanner, \Cref{alg:vft-greedy}.

\begin{algorithm}[H]
\caption{\textsc{VFT-Greedy $d\to r$ Spanner}$(\G,d,r,f)$}
\label{alg:vft-greedy}
\DontPrintSemicolon

$H\gets(\V,\emptyset)$\;

\ForEach{$x,y\in V$}{
  \ForEach{Path $P$ of length $d$ from $x$ to $y$}{
    \If{$\exists\,F\subseteq V$ with $|F|\le f$ and
      $F\cap V(P)=\emptyset$ s.t.\
      $\dist_{H-F}\big(x,y\big)>r$}{
      $E(H)\gets E(H)\cup E(P)$\;
    }
  }
}

\Return $H$\;
\end{algorithm}

Throughout, we denote by $(P_1,F_1),\ldots,(P_m,F_m)$ the pairs consisting of a path $P_i$ added by the VFT greedy algorithm and the corresponding fault set $F_i$ that causes it to be added. We use $H_i$ to denote the state of the spanner \emph{before} $P_i$ is added, that is, $H_i=\bigcup_{j<i} P_j$ for $i\geq 2$, and $H_1=(V,\emptyset)$.

We first verify that \Cref{alg:vft-greedy} outputs a valid spanner. Indeed, suppose that there are $F\subseteq V(G)$, $u, v\in V(G)$ that violate the definition of the VFT spanner for the constructed $H$, i.e. $|F| \leq f$, $\dist_{G-F}(u,v) = d$, but $\dist_{H-F}(u,v) > r$. Then there exists a path $P$ of length exactly $d$ that connects $u$ and $v$ in $G-F$. But then on the step of the algorithm that considers $P$ it should be added because of $F$. Thus $E(P)\subseteq E(H)$ and $V(P)\subseteq V(H) -F$, and consequently $\dist_{H-F}(u,v) = d \leq r$, contradiction. 

Thus our goal from now on is to upper bound the number of edges in the output of this algorithm. For the analysis, we study a path system contained in the spanner $H$ and track its evolution throughout the execution, so we begin by introducing some notation for path collections.

\paragraph{Notation for paths and path collections}
We view a path in a graph $G$ from $x$ to $y$ as a sequence $(x=x_0,x_1,\ldots,x_{\ell}=y)$ where
$\{x_i,x_{i+1}\}\in E(G)$.
We write $V(P)$ for the vertex set of $P$, and denote its length by $L(P):=\ell$.
For an oriented edge $(u,v)$ and a path $P$ starting at $v$, we write $(u,v)\circ P$ for the concatenation obtained by prepending $(u,v)$ to $P$.

Let $\pp$ be a collection of oriented paths in a graph $G$. For $J,F\subseteq V(G)$, define
\[
\pp[J]:=\{P\in\pp\mid J\subseteq V(P)\},
\qquad
\pp-F:=\{P\in\pp\mid V(P)\cap F=\emptyset\}.
\]
For $\ell\ge 0$ and $v,w\in V(G)$, let $\pp^\ell$ denote the paths of length exactly $\ell$, $\pp_v$ the paths starting at $v$, and $\pp_{v,w}$ the paths starting at $v$ and ending at $w$. For a path collection $\pp_u$ we denote $(v,u)\circ\pp_u = \{(v,u)\circ P \mid P \in \pp_u\}$.

We begin by defining \emph{critical} edges:

\begin{definition}
    Given an integer $s$, we say an edge $\{u,v\}=e\in E(P_i)$ is $s$-critical if
\[
\dist_{H_i-F_i}(u,v) > s.
\]
\end{definition}

\begin{claim}
    For every $s\leq \frac{r}{d}$, and $i\leq m$, the path $P_i$ contains an $s$-critical edge.
\end{claim}

\begin{proof}
    Let $x_i$ and $y_i$ be the endpoints of $P_i$. Assume the contrary, i.e. for every $e = \{u,v\} \in E(P_i)$ we have that $\dist_{H_i-F_i}(u,v)\leq s \leq \frac{r}{d}$. Then, since $P_i$ has exactly $d$ edges, we get, by the triangle inequality, that $\dist_{H_i-F_i}(x_i, y_i) \leq d\cdot \frac{r}{d} = r$, which contradicts the assumption that $P_i$ is added because of $F_i$. 
\end{proof}

We construct the path collection $\pp$ which will be used to bound $m$, the number of paths added by the greedy algorithm. The path collection $\pp$ will depend on two parameters: (1) An integer $D$ which is approximately the average degree in the output spanner, and (2) $s \leq r/d$, for which $s$-critical edges of paths added by the greedy spanner may be used to increase $\pp$. We also partition the paths added by the greedy algorithm into $\{P_1,\ldots,P_m\}=\mathcal{L}\overset{\cdot}{\cup}\mathcal{G}$ -- paths in $\mathcal L$ can be accounted for using our construction of $\pp$, while paths in $\mathcal G$ require a different counting argument, appearing in \Cref{sec:sqrtk-greedy,sec:vft2to2k} . Finally we set $N_{\ell}=D^{\ell}$ to be the bound on the number of paths in $\pp_v^{\ell}$, such that past this bound we think of $v$ as having sufficiently good length $\ell$ connections in $H$ to other vertices. We now describe how $\pp,\mathcal{L},\mathcal{G}$ are constructed.

\begin{algorithm}[H]
\caption{\textsc{Greedy Path-System}$(\{(P_i,F_i)\}_{i=1}^m, f, s, \{N_\ell\}_{\ell = 0}^{\lceil \frac{s}{2}\rceil})$}
\label{alg:greedy-path-system}

$\mathcal{G} \gets \emptyset$, \quad
$\mathcal{L} \gets \emptyset$, \quad
$\mathcal{P} \gets \{(v) \mid v \in V\}$\;

\For{$i=1$ \KwTo $m$}{
    
    \If{for every $s$-critical $e=\{u,v\}\subseteq P_i$, every $\ell\le \lceil \frac{s}{2}\rceil$, and every endpoint $w\in e$, we have
    $|\mathcal{P}_w^{\ell}| \ge N_\ell$}{
        
        $\mathcal{G} \gets \mathcal{G} \cup \{P_i\}$\;
    }
    \Else{

        Choose an $s$-critical edge $e=\{u_i,v_i\}\subseteq P_i$ and $0\leq \ell_i < \lceil \frac{s}{2}\rceil$ such that 
        $|\mathcal{P}_{v_i}^{\ell_i}| \ge N_{\ell_i}$ and 
        $|\mathcal{P}_{u_i}^{\ell_i + 1}| < N_{\ell_i+1}$\;

        $\mathcal{P} \gets \mathcal{P} \cup 
        \bigl((u_i,v_i)\circ(\mathcal{P}_{v_i}^{\ell_i}-F_i)\bigr)$\; 
        $\mathcal{L} \gets \mathcal{L} \cup \{P_i\}$\;
    }
}
\end{algorithm}

Observe that the index $\ell_i$ in the \textbf{else} branch is well defined. Indeed, if a critical edge 
$e=\{u,v\}$ violates the \textbf{if} condition, then consider the maximal $\ell$ for which the \textbf{if} condition holds for all lengths up to $\ell$ inclusively. Observe that this maximum is well defined since $N_0=1\le|\pp_v^{0}|$ for every $v\in V(G)$ throughout the algorithm, and $\ell \leq \lceil \frac{s}{2} \rceil-1$, as the \textbf{if} was violated. Setting $\ell_i=\ell$, by definition of $\ell$ both:
\begin{enumerate}
    \item $|\pp_{u_i}^{\ell_i}|\ge N_{\ell_i}$ and $|\pp_{v_i}^{\ell_i}|\ge N_{\ell_i}$; and

\item $|\pp_{u_i}^{\ell_i+1}|< N_{\ell_i+1}$ or $|\pp_{v_i}^{\ell_i+1}|< N_{\ell_i+1}$.
\end{enumerate} 
By renaming the endpoints of $e$ the claim follows.

In what follows, we shall sometimes refer to the modification of $\pp$ in the \textbf{else} case as "the update rule". For every $i$ such that $P_i\in \mathcal{L}$, we say that a \emph{boost} occurred; in this case, following the notation of \Cref{alg:greedy-path-system}, we say that $u_i$ is boosted at radius $\ell_i+1$.

Here and in what follows we denote by $\pp^i$ the path collection $\pp$ in \Cref{alg:greedy-path-system} before considering the $i$-th path. Similarly, $\pp^{i,\ell}$ denotes the collection of $\ell$-paths in $\pp$ before \Cref{alg:greedy-path-system} considers $P_i$, etc.
We record the following claim which is immediate from the definition of paths in $\mathcal{G}$.

\begin{claim}\label{cl:obvious-prop-g}
    Let $P_i \in \mathcal{G}$ and $e = \{u, v\}$ be an $s$-critical edge on $P_i$. Then $|\pp^{i,\ell}_u| \geq N_\ell$ and $|\pp^{i,\ell}_v| \geq N_\ell$ for all $\ell \in \{0, 1,..., \lceil \frac{s}{2}\rceil\}$.
\end{claim}

Our first main theorem is a bound on $\mathcal L$:

\begin{theorem}\label{thm:greedy-path-system-bounds}
    Let $D \geq 4s^2f$. Then $|\mathcal L| = O(nDs)$.
\end{theorem}

\begin{remark}
    The assumption $D \geq 4s^2f$ is essentially best possible, as a connectivity certificate needs degree $f$, and $s$ is at most $\beta$ for $(\alpha,\beta)$ spanners, hence $O(k)$ in the cases we analyze in the paper.
\end{remark}

To prove this theorem, we show that the path collection $\pp$ satisfies strong pairwise and collective fault-tolerant structure. To do so, we use the notion of parks \cite{PPST25} applied to our VFT setting. Parks provide the local combinatorial framework used to control the pairwise subcollections
$\pp_{v,w}^{i,\ell}$ throughout the update process.

\subsection{Parks and Fault Tolerance}\label{sec:parks}

Parks were introduced in \cite{PPST25} to provide a quantitative framework to track properties of certificates between two vertices in a graph. Parks depend on a \emph{score function}, which is a function that takes as input subsets of faulty elements (vertices/edges depending on the fault setting, in our case, vertices) and outputs a real number describing how non-faulty the collection is. 

We use a slight variant of the scores appearing in \cite{PPST25}.

\begin{definition}[Score function \cite{PPST25}]\label{def:score}
Let $0<\beta\le 1$.
The \emph{$(2,\beta)$-score function} $\sco^{\beta}(\cdot)$ assigns a path $P$ the score
\[
\sco^{\beta}(P) := \beta (2f)^{-(|L(P)|+1)}.
\]

For $J\subseteq V$, the \emph{induced score function on $J$} is
\[
\sco^{\beta}_J(P):=
\begin{cases}
\beta (2f)^{-\left(|L(P)|-|J|\right)-1} & \text{if $J\subseteq V(P)$,}\\
0 & \text{if $J\not\subseteq V(P)$.}
\end{cases}
\]
Thus, $\sco^{\beta}_{\emptyset}(\cdot)=\sco^{\beta}(\cdot)$.
We extend this definition to a path collection $\pp$ linearly:
\[
\sco^{\beta}_J(\pp) := \sum_{P\in \pp} \sco^{\beta}_J(P).
\]
Note that $\sco^{\beta}_J(\pp)=\sco^{\beta}_J(\pp[J])$.
\end{definition}

\begin{definition}[Park \cite{PPST25}]\label{def:park}
A collection of paths $\pp$ is a \emph{park} with respect to a score function $\sco(\cdot)$ if for all $J\subseteq V$
it holds that $\sco_J(\pp)\le 1$.%
\footnote{Since for a path $P$, $\sco_{V(P)}(P)=\beta$, we require $\beta\le 1$, otherwise all parks w.r.t.\ the score function are empty.}
\end{definition}

\begin{definition}[$J$-full park \cite{PPST25}]\label{def:full-park}
Let $\pp$ be a park with respect to $\sco(\cdot)$, and let $J\subseteq V$.
We say $\pp$ is \emph{$J$-full} if $\sco_J(\pp)>\tfrac12$.
\end{definition}

We recall the following fault-tolerance lemma, which motivates parks in fault-tolerant settings.
We include the proof for completeness.

\begin{lemma}[Fault-Tolerance of Parks \cite{PPST25}]\label{lemma:union_bound_scores}
Let $\pp$ be a park w.r.t.\ a $(2,\beta)$-score function $\sco(\cdot)$.
Let $F$ be a set of at most $f$ faulty vertices.
If there exists $J\subseteq V$ such that $J\cap F=\emptyset$ and $\sco_J(\pp)>\tfrac12$,
then there is a path $P\in \pp[J]$ with $V(P)\cap F=\emptyset$.
\end{lemma}

\begin{proof}
Note that $\bigcup_{v\in F}\pp[J\cup\{v\}]$ consists of all faulty paths in $\pp[J]$, and
\begin{align*}
\sco_J\!\left(\bigcup_{v\in F}\pp[J\cup\{v\}]\right)
&\le \sum_{v\in F} \sco_J(\pp[J\cup\{v\}]) && \text{(union bound)}\\
&= \sum_{v\in F} \frac{1}{2f}\,\sco_{J\cup\{v\}}(\pp[J\cup\{v\}]) && \text{(by definition of $\sco$, since $v\notin J$)}\\
&\le \frac{1}{2} && \text{(as $\pp$ is a park and $|F|\le f$)}.
\end{align*}
Since $\sco_J(\pp)>\tfrac12$, the $J$-link $\pp[J]$ must contain a non-faulty path.
\end{proof}

Finally, we record a useful fact on the actual size of $J$-links in a park.

\begin{observation}\label{obs:size-bounds-park-links}
    Let $\pp$ be a park, consisting of paths of length $\ell$ w.r.t. $\sco=\sco^{\beta}$, then $|\pp[J]| \leq (1/\beta) (2f)^{\ell-|J|+1}$.
\end{observation}

\begin{proof}
    By definition of the score function, we have 
    $$\sco_{J}(\pp) = \sum_{P\in \pp[J]}\sco_J(P) = |\pp[J]|\cdot \beta(2f)^{-(\ell - |J|)-1}.$$
    Since by the park condition $\sco_{J}(\pp) \leq 1$, we get the desired bound on $|\pp[J]|$.
\end{proof}

\subsection{The Properties of $\pp$}

We now describe the properties of the path collection $\pp$. From here onwards we analyze properties of the greedy path system algorithm, \Cref{alg:greedy-path-system}, which crucially does not depend on $d,r$, but only on $s$. Hence we may use the letters $d,r$ to denote other quantities, which we do and use $r:=\lceil \frac{s}{2} \rceil$. Recall that we assume $N_\ell=D^{\ell}$ for some $D$, as the input for \Cref{alg:greedy-path-system}.  We now prove that $\pp$ has many properties that are useful both for bounding $|\mathcal L|$ as well as for fault tolerant applications.

\begin{center}
\fbox{\begin{minipage}{0.97\linewidth}
\textbf{Invariant (I1--I7).}
For every iteration $i$, every $\ell\in\{0,1,\ldots,r\}$, and all $v,w\in V(G)$:
\begin{enumerate}[label=\textbf{I\arabic*:}, leftmargin=*, itemsep=2pt, topsep=2pt]
\item\label{inv:I1} (\emph{Well defined}) Every path in $\pp^i$ is contained in $H_i$.

\item\label{inv:I2} (\emph{Strong collective FT}) For every vertex $v$, every nonempty $W\subseteq V(G)-\{v\}$, and every $\ell\in[r]$,
\[
|\pp_v^{i,\ell}[W]|\ \le\ 3\ell\cdot s^{2(|W|-1)}\cdot D^{\ell-|W|}.
\]

\item\label{inv:I3} (\emph{Pairwise a park}) For every $\ell\in\{0,\ldots,r\}$ and every $v,w$,
the collection $\pp_{v,w}^{i,\ell}$ is a park w.r.t.\ $\lsco^\ell(\cdot)=\sco^{\beta_\ell}(\cdot)$, where
$\beta_\ell:=10^{-\ell}$.

\item\label{inv:I4} (\emph{Weak collective FT}) For every $u\neq v$ and $\ell\in[r]$,
\[
|\pp_v^{i,\ell}[\{u\}]|\ \le\ 3\ell\cdot D^{\ell-1}.
\]

\item\label{inv:I5} (\emph{Bounded size}) For every $v$ and $\ell\in\{0,\ldots,r\}$,
\[
|\pp_v^{i,\ell}|\ \le\ \frac{D^{\ell+1}-1}{D-1}.
\]

\item\label{inv:I6} (\emph{Progress per boost}) If $u_i$ is boosted at radius $\ell_i+1$, then
\[
|\pp_{u_i}^{i+1,\ell_i+1}|-|\pp_{u_i}^{i,\ell_i+1}|
\ \ge\
D^{\ell_i}-3\ell_i f D^{\ell_i-1}.
\]

\item\label{inv:I7} (\emph{Bounded number of boosts}) For every vertex $v$ and every $\ell\in[r]$,
$v$ is boosted at radius $\ell$ at most $D+4sf$ times.
\end{enumerate}
\end{minipage}}
\end{center}

Each of the mentioned invariants captures different features of the path system . Invariant \textbf{I3} states
that each pairwise collection $\pp_{v,w}^{i,\ell}$ is a park. This ensures that when the path collection $\pp$ grows, it encodes a new set of paths that can handle faults that the previous $\pp$ couldn't. 
\textbf{I2} and \textbf{I4} control how concentrated the larger collections $\pp_v^{i,\ell}$ can be
on small vertex sets, we call these invariants \emph{collective}, as they capture the collective structure of all paths in $\pp_v$, and not merely pairwise properties. We also call the conditions \textbf{I2} and \textbf{I4} fault-tolerant, since when $|\pp_v^{\ell}|\geq N_{\ell}$, this invariants translate to concrete statements on how $\pp_v^{\ell}$ survives under faults, e.g. as explained in \Cref{sec:technical-over-path-systems}, \textbf{I4} implies that $\pp_v-F \neq \emptyset $ for $F\subseteq V-\{v\}$, of size at most $f$. Finally, \textbf{I5}--\textbf{I7} turn each boost into a
quantitative progress statement in terms of the number of paths in $\pp$, which yields the bound on $|\mathcal L|$.
We will prove that all these invariants are preserved in the following lemma. 

\begin{lemma}\label{lem:pp-invariants}
Assuming $s\geq 2$, $f\geq 1$, and $D\geq 4s^2f$, the path collections $\pp^i$ satisfy Invariant \textbf{I1--I7}, where we denote $r := \lceil \frac{s}{2}\rceil$.
\end{lemma}

The proof of the lemma proceeds via a delicate simultaneous induction over all the maintained properties. A key difficulty is that these properties are tightly interdependent, so that none can be established in isolation. This interlocking structure necessitates a particularly careful argument, which we view as one of the main technical contributions of this work.

We defer the proof of \Cref{lem:pp-invariants} to a later section, and first demonstrate its utility through several applications. As a first consequence, we derive \Cref{thm:greedy-path-system-bounds}:

\begin{proof}[Proof of \Cref{thm:greedy-path-system-bounds}]
At each step $i$ such that $P_i \in \mathcal{L}$ one vertex is boosted at some radius $\ell \leq r$. By \textbf{I7} each vertex can be boosted at each radius at most $D + 4sf \leq 2D$ times. Consequently, there are at most $2D\cdot r\cdot n = O(sDn)$ paths in $\mathcal{L}$.
\end{proof}

Second, we show that \Cref{lem:pp-invariants} gives essentially tight bounds on the size of the VFT $(2k-1,0)$-spanner obtained by the standard VFT greedy algorithm. This can be seen as an alternative proof to \cite{BodwinP19}, which is only $poly(k)$ size larger. We include this proof as (1) it is very short given \Cref{lem:pp-invariants} (2) It shows how the park condition for pairs of vertices, \textbf{I3} may be used for bounding $\mathcal{G}$.

\begin{theorem}\label{thm:st-mult-vft}
    The VFT $1\to (2k-1)$ greedy spanner algorithm constructs a multiplicative $(2k-1)$-spanner with $O(k^3f^{1-1/k}n^{1 + 1/k})$ edges.
\end{theorem}

\begin{proof}
    As we consider $1\to 2k-1$ VFT spanners, we denote by $e_i:=P_i=(u_i,v_i)$ the $i^{th}$ edge added by the VFT algorithm. Using \Cref{alg:greedy-path-system} with $s = 2k - 1$ and $D = 30k^2f^{1-1/k}n^{\frac{1}{k}} \geq 4s^2f$ (since $n \geq f$) we decompose the edges added by the algorithm to $\mathcal{L}$ and $\mathcal{G}$.
    
    By \Cref{thm:greedy-path-system-bounds}, $|\mathcal{L}|= O(nDs)=O(k^3f^{1-1/k}n^{1+1/k})$, thus it remains to bound $\mathcal{G}$. We show that for the chosen $D,s$ there are no paths in $\mathcal{G}$. 
    
    Indeed, assume $e_i \in \mathcal{G}$.
    By \Cref{cl:obvious-prop-g}, this means that $|\pp^{i,k}_v| \geq D^k$ for the two endpoints of $e_i$ (as there is no other edges on the path, it is the $s$-critical edge). We claim that $|\pp^{i,k}_v| \geq D^k$ can never hold for any $i$ and $v$. Indeed, it follows from  \textbf{I3}, i.e 
    each $\pp^{i,k}_{v,u}$ being a park, and \Cref{obs:size-bounds-park-links}, that for all $u\neq v$ we have $$|\pp^{i,k}_{v,u}| = |\pp^{i,k}_{v,u}[\{u,v\}]| \leq 10^{k}\cdot (2f)^{k-1}.$$
    Thus 
    $$|\pp^{i,k}_v| = \sum_{u \neq v}|\pp^{i,k}_{v,u}| \leq 10^k\cdot (2f)^{k-1}\cdot n \leq  (20f^{1-1/k}n^{1/k})^k < D^k.$$
\end{proof}

Finally, another consequence of \Cref{lem:pp-invariants}, is that we can also replace the guarantee of \Cref{alg:greedy-path-system}, as appears in \Cref{cl:obvious-prop-g}, to be a guarantee in the flavor of the FT-clustering methods \cite{Parter22,PPST25}, which gives for a collection of paths ending in a \emph{small} set of \emph{centers} $S$. We now formally describe this proposition (in a form which will be useful in application in  \Cref{sec:sqrtk-greedy}). The proof of the proposition appears in \Cref{sec:sampling-appendix}.

\begin{proposition}\label{prop:vft-d-to-r-structure}
    Let $P_1, \ldots, P_m$ be the paths added by the VFT greedy $d \to r$ spanner. Let $F_i$ be the fault set causing $P_i$ to be added. Let $s\leq \frac{r}{d}$ be an integer, $R = \lceil\frac{s}{2}\rceil$, $C'\ge 2$ some constant and let $D \geq 12C's^3f$. Then there exist $S \subseteq V$, $|S|=\tilde{O}\left(\Bigl(\frac{10\cdot C'sf}{D}\Bigr)^R \cdot n \right)$, 
    and a partition of paths $P_1, \ldots, P_m$ into two disjoint families $\mathcal L$ and $\mathcal G$ such that:
    \begin{enumerate}
        \item $|\mathcal L| = O(nDs)$.

        \item For any $P_i \in \mathcal G$, let $H_i$ be the graph right before $P_i=(x_0,\ldots,x_{d})$ was added, then there exist $0\leq \ell <q \leq d$ such that: 
        \begin{itemize}
            \item $\dist_{H_i-F_i}(x_0,x_\ell) \leq \ell\cdot s$ and $\dist_{H_i-F_i}(x_q, x_{d}) \leq (d-q)\cdot s$. 

            \item There is a non-empty path collection $\prepp_{x_{\ell}}$ of paths from $x_\ell$ to $S$ of length $R$ in $H_i$, satisfying \[ \forall U \subseteq V\setminus\{x_\ell\}, \qquad |\prepp_{x_{\ell}}[U]|\leq |\prepp_{x_{\ell}}|(C' s f)^{-|U|}. \] Similarly, there is a path collection $\prepp_{x_q}$ with the analogous property.
        \end{itemize}
    \end{enumerate}
\end{proposition}

\subsection{Proof of \Cref{lem:pp-invariants}}\label{sec:complicated-invariants-lemma}

For the proof we will need the following inequality for scores of concatenations, we include proof for completeness.

\begin{claim}[\cite{PPST25}]\label{cl:concat}
    For any path $P$ that begins in a vertex $v$, an edge $e = uv$ and any $J \subseteq V(G)$ we have
    $$\sco_J(e\circ P) \leq \sco_J(P) + \sco_{J-u}(P).$$
\end{claim}
\begin{proof}
    Since the right side of the inequality is always non-negative, it is enough to prove the claim when left side is positive, i.e. $J \subseteq V(e\circ P)$. Assuming $\sco=\sco^{\beta}$ we have the following cases:
    \begin{itemize}
        \item $J \subseteq V(P)$: then $\sco_J^{\beta}(e\circ P) = \beta(2 f)^{-(L(P)+1 - |J|)-1} = \frac{1}{2 f}\sco_J^{\beta}(P) \leq \sco_J^{\beta}(P)$;
        \item $J = \{u\}\sqcup J'$ and $J' \subseteq V(P)$: then $\sco_J^{\beta}(e\circ P) = \beta(2 f)^{-(L(P)+1 - |J|)-1} = \beta(2 f)^{-(L(P) - |J'|)-1} = \sco_{J'}^{\beta}(P) = \sco_{J-u}^{\beta}(P)$.
    \end{itemize}
\end{proof}

The following lemma is at the heart of the construction. It implies that the update rule of $\pp$ doesn't violate \emph{any} of the pairwise a park conditions.

\begin{lemma}[FT path-exchange lemma]\label{lem:ft-neighborhood-exchange-lemma-gen}
Let $G$ be an undirected graph, let $s\ge 2$, set $r:=\lceil s/2\rceil$, and fix $\ell\le r-1$.
Let $u,v\in V(G)$ and let $F\subseteq V(G)\setminus\{u,v\}$ be a fault set such that 
\[
\dist_{G-F}(u,v)\ >\ s.
\]
Let $\pp$ be a collection of oriented paths in $G$ and assume that for every $w\in V(G)$ the subcollection
$\pp_{u,w}^{\ell+1}$ is a park w.r.t. some score function $\sco$.
Then for every $w\in V(G)$ and every $J\subseteq V(G)$, if
\[
(\pp_{v,w}^{\ell}-F)\bigl[J\setminus\{u\}\bigr]\ \neq\ \emptyset,
\]
we must have that $\pp_{u,w}^{\ell+1}$ is \emph{not} $J$-full (w.r.t.\ $\sco$), i.e.
$\sco_{J}(\pp_{u,w}^{\ell+1})\le \tfrac12$.
\end{lemma}

\begin{proof}
Assume for contradiction that there exist $w\in V(G)$ and $J\subseteq V(G)$ such that
$(\pp_{v,w}^{\ell}-F)[J\setminus\{u\}]\neq\emptyset$ and $\pp_{u,w}^{\ell+1}$ is $J$-full.

Pick a path $P\in (\pp_{v,w}^{\ell}-F)[J\setminus\{u\}]$. Then $P$ avoids $F$ and contains every vertex in
$J\setminus\{u\}$, hence $(J\setminus\{u\})\cap F=\emptyset$. Since also $u\notin F$, we have $J\cap F=\emptyset$.

Because $\pp_{u,w}^{\ell+1}$ is a park and is $J$-full, we may apply \Cref{lemma:union_bound_scores} with the fault set
$F$ (which is disjoint from $J$) to deduce that
\[
(\pp_{u,w}^{\ell+1}-F)[J]\ \neq\ \emptyset.
\]
Let $P'\in(\pp_{u,w}^{\ell+1}-F)[J]$ be such a surviving path. Thus $P'$ is a $u\to w$ path in $G-F$ of length
$\ell+1$, and $P$ is a $v\to w$ path in $G-F$ of length $\ell$. Their concatenation is a path of length $2\ell+1 \leq s$ from $u$ to $v$ that isn't damaged by $F$, contradiction.
\end{proof}

Now we prove the desired properties of $\pp^i$.

\begin{proof}[Proof of \Cref{lem:pp-invariants}]
We prove by induction on $i$.

\emph{Base case ($i=1$).}
By construction, $\pp^1$ consists of all length-$0$ paths, hence \textbf{I1--I5} hold trivially and \textbf{I6--I7} are
vacuous.

\emph{Inductive step.}
Assume \textbf{I1--I7} hold for $\pp^i$. If $P_i \in \mathcal{G}$, then $\pp^{i+1}=\pp^i$ and there is nothing
to prove. Otherwise, for some critical edge $e_i=\{u_i,v_i\}$ with witness fault set $F_i$ and radius $\ell_i$ we have
\[
\pp^{i+1}\ :=\ \pp^{i}\ \cup\ \bigl((u_i,v_i)\circ(\pp_{v_i}^{i,\ell_i}-F_i)\bigr),
\]
and in particular $\pp_x^{i+1,\ell}=\pp_x^{i,\ell}$ unless $x=u_i$ and $\ell=\ell_i+1$.

\medskip
\noindent\textbf{Invariant \textbf{I1}.}
This is immediate from the update rule: all newly added paths lie in $H_{i+1}$.

\paragraph{Subclaim A (Bounded size: \textbf{I5}).}
We prove \textbf{I5} for the only modified block $(u_i,\ell_i+1)$.
By definition of $\ell_i$ we have 
$|\pp_{u_i}^{i,\ell_i+1}|<D^{\ell_i+1}$.
Also, by the inductive hypothesis \textbf{I5},
\[
|\pp_{v_i}^{i,\ell_i}|\ \le\ \frac{D^{\ell_i+1}-1}{D-1}.
\]
Therefore,
\begin{align*}
|\pp_{u_i}^{i+1,\ell_i+1}|
&=
\Big|\pp_{u_i}^{i,\ell_i+1}\ \cup\ \bigl((u_i,v_i)\circ(\pp_{v_i}^{i,\ell_i}-F_i)\bigr)\Big|\\
&\le
|\pp_{u_i}^{i,\ell_i+1}|+|\pp_{v_i}^{i,\ell_i}|\\
&<
D^{\ell_i+1}+\frac{D^{\ell_i+1}-1}{D-1}.
\end{align*}
Finally, note that
\[
D^{\ell_i+1}+\frac{D^{\ell_i+1}-1}{D-1}
=
\frac{D^{\ell_i+2}-1}{D-1},
\]
Thus
\(
|\pp_{u_i}^{i+1,\ell_i+1}|\le \frac{D^{\ell_i+2}-1}{D-1},
\)
and \textbf{I5} follows.

\paragraph{Subclaim B (Progress lower bound: \textbf{I6}).} 
By the update rule,
\[
|\pp_{u_i}^{i+1,\ell_i+1}|-|\pp_{u_i}^{i,\ell_i+1}|
=
|(u_i,v_i)\circ(\pp_{v_i}^{i,\ell_i}-F_i)|
=
|\pp_{v_i}^{i,\ell_i}-F_i|.
\]
We lower bound this by subtracting the paths damaged by faults:
\[
|\pp_{v_i}^{i,\ell_i}-F_i|
\ \ge\
|\pp_{v_i}^{i,\ell_i}|-\sum_{w\in F_i}|\pp_{v_i}^{i,\ell_i}[\{w\}]|.
\]
By definition of $\ell_i$ we have $|\pp_{v_i}^{i,\ell_i}|\ge D^{\ell_i}$.
By \textbf{I4} (weak collective FT), for each $w\in F_i$,
$|\pp_{v_i}^{i,\ell_i}[\{w\}]|\le 3\ell_i D^{\ell_i-1}$. As $|F_i|\le f$, we obtain
\[
|\pp_{u_i}^{i+1,\ell_i+1}|-|\pp_{u_i}^{i,\ell_i+1}|
\ \ge\
D^{\ell_i}-3f\ell_i D^{\ell_i-1},
\]
which is exactly \textbf{I6}.

\paragraph{Subclaim C (Bounded number of boosts: \textbf{I7}).}
Fix a vertex $x$ and a radius $\ell\in[r]$. Each boost of $x$ at radius $\ell$ increases $|\pp_x^{\ell}|$ by at least
$D^{\ell-1}-3f(\ell-1)D^{\ell-2}$,
by Subclaim B (with $\ell_i=\ell-1$). By \textbf{I5}, we always have
$|\pp_x^{\ell}|\le \frac{D^{\ell+1}-1}{D-1}$.
Hence, if $x$ were boosted at radius $\ell$ more than
\[
\frac{D^{\ell+1}-1}{(D-1)\bigl(D^{\ell-1}-3f(\ell-1) D^{\ell-2}\bigr)}
\]
times, then $|\pp_x^{\ell}|$ would exceed its cap from \textbf{I5}, a contradiction.

We now use the relation of $D,f,s,r$ to conclude the following upper bound: 
\begin{multline*}
    \frac{D^{\ell+1}-1}{(D-1)(D^{\ell-1} - 3f(\ell-1) D^{\ell-2})} \leq D\frac{1}{1 - (3f(\ell-1) + 1)/D} \le\\\le D\frac{1}{1 - 2fs/D} \leq D(1+4fs/D) = D + 4fs,
\end{multline*}
where the second inequality holds for $\ell > 1$ since $f(\ell-1) \geq 1$ and $\ell-1 \leq r-1\leq s/2$ and for $\ell = 1$ since $2fs > 1$, and the third holds because $4fs/D \leq 1/s \leq 1$ by choice of $D,f,s,r$. 
This proves \textbf{I7}.

\paragraph{Subclaim D (Weak and strong collective FT bounds: \textbf{I4} and \textbf{I2}).}
As noted above, only $\pp_{u_i}^{\ell_i+1}$ changes, so we prove \textbf{I4} and \textbf{I2} for
$v=u_i$ and $\ell=\ell_i+1$.

Let $i_1<i_2<\cdots<i_p=i$ be all indices up to $i$ in which $u_i$ was boosted at radius $\ell_i+1$; equivalently,
$e_{i_j}=\{u_i,v_{i_j}\}$ and the update appended $(u_i,v_{i_j})\circ(\pp_{v_{i_j}}^{i_j,\ell_i}-F_{i_j})$.
Then by construction,
\[
\pp_{u_i}^{i+1,\ell_i+1}
=
\bigcup_{j=1}^p (u_i,v_{i_j})\circ(\pp_{v_{i_j}}^{i_j,\ell_i}-F_{i_j}).
\]
By Subclaim C, $p\le D+4sf$.

\smallskip
\noindent\underline{\textbf{I4} (the $|W|=1$ case).}
Fix $w\neq u_i$.
A path in $\pp_{u_i}^{i+1,\ell_i+1}[\{w\}]$ is of one of two types:
\begin{itemize}[leftmargin=*]
\item it comes from a boost along the edge $(u_i,w)$ (i.e.\ $v_{i_j}=w$), hence belongs to
$(u_i,w)\circ \pp_{w}^{i_j,\ell_i}$; or
\item it comes from a boost along $(u_i,v_{i_j})$ with $v_{i_j}\neq w$, hence belongs to
$(u_i,v_{i_j})\circ \pp_{v_{i_j}}^{i_j,\ell_i}[\{w\}]$.
\end{itemize}
Therefore,
\[
|\pp_{u_i}^{i+1,\ell_i+1}[\{w\}]|
\le
|\pp_{w}^{i,\ell_i}|+\sum_{j:\, v_{i_j}\neq w}|\pp_{v_{i_j}}^{i_j,\ell_i}[\{w\}]|.
\]
By the inductive hypotheses \textbf{I5} and \textbf{I4} (at level $\ell_i$),
\[
|\pp_{w}^{i,\ell_i}|
\le
\frac{D^{\ell_i+1}-1}{D-1},
\qquad
|\pp_{v_{i_j}}^{i_j,\ell_i}[\{w\}]|
\le
3\ell_i D^{\ell_i-1}.
\]
Thus
\[
|\pp_{u_i}^{i+1,\ell_i+1}[\{w\}]|
\le
\frac{D^{\ell_i+1}-1}{D-1}+3p\ell_i D^{\ell_i-1}.
\]
Now we expand :
\[
\frac{D^{\ell_i+1}-1}{D-1}
\le
D^{\ell_i}+2D^{\ell_i-1}
\quad(\text{since }D\ge 4),
\qquad
3p\ell_i D^{\ell_i-1}
\le
3\ell_i(D+4sf)D^{\ell_i-1}
=
3\ell_i D^{\ell_i}+12\ell_i s f D^{\ell_i-1}.
\]
Hence
\[
|\pp_{u_i}^{i+1,\ell_i+1}[\{w\}]|
\le
(3\ell_i+1)D^{\ell_i}+\bigl(2+12\ell_i sf\bigr)D^{\ell_i-1}
\le
(3\ell_i+1)D^{\ell_i}+8s^2fD^{\ell_i-1}
\le
3(\ell_i+1)D^{\ell_i},
\]
where the last step uses $4s^2f\le D$ (so $8s^2fD^{\ell_i-1}\le 2D^{\ell_i}$).
This is \textbf{I4} for $(u_i,\ell_i+1)$.

\smallskip
\noindent\underline{\textbf{I2} (general $W$).}
Let $W\neq\emptyset$ with $u_i\notin W$ and $|W|\le \ell_i+1$ (otherwise the link is empty). We also can assume now that $|W| \geq 2$ since the case $|W| = 1$ is already done.
Using again the union representation above and splitting by whether the first neighbor is in $W$,
\[
\pp_{u_i}^{i+1,\ell_i+1}[W]
\subseteq
\bigcup_{w\in W}(u_i,w)\circ \pp_{w}^{i,\ell_i}[W\setminus\{w\}]
\ \cup\
\bigcup_{j:\, v_{i_j}\notin W}(u_i,v_{i_j})\circ \pp_{v_{i_j}}^{i_j,\ell_i}[W].
\]
Taking sizes and applying the inductive hypothesis \textbf{I2} (at level $\ell_i$) gives
\begin{align*}
|\pp_{u_i}^{i+1,\ell_i+1}[W]|
&\le
\sum_{w\in W} |\pp_{w}^{i,\ell_i}[W\setminus\{w\}]|
\;+\;
\sum_{j:\, v_{i_j}\notin W} |\pp_{v_{i_j}}^{i_j,\ell_i}[W]|\\
&\le
|W|\cdot 3\ell_i\, s^{2(|W|-2)}D^{\ell_i-(|W|-1)}
\;+\;
p\cdot 3\ell_i\, s^{2(|W|-1)}D^{\ell_i-|W|}.
\end{align*}
Factor out $3\ell_i s^{2(|W|-1)}D^{\ell_i+1-|W|}$ to obtain
\[
|\pp_{u_i}^{i+1,\ell_i+1}[W]|
\le
3\ell_i s^{2(|W|-1)}D^{\ell_i+1-|W|}
\left(\frac{|W|}{s^2}+\frac{p}{D}\right).
\]
Now use $|W|\le r\le s$ and $p\le D+4sf\le D(1+\frac{1}{s})$ (since $D\ge 4s^2f$) to get
\[
\frac{|W|}{s^2}+\frac{p}{D}
\le
\frac{1}{s}+\left(1+\frac{1}{s}\right)
\le
1+\frac{2}{s}.
\]
Finally, since $\ell_i\ge 1$ (can assume this because $|W| \geq 2$) and $2/s\le 1/\ell_i$ for $\ell_i\le r-1\le s/2$, we conclude
\[
|\pp_{u_i}^{i+1,\ell_i+1}[W]|
\le
3\ell_i s^{2(|W|-1)}D^{\ell_i+1-|W|}
\left(1+\frac{1}{\ell_i}\right)
=
3(\ell_i+1)s^{2(|W|-1)}D^{\ell_i+1-|W|},
\]
which is exactly \textbf{I2} for $(u_i,\ell_i+1)$.

\paragraph{Subclaim E (Pairwise parks preserved: \textbf{I3}).}
Fix $w\in V(G)$ and $J\subseteq V(G)$. We must show that $\pp_{u_i,w}^{i+1,\ell_i+1}$ is a park w.r.t.\
$\lsco^{\ell_i+1}$. By construction,
\[
\pp_{u_i,w}^{i+1,\ell_i+1}
=
\pp_{u_i,w}^{i,\ell_i+1}
\ \cup\
\Bigl((u_i,v_i)\circ(\pp_{v_i,w}^{i,\ell_i}-F_i)\Bigr).
\]
We have two cases:

\smallskip
\noindent\underline{Case 1: $(\pp_{v_i,w}^{i,\ell_i}-F_i)[J\setminus\{u_i\}]=\emptyset$.}
Then the added part contributes zero score to the $J$-link, hence
\[
\lsco_J^{\ell_i+1}(\pp_{u_i,w}^{i+1,\ell_i+1})
\le
\lsco_J^{\ell_i+1}(\pp_{u_i,w}^{i,\ell_i+1})
\le
1,
\]
by the inductive hypothesis \textbf{I3}.

\smallskip
\noindent\underline{Case 2: $(\pp_{v_i,w}^{i,\ell_i}-F_i)[J\setminus\{u_i\}]\neq\emptyset$.}
Then \Cref{lem:ft-neighborhood-exchange-lemma-gen} implies that $\pp_{u_i,w}^{i,\ell_i+1}$ is \emph{not} $J$-full, so
\[
\lsco_J^{\ell_i+1}(\pp_{u_i,w}^{i,\ell_i+1})\ \le\ \tfrac12.
\]
For the newly added paths, we apply \Cref{cl:concat} summed over the collection
$\pp_{v_i,w}^{i,\ell_i}$, and use the scale relation $\beta_{\ell_i+1}=10^{-(\ell_i+1)}=\frac{1}{10}\beta_{\ell_i}$:
\begin{align*}
\lsco_J^{\ell_i+1}\bigl((u_i,v_i)\circ \pp_{v_i,w}^{i,\ell_i}\bigr)
&=
\frac{1}{10}\,\lsco_J^{\ell_i}\bigl((u_i,v_i)\circ \pp_{v_i,w}^{i,\ell_i}\bigr)\\
&\le
\frac{1}{10}\Bigl(\lsco_J^{\ell_i}(\pp_{v_i,w}^{i,\ell_i})+\lsco_{J\setminus\{u_i\}}^{\ell_i}(\pp_{v_i,w}^{i,\ell_i})\Bigr)\\
&\le
\frac{1}{10}(1+1)
\le
\frac12,
\end{align*}
where the penultimate inequality uses the inductive hypothesis \textbf{I3} (parks at level $\ell_i$).
Therefore,
\[
\lsco_J^{\ell_i+1}(\pp_{u_i,w}^{i+1,\ell_i+1})
\le
\lsco_J^{\ell_i+1}(\pp_{u_i,w}^{i,\ell_i+1})
+
\lsco_J^{\ell_i+1}\bigl((u_i,v_i)\circ \pp_{v_i,w}^{i,\ell_i}\bigr)
\le
\tfrac12+\tfrac12
=
1.
\]
This proves \textbf{I3}.

\medskip
All invariants are established for $\pp^{i+1}$, completing the induction.
\end{proof}

%% file: Sqrtk-Sec.tex
\section{The VFT greedy $ \sqrt{k} \to O(k)$ spanner}\label{sec:sqrtk-greedy}

In this section we complete the global analysis of the VFT $\sqrt{k}\to O(k)$ spanner and prove \Cref{thm:sqrtk-greedy}. To provide intuition for the argument, we first recall the recent analysis of the \emph{non-faulty} $\sqrt{k} \to O(k)$ spanner due to \cite{PopovaT26}.

\subsection{The $\sqrt{k} \to O(k)$ greedy spanner}\label{sec:greedy-sqrtk-no-faults}

We now describe the analysis of the greedy $\lfloor \sqrt{k} \rfloor \to O(k)$ spanner, as motivation for the main steps in the VFT $\lfloor \sqrt{k} \rfloor \to O(k)$ setting.

The following proposition summarizes a key part of the analysis of \cite{PopovaT26}, and can be extracted implicitly from their proof. In any case, we do not use any part of their analysis in our FT analysis.

\begin{proposition}[\cite{PopovaT26}]\label{prop:sqrt-to-ok}
    Let $k \geq 9$. Let $P_1, \ldots, P_m$ be the paths added by the greedy $\lfloor \sqrt{k} \rfloor \to 12k$ spanner. Then there exist $S \subseteq V$, $|S|=\tilde{O}(n^{1-\lceil \sqrt{k}\rceil/k})$, and a subcollection $\mathcal{L}=\{P_{i_1},\ldots, P_{i_b}\}$ of paths such that: 
    \begin{enumerate}
        \item $|\mathcal{L}|=O(\sqrt{k}n^{1+1/k})$.

        \item For any $P_i \not \in \mathcal{L}$, let $H_i$ be the graph right before $P_i=(x_0,\ldots,x_{\lfloor\sqrt{k}\rfloor})$ was added, then there exist $0\leq \ell <r \leq \lfloor\sqrt{k}\rfloor$\footnote{Unlike the majority of other sections, in this one $\ell$ and $r$ usually mean the leftmost and the rightmost vertices with a special property on a given path (and not lengths)} such that: 

        \begin{itemize}
            \item $\dist_{H_i}(x_0,x_\ell) \leq \ell\cdot (2\lceil\sqrt{k}\rceil-1)$ and $\dist_{H_i}(x_r, x_{\lfloor\sqrt{k}\rfloor}) \leq (\lfloor\sqrt{k}\rfloor-r)\cdot(2\lceil\sqrt{k}\rceil-1)$.

            \item There is a path $Q_\ell$ from $x_\ell$ to $s_\ell \in S$ of length $\lceil\sqrt{k}\rceil$ in $H_i$ and a path $Q_{r}$ from $x_r$ to $s_r \in S$ of length $\lceil\sqrt{k}\rceil$ in $H_i$.
        \end{itemize}
    \end{enumerate}
\end{proposition}

At a high level, the proposition says that the exceptional set $\mathcal{L}$ consists of the paths responsible for creating \emph{local} structure. The first item states that only $O(\sqrt{k}n^{1+1/k})$ paths may fail to admit the required local connection to $S$.
For every path outside $\mathcal{L}$, the second item guarantees two vertices $x_\ell,x_r$ along the path that already have stretch $2\lceil\sqrt{k}\rceil-1$ to the respective endpoints $x_0$ and $x_{\lfloor\sqrt{k}\rfloor}$. Moreover, both $x_\ell$ and $x_r$ are connected to $S$ by short paths of length $\lceil\sqrt{k}\rceil$, denoted $Q_\ell,Q_r$.

\begin{figure}
    \centering
    \includegraphics[width=0.8\linewidth]{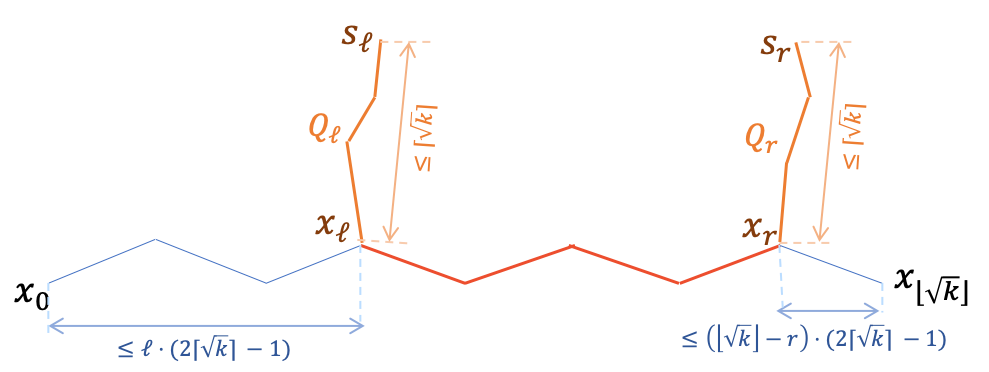}
    \caption{The structure of a path added by the greedy algorithm not in $\mathcal{L}$}
    \label{fig:placeholder}
\end{figure}

To finish the analysis of the $\lfloor\sqrt{k}\rfloor \to 12k$ spanner, one needs to bound the number of paths of the second type (not in $\mathcal{L}$) appearing in \Cref{prop:sqrt-to-ok}.

For this, we need the following observation.

\begin{observation}\label{obs:glob-greedy-change}
    Following the notations of \Cref{prop:sqrt-to-ok}, let $P_i$ be a path not in $\mathcal{L}$. Then 
    \begin{enumerate}
        \item(Big distance before adding $P_i$) $\dist_{H_i}(s_\ell,s_r) > 3\lceil \sqrt{k}\rceil(2\lceil \sqrt{k}\rceil - 1)$.
        \item(Small distance after adding $P_i$) $\dist_{H_{i+1}}(s_\ell,s_r)\leq 3 \lceil \sqrt{k}\rceil$.
    \end{enumerate}
\end{observation}

\begin{proof}
    Assume the contrary to the first item: $\dist_{H_i}(s_\ell,s_r) \leq 3\lceil \sqrt{k}\rceil(2\lceil \sqrt{k}\rceil - 1)$. Then, by the triangle inequality, we have that
    \begin{multline*}
        \dist_{H_i}(x_0,x_{\lfloor\sqrt{k}\rfloor}) \leq\\\leq \dist_{H_i}(x_0,x_{\ell}) + \dist_{H_i}(x_\ell,s_\ell) + \dist_{H_i}(s_\ell,s_r) + \dist_{H_i}(s_r,x_r) + \dist_{H_i}(x_r,x_{\lfloor\sqrt{k}\rfloor}) \leq \\ \leq \ell\cdot (2\lceil\sqrt{k}\rceil-1) + \lceil\sqrt{k}\rceil + 3\lceil \sqrt{k}\rceil(2\lceil \sqrt{k}\rceil - 1) +\lceil\sqrt{k}\rceil + (\lfloor\sqrt{k}\rfloor-r)(2\lceil\sqrt{k}\rceil-1) \leq 12k,
    \end{multline*}
    which contradicts the fact that the greedy algorithm adds $P_i$ to the spanner. Here we used the assumption together with the guarantees of \Cref{prop:sqrt-to-ok} to bound the distances in $H_i$, and the final inequality follows from $\sqrt{k}\geq 3$.

    For the second item it is enough to apply the triangle inequality: $$\dist_{H_{i+1}}(s_\ell, s_r) \leq \dist_{H_{i+1}}(s_\ell, x_\ell) + \dist_{H_{i+1}}(x_\ell, x_r) + \dist_{H_{i+1}}(x_r, s_r) \leq \lceil\sqrt{k}\rceil + L(P_i) + \lceil\sqrt{k}\rceil \leq 3\lceil\sqrt{k}\rceil.$$

\end{proof}

The observation above suffices to conclude the proof as follows: define a virtual graph $\mathbf{V}=(S,\emptyset)$. Whenever a path $P_i \not\in \mathcal{L}$ is added by the greedy spanner add to $\mathbf{V}$ the edge $(s_\ell,s_r)$ (where $s_\ell,s_r$ are the centers of the path guaranteed by \Cref{prop:sqrt-to-ok}).

\begin{claim}\label{clm:sqrt-k-to-ok}
    Each virtual edge addition can be coupled with an execution of the greedy $(2 \lceil\sqrt{k}\rceil - 1)$-spanner algorithm on a graph with vertex set $S$.
    As a consequence, the number of paths added by the $\lfloor\sqrt{k}\rfloor \to 12k$ greedy algorithm that are not in $\mathcal{L}$ is:
    \[O(|S|^{1+1/\lceil\sqrt{ k}\rceil})=\tilde{O}(n).\]
\end{claim}
\begin{proof}
    It is enough to show that for every edge added in the process of construction of the virtual graph defined above the distance between its ends \emph{in $\mathbf{V}$} is greater than $2 \lceil\sqrt{k}\rceil - 1$. From the second part of \Cref{obs:glob-greedy-change}, it follows that if two centers $s_1$ and $s_2$ are adjacent in the virtual graph before the $i$-th step of the $\lfloor\sqrt{k}\rfloor \to 12k$ greedy algorithm, then $\dist_{H_i}(s_1, s_2) \leq 3 \lceil \sqrt{k}\rceil$. Consequently, by the triangle inequality, if $s_\ell$ and $s_r$ are connected in the virtual graph by a path of length $\leq 2 \lceil\sqrt{k}\rceil - 1$, then $\dist_{H_i}(s_\ell, s_r) \leq 3 \lceil \sqrt{k}\rceil(2 \lceil\sqrt{k}\rceil - 1)$. But the latter contradicts the first part of \Cref{obs:glob-greedy-change}.  
\end{proof}

Consequently, this concludes the analysis: \Cref{prop:sqrt-to-ok} bounds the number of paths in $\mathcal{L}$ by $O(\sqrt{k}n^{1+1/k})$, and \Cref{clm:sqrt-k-to-ok} bounds the number of remaining paths added by the greedy $\sqrt{k}\to O(k)$ spanner by $\tilde{O}(n)$.

\subsection{The Faulty Setting}\label{sec:VFT-Faults-Describing-Virt}

We now prove: 

\begin{theorem}\label{thm:main}
    For every integer $k$, fault parameter $f$, and $0<\delta<\frac{1}{2}$, the spanner produced by the $\lfloor\sqrt{k}\rfloor \to O(k/\delta)$ VFT algorithm has size $\tilde{O}_{\delta}(f^{1+\delta}n^{1+1/k})$.
\end{theorem}

\begin{remark}
    We assume $k\geq 4$, as otherwise $\lfloor\sqrt{k}\rfloor=1$, hence we may just use the optimal multiplicative spanners of \cite{BodwinP19}.
\end{remark}
We begin with the following proposition, which follows from \Cref{prop:vft-d-to-r-structure} proved in \Cref{sec:greedyFTclustering}. We denote $s:= 2\lceil\frac{\sqrt{k}}{\delta}\rceil - 1$, $R := \lceil\frac{\sqrt{k}}{\delta}\rceil$ and let $C' \geq 2$ be some constant to be specified later.

\begin{proposition}\label{prop:sqrt-to-ok-vft}
    Let $P_1, \ldots, P_m$ be the paths added by the VFT greedy $\lfloor\sqrt{k}\rfloor \to 12k/\delta$ spanner. Let $F_i$ be the fault set causing $P_i$ to be added. Then there exist $S \subseteq V$, $|S|=\tilde{O}(f^{- \delta R} n^{1 - \frac{R}{k}})$, and a partition of $\{P_1, ..., P_m\}$ into two disjoint subfamilies $\mathcal{L}$ and $\mathcal G$ such that:
    \begin{enumerate}
        \item $|\mathcal L|=O(\delta^{-4}k^2f^{1+\delta}n^{1+1/k})$.

        \item For any $P_i \in \mathcal G$, let $H_i$ be the graph right before $P_i=(x_0,\ldots,x_{\lfloor \sqrt{k}\rfloor})$ was added, then there exist $0\leq \ell <r \leq \lfloor\sqrt{k}\rfloor$ such that: 
        \begin{itemize}
            \item $\dist_{H_i-F_i}(x_0,x_\ell) \leq \ell\cdot  (2\lceil \sqrt{k}/\delta\rceil-1)$, and $\dist_{H_i-F_i}(x_r, x_{\lfloor \sqrt{k}\rfloor}) \leq (\lfloor\sqrt{k}\rfloor-r)(2\lceil \sqrt{k}/\delta\rceil-1)$.

            \item There is a path collection $\prepp_{x_{\ell}}$ of paths from  $x_\ell$ to $S$ of length $R$ in $H_i$, satisfying $\forall U \subseteq V-\{x_{\ell}\}, |\prepp_{x_{\ell}}[U]|\leq |\prepp_{x_{\ell}}|(C' s f)^{-|U|}$. Similarly, there is a path collection $\prepp_{x_r}$ with the same properties.
        \end{itemize}
    \end{enumerate}
\end{proposition}

\begin{proof}
    We substitute $s = 2\lceil \sqrt{k}/\delta\rceil-1$ and $D = 12C's^3f^{1+\delta}n^{\frac{1}{k}}$ into \Cref{prop:vft-d-to-r-structure}.
\end{proof}
Armed with \Cref{prop:sqrt-to-ok-vft}, our task is to bound the number of paths \emph{not in} $\mathcal L$, in other words, paths in $\mathcal{G}$. 
We shall need the following notation:
\begin{definition}
    For a collection of oriented paths $\pp$, we denote by $S(\pp)$ the set of last vertices of the paths in $\pp$. For instance, in the notation of the previous proposition, $S(\prepp_{x_\ell}), S(\prepp_{x_r}) \subseteq S$.
\end{definition}

We start with the following analog of \Cref{obs:glob-greedy-change}:

\begin{observation}\label{obs:virtual-fault-proc}
    Using the notation of \Cref{prop:sqrt-to-ok-vft}, for any $P_i \in \mathcal{G}$ and any $s_\ell\in S(\prepp_{x_\ell}-F_i)$, $s_r\in S(\prepp_{x_r} - F_i)$, the following hold: 
    \begin{enumerate}
        \item $\dist_{H_{i}-F_i}(s_\ell,s_r) > \frac{8k}{\delta}$.

        \item $\dist_{H_{i+1}-F_i}(s_\ell,s_r)\leq 2 \lceil\sqrt{k}/\delta \rceil + \lfloor \sqrt{k}\rfloor \leq 3\frac{\sqrt{k}}{\delta}$. 
    \end{enumerate}
\end{observation}

\begin{proof}
    Let $Q_\ell \in \prepp_{x_\ell}-F_i$ and $Q_r \in \prepp_{x_r}-F_i$ be paths (of length $R$) that end in centers $s_\ell \in S$ and $s_r \in S$ respectively. To show the first bullet assume the contrary: $\dist_{H_{i}-F_i}(s_\ell,s_r) \leq \frac{8k}{\delta}$. Then, by the triangle inequality we get: 
    \begin{multline*}
        \dist_{H_i - F_i}(x_0, x_{\lfloor \sqrt{k}\rfloor}) \leq\\\leq  \dist_{H_i - F_i}(x_0, x_\ell) + L(Q_\ell) + \dist_{H_{i}-F_i}(s_\ell,s_r) + L(Q_r) + \dist_{H_i - F_i}(x_r, x_{\lfloor \sqrt{k}\rfloor}) \leq 12\frac{k}{\delta},
    \end{multline*}
    contradiction with $F_i$ being the faults that cause the addition of $P_i$. 

    The second bullet follows from the triangle inequality: 
    $$\dist_{H_{i+1}-F_i}(s_\ell,s_r)\leq L(Q_\ell) + (r - \ell) + L(Q_r).$$
\end{proof}

\paragraph{The Virtual Process}

In \Cref{sec:greedy-sqrtk-no-faults}, the paths not in $\mathcal{L}$ were bounded by coupling paths added by the (non-faulty) greedy, $\sqrt{k}\to O(k)$ algorithm with the standard greedy spanner. Concretely, \Cref{clm:sqrt-k-to-ok} couples each addition of a path $P_i \notin \mathcal{L}$ of the $\sqrt{k}\to O(k)$ greedy algorithm, with an edge added in an execution of the multiplicative greedy spanner, on a virtual graph. 

Our virtual process would have a similar spirit, but we can't simply couple the process with the classical greedy spanner - because we need to take faults into account.
We will have a virtual graph $\mathbf{V}$ with vertex set $S$, which is updated with each path in $\mathcal{G}$. Let $\mathbf{V}_i$ denote the state of the virtual graph immediately before $P_i$ is added to the spanner. Edges in $\mathbf{V}_i$ correspond to paths in $H_i$ of length at most $3 \frac{\sqrt{k}}{\delta}$.

To model \Cref{obs:virtual-fault-proc}, we need to know which virtual edges in $\mathbf{V}_i$ are faulty (that is, whose represented path intersects $F_i$) and which are not. This is crucial, since only edges in $\mathbf{V}_i$ for which \emph{all} vertices on the represented path in $H_i$ are \emph{non-faulty} can be used to form non-faulty low-stretch paths between centers.

To model this, we associate with every edge $e$ in $\mathbf{V}$ a list $L(e)$ of the vertices along the path it represents. Whenever $L(e) \cap F = \emptyset$, we say $e$ is non-faulty under the fault set $F$, and otherwise, $e$ is faulty. We now define the virtual graph $\mathbf{V}$ and its update rule.

Initially, $\mathbf{V}_1=(S,\emptyset)$. If $P_i \in \mathcal{L}$, we set $\mathbf{V}_{i+1}=\mathbf{V}_i$. If $P_i \in \mathcal{G}$, we define a \emph{multigraph} $B_i$ on vertex set $S$ as follows: for every path $Q_{\ell} \in \prepp_{x_\ell}-F_i$ ending at $s_{\ell}$ and every path $Q_{r} \in \prepp_{x_r}-F_i$ ending at $s_{r}$, we add a virtual edge $e$ with endpoints $s_\ell$ and $s_r$. This edge is assigned the list
\[
L(e) = V(Q_{\ell}) \cup V(P'_i) \cup V(Q_{r}),
\]
where $P'_i$ is the subpath of $P_i$ from $x_{\ell}$ to $x_r$. We then set $\mathbf{V}_{i+1}=\mathbf{V}_i \cup B_i$. By $\mathbf{V}_i-F$ we denote the graph on vertex set $S$ whose edge set consists of all non-faulty edges, namely those satisfying $L(e)\cap F=\emptyset$. See \Cref{fig:to_clft} for an illustration of the construction of $B_i$.

\begin{figure}[H]
    \centering
    \includegraphics[width=0.9\linewidth]{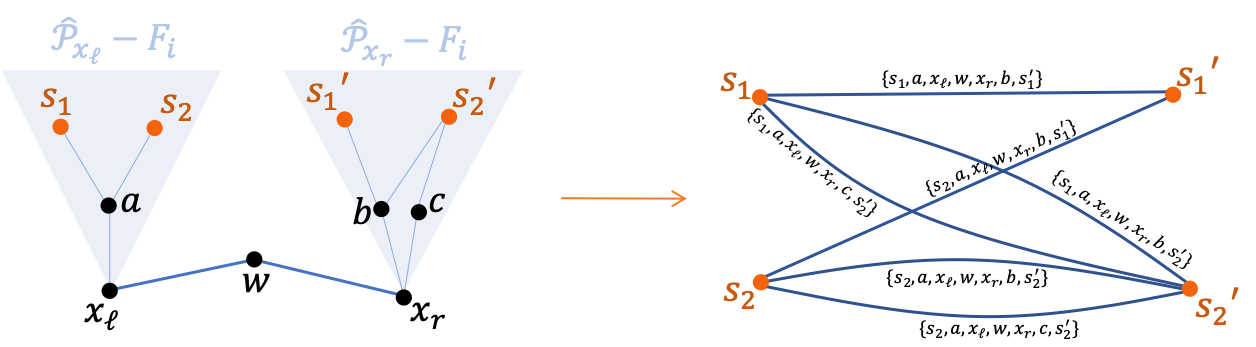}
    \caption{Obtaining $B_i$ from the path in $\mathcal{G}$}
    \label{fig:to_clft}
\end{figure}
\noindent We have the following lemma which summarizes the main properties of the construction:

\begin{lemma}\label{lem:properties-of-virt}
    Let $\mathbf{V}$ be as described above. Fix $i$ such that $P_i \in \mathcal{G}$:
    \begin{enumerate}
        \item The multiset of lists $\{L(e)\}_{e\in B_i}$ is in bijection with the following path collection:
        \[(\prepp_{x_\ell}-F_i)^{rev} \circ P'_i \circ (\prepp_{x_r}-F_i).\]
        \item For every edge $e$, we have $L(e) \cap F_i = \emptyset$.
        \item For every edge $e\in B_i$, there exists a path of length at most $3 \frac{\sqrt{k}}{\delta}$ corresponding to $e$, and this path is contained in $H_{i+1}$, hence in all $H_{j}$ for $j>i$.
        \item For every edge $(s,s')=e \in B_i$, we have $\dist_{\mathbf{V}_i-F_i}(s,s')>2 \lceil\sqrt{k}\rceil -1$.
    \end{enumerate}
\end{lemma}

\begin{proof}
    The first claim follows from the definition of $B_i$. 
    Given (1), to prove (2) it suffices to show that every path in $(\prepp_{x_\ell}-F_i)^{rev} \circ P'_i \circ (\prepp_{x_r}-F_i)$ avoids $F_i$. For the subcollections $(\prepp_{x_\ell}-F_i)^{rev}$ and $(\prepp_{x_r}-F_i)$ this is immediate from the definition. Moreover, $V(P_i)\cap F_i=\emptyset$, since $P_i$ was added by the VFT greedy algorithm under the fault set $F_i$. As $P_i'$ is a subpath of $P_i$, it also avoids $F_i$. This proves (2).
    
    To prove (3), by \Cref{prop:sqrt-to-ok-vft} $\prepp_{x_{\ell}}$ consists of paths in $H_i$, and as $\prepp_{x_\ell}-F_i$ is a subset of $\prepp_{x_{\ell}}$, it is contained in $H_i$. Similarly $\prepp_{x_r}-F_i$ is contained in $H_i$. Finally, as $H_{i+1}:=H_i \cup P_i$, we have $P_{i}'$ contained in $H_{i+1}$. Hence all paths corresponding to edges of $B_i$ are concatenations of paths in $H_{i}$ and $H_{i+1}$, hence all paths corresponding to edges of $B_i$ lie in $H_{i+1}$. These paths have length at most $3\sqrt{k}/\delta$ by construction and by \Cref{obs:virtual-fault-proc}.

    To prove the final claim, assume that $(s,s')=e \in B_i$ satisfies $\dist_{\mathbf{V}_i-F_i}(s,s')\leq 2 \lceil\sqrt{k}\rceil -1$. Then there exists a path of virtual edges in $\mathbf{V}_i-F_i$ from $s$ to $s'$ of length $h\leq 2 \lceil\sqrt{k}\rceil -1$. Let $e_1,\ldots,e_h$ be the edges of this path, in order. By definition each $e_j \in \{e_1,\ldots,e_h\}$ was added in some $B_{i'}$ with $i'<i$, hence by (3) each such edge $e_j$ corresponds to a path of length $\leq 3 \frac{\sqrt{k} }{\delta}$ in $H_{i}$, call this path $\Gamma_j$. By assumption $e_1,\ldots,e_h$ are all in $\mathbf{V}_i-F_i$ hence non-faulty, meaning $\Gamma_j$ is a non-faulty path in $H_i$. Concatenating $\Gamma_1,\ldots,\Gamma_h$ gives a path in $H_{i}-F_i$ from $s$ to $s'$ of length $\leq h \cdot 3 \frac{\sqrt{k} }{\delta} \leq \left(2 \lceil\sqrt{k}\rceil -1\right) \cdot 3 \frac{\sqrt{k} }{\delta} \leq \frac{8k}{\delta}$ (since $\sqrt{k} \geq 2$), contradicting \Cref{obs:virtual-fault-proc} (1).
\end{proof}

To bound $\mathcal{G}$, we would like to bound the virtual process defined above. As stated so far, however, this virtual process may add an unbounded number of graphs $B_i$. Indeed, the discussion up to this point is \emph{independent of the structure} of the path collections $\prepp_{x_{\ell}}-F_i$ and $\prepp_{x_{r}}-F_i$ underlying the virtual graph. In particular, if for example $\prepp_{x_{\ell}}-F_i$ or $\prepp_{x_{r}}-F_i$ is empty, then so is $B_i$, and the virtual process may continue forever, yielding no effective bound on $\mathcal{G}$.

To remedy this, we need to encode the \emph{actual} structure of the graphs $B_i$ obtained from \Cref{prop:sqrt-to-ok-vft}. We claim not only that both $\prepp_{x_{\ell}}-F_i$ and $\prepp_{x_{r}}-F_i$ are non-empty, but also that one may view $B_i$ as creating \emph{many robust connections} between centers all at once. Once we incorporate this extra structure into the definition of each $B_i$, we can effectively bound $\mathcal{G}$.

The rest of the section is organized as follows. In \Cref{sec:robust-sunflower}, we show that the extra property satisfied by our $B_i$ is that their lists form a \emph{robust sunflower}. In \Cref{sec:CLFT}, we give tight bounds on the virtual process under this assumption. Finally, in \Cref{sec:concluding-sqrt-ok}, we use the bounds obtained in \Cref{sec:CLFT} to conclude the proof of \Cref{thm:main}.

We now turn to quantify the structure of connection between centers each $P_i \in \mathcal{G}$ introduces. 
    
\subsection{Modeling the Center Connections Formed by $\mathcal{G}$}\label{sec:robust-sunflower}

We now describe how paths in $\mathcal{G}$ connect different centers. For a hypergraph $\mathcal{H}=(\mathcal{U},E)$\footnote{We allow $\mathcal{H}$ to have multiple edges} and $J\subseteq \mathcal{U}$, define the \emph{$J$-link} by
\[
\mathcal{H}[J] := \{e\in E(\mathcal{H}) \mid J\subseteq e\}.
\]

We will need:

\begin{definition}[$q$-spread hypergraph]
A hypergraph $\mathcal{H}=(\mathcal{U},E)$ is \emph{$q$-spread} if for every $J\subseteq \mathcal{U}$,
\[
|\mathcal{H}[J]| \ \le\ q^{-|J|}\,|\mathcal{H}|.
\]
\end{definition}

We will use the celebrated spread lemma, originally by \cite{LovettSunflowers}. The following formulation follows from \cite{raoSunflowers}, Lemma 4, by setting $\varepsilon=1/4$, which suffices for our needs.
\begin{lemma}[Spread Lemma]\label{lem-spread-lemma}
Let $\mathcal{H}$ be a hypergraph on a vertex set $U$ such that all $e\in E(\mathcal{H})$ satisfy $|e|\le w$
\footnote{Usually $\mathcal{H}$ is assumed to be $w$-uniform but by adding dummy elements to each set the same conclusion applies to $\leq w$-uniform hypergraphs.} , and assume $\mathcal{H}$ is $q$-spread.
There exists a universal $C>0$ such that if $U'$ is a random subset of $U$ in which each element is
present with probability $p= \frac{C\log(w+2)}{q}$, then
\[
\Pr\bigl[\exists\, e\in \mathcal{H}\text{ such that } e\subseteq U'\bigr]\ \ge\ \frac{3}{4}.
\]
\end{lemma}

We also need the following.
\begin{observation}\label{obs:spread-delete}
Let $\mathcal{H}$ be $q$-spread and let $A\subseteq \mathcal{U}$ with $q>|A|$.
Then the subhypergraph
\[
\mathcal{H}-A := \{e\in \mathcal{H} \mid e\cap A=\emptyset\}
\]
is $(q-|A|)$-spread.
\end{observation}

\begin{proof}
Fix nonempty $J\subseteq \mathcal{U}$. Set $q':=q-|A|>0$. Then
\[
\frac{|(\mathcal{H}-A)[J]|}{|\mathcal{H}-A|}
\le
\frac{|\mathcal{H}[J]|}{|\mathcal{H}-A|}
\le
\frac{|\mathcal{H}[J]|}{|\mathcal{H}|-\sum_{a\in A}|\mathcal{H}[\{a\}]|}
\le
\frac{|\mathcal{H}[J]|}{|\mathcal{H}|-\frac{|A|}{q}|\mathcal{H}|}
\le
q^{-|J|}\cdot\frac{1}{1-|A|/q}
\le
(q')^{-|J|}.
\]
Hence $\mathcal{H}-A$ is $q'$-spread.
\end{proof}
Finally we recall the following definition:

\begin{definition}[Robust sunflower, \cite{Rossman14}]
    Let $\mathcal F$ be a family of sets over a universe $U$. Let $K = \bigcap_{S \in \mathcal F} S$ be their common intersection. We say that $\mathcal F$ is a $(p, \varepsilon)$-robust sunflower (with kernel $K$) if 

    $$\Pr_{W \sim Bin(U, p)}[\exists S \in \mathcal F , S \subseteq K \cup W]\geq 1 - \varepsilon.$$
\end{definition}

\begin{lemma}\label{lem:path-collections-form-sunflower}
    For $P_i \in \mathcal{G}$, let $x_{\ell},x_r$ be as in \Cref{prop:sqrt-to-ok-vft}, and let $P'_i$ be the sub-path of $P_i$ from $x_{\ell}$ to $x_r$. Assume $C'\geq \max(2,10C)$ where $C$ is the constant guaranteed in \Cref{lem-spread-lemma}. Then:
    
      \[\left(V, (\prepp_{x_\ell}-F_i)^{rev} \circ P'_i \circ (\prepp_{x_r}-F_i)\right)\] is an $(1/2f,1/2)$-robust sunflower with kernel $V(P'_i)$\footnote{The hypergraph may contain parallel edges.}.
\end{lemma}
\begin{proof}
    We denote $\overline{\mathcal{P}}_{x_{\ell}}=\{V(Q_\ell)\setminus x_{\ell}\}_{Q_\ell \in \prepp_{x_\ell}}$, and $\overline{\mathcal{P}}_{x_{\ell}}-F_i=\{V(Q_\ell)\setminus x_{\ell}\}_{Q_\ell \in \prepp_{x_\ell}-F_i}$. Similarly define $\overline{\mathcal{P}}_{x_{r}}$, and $\overline{\mathcal{P}}_{x_{r}}-F_i$.

    By, \Cref{prop:sqrt-to-ok-vft}  $\overline{\mathcal{P}}_{x_{\ell}}$ and $\overline{\mathcal{P}}_{x_{r}}$ are $C'sf$-spread. By applying \Cref{obs:spread-delete} with $|F_i| \leq f$, we have the $\overline{\mathcal{P}}_{x_\ell}-F_i$ is $(C'sf-f)\geq \frac{1}{2}C'sf$-spread as $C' \geq 2$. Similarly  $\overline{\mathcal{P}}_{x_r}-F_i$ is $\frac{1}{2}C'sf$-spread.

    \noindent We now apply the spread lemma\footnote{$\overline{\mathcal{P}}_{x_\ell}-F_i$ is $q = \frac{C'sf}{2}$-spread and $\frac{C\log(w+2)}{q} \leq  \frac{2\log(\sqrt{k}+3)}{5f\sqrt{k}/\delta} \leq \frac{1}{2f}$ for $\sqrt{k} \ge 2$ and $\delta \leq \frac12$.
    },
    
    \Cref{lem-spread-lemma} to get $$\Pr_{W\sim Bin(V(G), \frac{1}{2f})}[\exists Q'_\ell \in \overline{\mathcal{P}}_{x_{\ell}}-F_i : V(Q'_\ell)\subseteq W] \geq \frac{3}{4},$$ and similarly,  $$\Pr_{W\sim Bin(V(G), \frac{1}{2f})}[\exists Q'_r \in \overline{\mathcal{P}}_{x_{r}}-F_i : V(Q'_r)\subseteq W] \geq \frac{3}{4},$$
    Hence, by a union bound on the complements we have:
    \[\Pr_{W\sim Bin(V(G), \frac{1}{2f})}[\exists Q'_\ell: V(Q'_\ell) \subseteq W  \text{ and  } \exists Q'_r : V(Q'_r)\subseteq W]\geq \frac{1}{2}\]
    But letting $Q_\ell$ be the unique path in $\prepp_{x_\ell}-F_i$ extending $Q_{\ell}'$ and similarly, $Q_r$ we get that:
    \[\Pr_{W\sim Bin(V(G), \frac{1}{2f})}[\exists Q_\ell\in \prepp_{x_\ell}-F_i, Q_r\in \prepp_{x_r}-F_i: V\left((Q_\ell)^{rev}\circ P_i' \circ Q_r \right) \subseteq V(P_i')\cup W]\geq \frac{1}{2}\]
    As needed\footnote{Notice that the above analysis show only that the kernel contains $V(P_i')$, but it is easy to see that it is exactly $V(P_i')$. Indeed, recall that $\overline{\mathcal{P}}_{x_{r}}-F_i$ and $\overline{\mathcal{P}}_{x_{\ell}}-F_i$ are each $q = \frac{C'sf}{2}$-spread, and by our choice of parameters $q>1$. It is easy to see that any $q$-spread hypergraph with $q>1$ has empty intersection, hence for any vertex $v \in V - V(P_i')$, we can find paths $Q_r \in \overline{\mathcal{P}}_{x_{r}}-F_i$, and $Q_\ell \in \overline{\mathcal{P}}_{x_{\ell}}-F_i$ not containing $v$, meaning $v$ is not in the kernel of the concatenated hypergraph from the lemma. Consequently, the kernel is precisely $V(P_i')$.}.    
\end{proof}

We now incorporate the robust sunflower conditions in the analysis of the previously mentioned virtual process.

\subsection{The Robust CLFT Process}\label{sec:CLFT}

We now describe an abstract process capturing the virtual process introduced in \Cref{sec:VFT-Faults-Describing-Virt}. Let $G$ be a graph, let $\mathcal{C}$ be a set of colors.  We say that $G$ is \emph{list colored} if every edge $e$ is assigned a list $L(e) \subseteq \mathcal{C}$. In the CLFT model introduced in \cite{PetruschkaST24_ITCS}, a set $F \subseteq \mathcal{C}$ of at most $f$ colors fails. We denote by $G-F$ the graph with vertex set $V(G)$ and edge set
\[
\{ e \mid L(e) \cap F = \emptyset \}.
\]

\noindent An \emph{$f$-CLFT $(2k-1)$-spanner} $H$ of $G$ is a subgraph $H$ such that
\[
\forall F \subseteq \mathcal{C},\ |F|\leq f,\quad \dist_{H-F}(u,v) \leq (2k-1)\dist_{G-F}(u,v).
\]

In \cite{PetruschkaST24_ITCS} it was proved that every list colored graph $G$, with all lists of size $\leq w$ admits an $f$-CLFT $(2k-1)$-spanner $H$ of size $O(f^{w}n^{1+\frac{1}{k}})$, which they proved is tight assuming the girth conjecture. The goal of this section is a robust, parallel version of this result.

We study an iterative procedure that maintains a growing subgraph $H_i$ of $G$. We start with $H_1=(V,\emptyset)$ on $n$ vertices and two integers $w$ and $t$. 
At each step $i\geq 1$, we are given the current graph $H_i$, and a list-colored multi-graph $B_i$; the rule decides whether to add all edges of $B_i$ to $H_i$.

\paragraph{Inputs at step $i$.}
At step $i\ge 1$ we are given:
\begin{enumerate}
    \item The current graph $H_i$.

    \item A \emph{list colored} multi-graph $B_i=(V,E_i)$. 

    \item\label{robustness} (Robustness) The collection $\{L(e)\}_{e \in E_i}$ is a $(\frac{1}{2f}, \frac{1}{2})$-robust sunflower, whose kernel $L_i:=\cap_{e\in B_i} L(e)$ is of size $\leq w$.
    
\end{enumerate}

\paragraph{Decision rule.}
If there exists a set $F_i\subseteq \mathcal{C}$ with $F_i \cap \bigcup_{e\in E_i}L(e)=\emptyset$ and $|F_i|\leq f$ such that for every $\{u,v\}=e \in E_i$ we have $\dist_{H_{i}-F_i}(u,v)>t$, then set $H_{i+1}=H_i\cup B_i$. Otherwise, set $H_{i+1}=H_i$.

\paragraph{Interpretation.}
Think of each edge in $B_i$ as a virtual edge between centers in $S$ that are being added by a path addition in the greedy algorithm. By \Cref{lem:properties-of-virt} the graphs $B_i$ described in \Cref{sec:VFT-Faults-Describing-Virt}, will all be added by the CLFT process with $w = \lfloor\sqrt{k}\rfloor +1$ and $t = 2\lceil \sqrt{k}\rceil - 1$. For paths in $\mathcal{G}$, and $B_i$ as we defined all decisions would be to add $B_i$, yet we find it natural to model the above process with a decision rule.  

\paragraph{Main bound.}
Our main result is a tight upper bound on the number of steps in which $B_i$ is added.

\begin{theorem}\label{thm:robust-clft-process}
    The number of steps in which the robust CLFT process added $B_i$ is
    \[
    O\!\left(w f^{w}  n^{1+1/\lceil t/2 \rceil }\right).
    \]
\end{theorem}

\paragraph{Relation to the color list failure.}
Our upper bound on the number of added steps matches the size upper bound for spanners against color-list faults (up to a multiplicative $w$ factor) and significantly generalizes that setting. In \cite{PetruschkaST24_ITCS}, the authors analyze a CLFT greedy spanner with stretch $t$ that considers iteratively a graph $B_i$ consisting of a \emph{single} edge $e_i$, with a single list of $\leq w$ colors $L_i$, with $B_i$ added whenever there exists an $F_i$ of size at most $f$ with $F_i \cap L_i=\emptyset$ such that the endpoints of $e_i$ are at distance $>t$ in $H_i-F_i$. This case follows from our proof, as a single list $L_i$ is a robust sunflower (with itself as the kernel). 
In this special case, our bound recovers the same size upper bound for the color-list fault-tolerant greedy spanners of \cite{PetruschkaST24_ITCS} up to a $w$ factor which is negligible. We also highlight that \cite{PetruschkaST24_ITCS} provided a matching lower bound for the CLFT greedy algorithm, which in particular implies that our upper bound is essentially tight.

\subsection*{Proof of \Cref{thm:robust-clft-process}}
We restrict attention to the steps in which the process adds $B_i$. Namely, assume that in steps $i=1,2,\ldots,m$ we decided to keep $B_i$ and add it to $H_i$. Our goal is to prove
\[
m= O\!\left(w f^w \cdot n^{1+1/\lceil t/2 \rceil}\right).
\]

For each such step $i$, by the decision rule there exists a fault set $F_i$ satisfying $\bigcup_{e\in E_i}L(e) \cap F_i=\emptyset$, $|F_i|\leq f$, and $\dist_{H_i-F_i}(u_i,v_i)>t$ for all $(u_i,v_i)\in E_i$.

We combine the blocking-set method with the spread lemma. Define the set of blocks
\[
\mathcal{B}=\{(L_i,c) \mid c \in F_i \}.
\]
Sample a random subset $\mathcal{C}' \subseteq \mathcal{C}$ by including each color independently with probability $\frac{\gamma}{f}$, where $\gamma=\frac{w}{w+1}$. Recall that for all $w\geq 1$, $\gamma^w\geq \frac{1}{3}$. We say that $L_i$ is \emph{blocked} if there exists a block $b=(L_i,c)$ with $L_i \cup \{c\} \subseteq \mathcal{C}'$. We say that the $i^{\text{th}}$ iteration is unblocked if $L_i$ is not blocked but survives, meaning that $L_i \subseteq \mathcal{C}'$. As $|L_i|\leq w$, and every $L_i$ participates in $\leq f$ blocks (corresponding to $F_i$), we have:

\begin{align*}
\Pr[\text{The } i^{\text{th}} \text{ iteration is unblocked}]
&\ge \left(\frac{\gamma}{f}\right)^{|L_i|}
   - \left(\frac{\gamma}{f}\right)^{|L_i|+1} f \\[4pt]
&\ge \left(\frac{\gamma}{f}\right)^w
   - \left(\frac{\gamma}{f}\right)^{w+1} f \\[4pt]
&= \gamma^w (1 - \gamma)\frac{1}{f^w} \\[4pt]
&\ge \frac{1}{3(w+1)} f^{-w}.
\end{align*}

We call the $i^{\text{th}}$ round \emph{successful} if (1) $L_i$ is unblocked and (2) there exists an edge $e\in E_i$ such that $L(e)\subseteq \mathcal{C}'$.

\begin{claim}\label{clm:clft-spread-round}
    \[
    \Pr\!\left[
      \text{The } i^{\text{th}} \text{ round is successful }
      \,\middle|\,
      \text{The } i^{\text{th}} \text{ iteration is unblocked}
    \right] \geq \frac{1}{2}.
    \]
\end{claim}

\begin{proof}
    The event that the $i^{\text{th}}$ iteration is unblocked is exactly $L_i \subseteq \mathcal{C}'$ and $F_i \cap \mathcal{C}'=\emptyset$. We condition on this event by first exposing membership in $\mathcal{C}'$ for colors in $L_i \cup F_i$, and then sampling the remaining colors independently. 
    Accordingly, write $\mathcal{C}'=L_i \cup \mathcal{C}''$, where $\mathcal{C}'' \subseteq \mathcal{C}-L_i-F_i$ is formed by including each color independently with probability $\frac{\gamma}{f}$.

    Under this conditioning, it suffices to lower bound the probability that some edge $e \in B_i$ satisfies $L(e)\subseteq \mathcal{C}'$.

    But since $F_i \cap \bigcup_{e\in E_i}L(e) = \emptyset$, we have 
    $$\Pr_{\mathcal{C}''\sim Bin(\mathcal{C}-L_i - F_i, \frac{\gamma}{f})}\big[\exists\, e:\ L(e)\subseteq L_i\cup \mathcal{C}''\big] = \Pr_{\mathcal{C}'\sim Bin(\mathcal{C}, \frac{\gamma}{f})}\big[\exists\, e:\ L(e)\subseteq \mathcal{C}'\cup L_i\big] \geq \frac{1}{2},$$
    where the last inequality follows because $\{L(e)\}_{e\in E_i}$ being a $(\frac{1}{2f}, \frac12)$-robust sunflower and $\gamma \geq \frac12$ for all $w \geq 1$. 

\end{proof}

By the law of total expectation and linearity,
\begin{align*}
\mathbb{E}[\text{Number of successful rounds}]
&\ge \sum_{i=1}^m 
   \Pr\big[
      \text{The } i^{\text{th}} \text{ round is successful} 
      \,\big|\, 
      \text{The } i^{\text{th}} \text{ iteration is unblocked}
   \big]  \\
&\hspace{2cm} \cdot\,
   \Pr\big[\text{The } i^{\text{th}} \text{ iteration is unblocked}\big] \\[4pt]
&\ge m \cdot \frac{1}{2} \cdot \frac{1}{3(w+1)f^w}.
\end{align*}

Fix a choice of $\mathcal{C}'$ for which the number of successful rounds is
\[
r \ge \frac{m}{6(w+1)f^w}.
\]
Let $e_1^*,\ldots,e_r^*$ be a choice of an edge from the successful rounds in order. By success, we have:
\begin{itemize}
    \item (Safe edge) $L(e_i^*) \subseteq \mathcal{C}'$.
    \item (Unblocked iteration) $F_i \cap \mathcal{C}'=\emptyset$.
\end{itemize}

Define $G_i=(V,\{e_1^*,\ldots,e_{i-1}^*\})$. Then $G_i \subseteq H_i-F_i$, since $G_i \subseteq H_i[\mathcal{C}']$\footnote{ $H_i[\mathcal{C}']$ denotes all edges in $H_i$, such that their color lists are contained in $\mathcal{C}'$.} and the $i^{\text{th}}$ iteration is unblocked. Then for every $e_i^*=\{u_i,v_i\}$ we have by the decision rule that
\[
\dist_{G_i}(u_i,v_i) \geq \dist_{H_i-F_i}(u_i,v_i) > t.
\]

Thus, when scanning the edges of $G_i$ in order using the greedy $1\to t$ spanner algorithm, each $e_i^*$ is added. Since the greedy $1\to t$ algorithm adds at most $O(n^{1+1/\lceil t/2 \rceil})$ edges, we have $r=O(n^{1+1/\lceil t/2 \rceil})$. Combining this with $r \ge \frac{m}{6(w+1)f^w}$ yields
\[
m = O\!\left( (w+1)f^w n^{1+1/\lceil t/2 \rceil}\right).
\]

\subsection{Proof of \Cref{thm:main}}\label{sec:concluding-sqrt-ok}

By \Cref{lem:properties-of-virt} and \Cref{lem:path-collections-form-sunflower} the virtual construction process of the graph $\mathbf{V}$ defined earlier in the section satisfies the definition of the robust CLFT process with $w = \lfloor \sqrt{k} \rfloor +1$. Moreover, specifically by \Cref{lem:properties-of-virt}(4) this CLFT process does add all the $B_i$ for $t = 2\lceil \sqrt{k}\rceil - 1$. Since $B_i$ is then added every time when $P_i \in \mathcal{G}$, \Cref{thm:robust-clft-process} implies that 
$$|\mathcal{G}| = O(\sqrt{k}f^{\lfloor\sqrt{k}\rfloor+1}(|S|)^{1 + \frac{1}{\lceil \sqrt{k}\rceil}}) = \tilde{O}(\sqrt{k}f^{\lfloor\sqrt{k}\rfloor+1}(f^{-\sqrt{k}}n^{1 - \frac{1}{\delta \sqrt{k}}})^{1 + \frac{1}{\lceil \sqrt{k}\rceil}}) = \tilde{O}(fn).\footnote{In the technical overview (\Cref{sec:sqrtk-technical-overview-greedy}) the bound obtained for $\mathcal{G}$ was $\tilde{O}(n)$ instead of the bound $\tilde{O}(fn)$ appearing here. This is due to integrality issues which were not covered there for simplicity of exposition, and are insignificant.}$$
Summing this with the bound on $|\mathcal{L}|$ from \Cref{prop:sqrt-to-ok-vft}, we get the desired bound.

%% file: Section2-2k.tex
\section{Analysis of Vertex Fault Tolerant $2 \to 2k$-spanners}\label{sec:vft2to2k}

In this section we prove:

\begin{theorem}\label{thm:2-2k-ft-greedy}
    The greedy VFT $2\to 2k$ algorithm outputs a spanner with $O(k^3f n^{1+\frac{1}{k}})$ edges.
\end{theorem}

Throughout this section we denote $R := \lceil\frac{k}{2}\rceil$ and define $I_{odd} = \mathbf{1}[k \text{ is odd}]$, so that $2R - I_{odd} = k$ for all $k$. Also, as in previous section, we denote $\pp^{rev}$ the collection of all paths in $\pp$ with directions reversed.

We apply \Cref{alg:greedy-path-system} with $s = k$ and some $D$ to be chosen later. This procedure provides us with the partition of paths into $\mathcal{L}$ and $\mathcal{G}$ and the sequence of path collections $\pp^i$ that satisfy invariants \textbf{I1-I7}. Denote, as usual, $P_i$ the $i^{\text{th}}$ path added by the greedy and let $P_i = (x_i,m_i,y_i)$. We can without loss of generality assume that $\{x_i, m_i\}$ is a $k$-critical edge.

Note also that we can (and will, from now on) assume that both edges $\{x_i, m_i\}$ and $\{m_i, y_i\}$ are not in the current spanner $H_i$. Indeed, consider all the added $2$-paths $P_{i_j}$ for which only one edge $e_{i_j}$ is actually new. Then it is easy to see that the edges $e_{i_j}$ can be sequentially added by the usual multiplicative greedy VFT spanner algorithm with the stretch $2k - 1$. See the analogous statement for the EFT case in \cite{PopovaT26} for details. Thus the number of such paths is $O(f^{1-\frac{1}{k}}n^{1+\frac{1}{k}})$, which is less than the desired bound.

We will use \Cref{thm:greedy-path-system-bounds} to bound $|\mathcal{L}|$, and so it remains to bound $|\mathcal{G}|$. If $P_i \in \mathcal G$ then by \Cref{cl:obvious-prop-g} we have $|\pp_{x_i}^{i,\ell}|, |\pp_{m_i}^{i,\ell}| \geq D^\ell$ for all $0\leq \ell \leq R$. We will sometimes say $i$ is a boosting step instead of $P_i\in \mathcal L$ and $i$ is a non-boosting step instead of $P_i\in \mathcal G$.

\subsection{The Analysis}

As the analysis in this section was not covered in \Cref{sec:greedy-tech-overview}, we include a short proof overview. The approach is to maintain extra path collections \emph{besides} $\pp$, to bound $\mathcal{G}$ - Just as $\pp$ is used to bound $\mathcal{L}$ in \Cref{thm:greedy-path-system-bounds}. 

\paragraph{High-level strategy: a charging argument.}
The analysis tracks the greedy algorithm's progress in the \emph{non-boosting steps} via three evolving path collections $\Gg^i$,$\Ggh^i$, and $\Ll^i$. The path collections $\Gg^i$,$\Ggh^i$, and $\Ll^i$, will have paths of length $k+1,k,R$, respectively. Initially  $\Gg^i$,$\Ggh^i$, and $\Ll^i$ are all empty. The crux of the proof is that at each non-boosting step, we \emph{charge} one of these collections by growing it by a \emph{substantial} amount. To complement this, we also show that the path collections we construct satisfy a "pairwise a park" condition, similar to those appearing in \Cref{sec:greedyFTclustering}, which will force them to be small, see e.g. \Cref{thm:st-mult-vft}, on how the park property forces path collections to be small. 

In particular, $\Gg$ and $\Ggh$ satisfy: for any $u,v$, the subcollection $\Gg_{u,v}$ (resp.\ $\Ggh_{u,v}$) is a park. This means that every pair of vertices has at most $O_k(f^{\ell-1})$ paths of length $\ell$ between them, where $\ell=k$ is relevant for $\Ggh$, and $\ell=k+1$ is relevant for $\Gg$. By applying \Cref{obs:size-bounds-park-links}, we get $|\Gg| \leq n^2 \cdot O_k(f^k)$. We will show that each time we charge $P_i \in \mathcal{G}$ for increasing $\Gg$ substantially, then $\Gg$ will increase by at least $\Omega(D^{k-1}) = \Omega_k(f^{k-1} n^{(k-1)/k})$, so $\Gg$ can be charged at most
\[
\frac{n^2 \cdot O_k(f^k)}{\Omega_k(f^{k-1} n^{(k-1)/k})} = O_k(f \cdot n^{1+1/k})
\]
times. Which is the desired bound. The same applies to $\Ggh$. 

It turns out that we may not always charge either $\Gg$ or $\Ggh$, hence we use $\Ll$ as well.

The collection $\Ll^i$ consists of $R$-long paths, and it plays a role analogous to the paths in $\pp$ constructed in \Cref{alg:greedy-path-system}. It is useful to think of $\mathcal{R}$ as having similar properties to $\pp$. The key difference is that $\mathcal{R}$ will grow at each step around the vertex $y_i$ \emph{that isn't necessarily adjacent to a $k$-critical edge}, hence for this reason, we need to devise a different way to augment the path collection which is more elaborate, in contrast to the simple way in \Cref{alg:greedy-path-system}.

\paragraph{The update rule.}

We initialize $\Gg^1 = \emptyset$, $\Ggh^1 = \emptyset$ and $\Ll^1 = \emptyset$.

If $P_i \in \mathcal L$ we do nothing. Otherwise, we behave differently depending on whether $|\Ll^i_{y_i}| \geq D^R$.

    \begin{itemize}
        \item \textbf{Case 1: $|\Ll^i_{y_i}| \geq D^R$ (charge $\Gg$).} We update $\Gg^{i+1} = \Gg^i\cup ((\pp^{i, R-I_{odd}-1}_{x_i} - F_i)^{rev}\circ (x_i, m_i, y_i)\circ (\Ll^i_{y_i} - F_i))$ (and $\Ll^{i+1} = \Ll^i$ and $\Ggh^{i+1} = \Ggh^i$).
        \item \textbf{Case 2: $|\Ll^i_{y_i}| < D^R$ (partition paths from $m_i$).} Take two path collections $\pp\mathcal{R}^i$ and $\pp\mathcal{K}^i$ initially empty. Take a path $P\in \pp_{m_i}^{i, R-1}-F
        _i$ that ends in a vertex $v$. If there exists some $J \subseteq V(P) \cup \{y_i\}$ such that the park $\Ll^i_{y_i, v}$ is $J$-full (w.r.t. $\lsco^R$), then put $P$ into $\pp\Ggh^i$, otherwise put it into $\pp\Ll^i$.
        \begin{itemize}
            \item \textbf{Case 2a: $|\pp\Ll^i| \geq \frac{1}{4}D^{R-1}$ (charge $\Ll$).} Update $\Ll^{i+1} = \Ll^i\cup ((y_i,m_i)\circ \pp\Ll^i)$ (and $\Gg^{i+1} = \Gg^i$ and $\Ggh^{i+1} = \Ggh^i$).
            \item \textbf{Case 2b: $|\pp\Ll^i| < \frac{1}{4}D^{R-1}$ (charge $\Ggh$).} Update $\Ggh^{i+1} = \Ggh^i\cup((\pp^{i, R-I_{odd}}_{x_i} - F_i)^{rev}\circ(x_i, m_i)\circ \pp\Ggh^i)$ (and $\Ll^{i+1} = \Ll^i$ and $\Gg^{i+1} = \Gg^i$).
        \end{itemize}
    \end{itemize}

We now state the main technical lemmas. 

\begin{lemma}[Structural properties]\label{lem:path-collections-structural}
    Suppose that $D\geq \max\{48kf, 4k^2f\}$. Then the three sequences of path collections $\Gg^i$, $\Ggh^i$ and $\Ll^i$ defined above satisfy:
    \begin{enumerate}
        \item $E(\Gg^i), E(\Ggh^i), E(\Ll^i)\subseteq E(H_i)$;
        \item $\Gg^{i} \subseteq \Gg^{i+1}$, $\Ggh^{i} \subseteq \Ggh^{i+1}$ and $\Ll^{i} \subseteq \Ll^{i+1}$;
        \item All paths in $\Gg^i$ have length $k+1$;
        \item All paths in $\Ggh^i$ have length $k$; 
        \item All paths in $\Ll^i$ have length $R$;
        \item If the $i$-th step is boosting, $\Gg^{i+1} = \Gg^i$, $\Ggh^{i+1} = \Ggh^i$ and $\Ll^{i+1} = \Ll^i$.
    \end{enumerate}
\end{lemma}

\begin{lemma}[Growth guarantee]\label{lem:path-collections-growth}
    Under the same hypotheses, if $P_i \in \mathcal G$, then \emph{at least one} of the following holds:
    \begin{enumerate}[label=(\alph*)]
        \item $|\Gg^{i+1}| - |\Gg^i| \geq \frac14 D^{k-1}$,
        \item $|\Ggh^{i+1}| - |\Ggh^i| \geq \frac18 D^{k-1}$,
        \item $|\Ll^{i+1}| - |\Ll^{i}| \geq \frac{1}{4}D^{R-1}$.
    \end{enumerate}
\end{lemma}

\begin{lemma}[$\Gg$ is a park]\label{lem:G-is-park}
    For any $u, v \in V(G)$, the collection $\Gg^i_{u,v}$ is a park w.r.t.\ $\lsco = \sco^{\beta}$ with $\beta = 10^{-(k-1)}2^{-(2k + 5)}$.
\end{lemma}

\begin{lemma}[$\Ggh$ is a park]\label{lem:Gtilde-is-park}
    For any $u, v \in V(G)$, the collection $\Ggh^i_{u,v}$ is a park w.r.t.\ $\hat\lsco = \sco^{\hat\beta}$ with $\hat\beta = 10^{-k+1}\cdot 2^{-2k-3}$. 
\end{lemma}

\begin{lemma}[Properties of $\Ll$]\label{lem:L-properties}
    The collection $\Ll^i$ satisfies:
    \begin{enumerate}
        \item $|\Ll^i_v| \leq 2D^{R}$ for any $v\in V(G)$;
        \item $|\Ll_v^i[u]| \leq 24R\cdot D^{R-1}$ for any $v \neq u \in V(G)$;
        \item $\Ll_{u,v}^i$ is a park w.r.t.\ $\lsco^R = \sco^{\beta_R}$, $\beta_R = 10^{-R}$.
    \end{enumerate}
\end{lemma}

We prove all lemmas together by induction, as properties stated in one lemma rely on one another. 

\begin{proof}[Proof of \Cref{lem:path-collections-structural,lem:path-collections-growth,lem:G-is-park,lem:Gtilde-is-park,lem:L-properties}]
    We verify all properties simultaneously by induction.

    As the collections are initially empty, the base is trivial. We now consider the inductive step.
    Properties (1)--(6) of \Cref{lem:path-collections-structural} are evident from the construction. We now verify the remaining properties case by case.

    \paragraph{Case 1: $\Ll^i_{y_i}$ is already big enough ($|\Ll^i_{y_i}| \geq D^{R}$).}
    In this case only $\Gg$ changes, so we must verify \Cref{lem:path-collections-growth}(a) and \Cref{lem:G-is-park}.

    \emph{Verifying the growth guarantee.} We bound the size of each surviving collection separately.
    For $\Ll^i_{y_i} - F_i$:
    \begin{align}
        |\Ll^{i}_{y_i} - F_i|
        &\geq |\Ll^i_{y_i}| - \sum_{v\in F_i}|\Ll^i_{y_i}[v]|
            && \text{(union bound on faulty paths)} \notag\\
        &\geq D^{R} - f \cdot 24R \cdot D^{R-1}
            && \text{(by \Cref{lem:L-properties}(b) with $|F_i| \le f$)} \notag\\
        &= D^{R}\left(1 - \frac{24Rf}{D}\right) \notag\\
        &\geq \frac{1}{2}D^R
            && \text{(since $D \geq 48Rf$).}
    \end{align}

    For $\pp^{i,R-I_{odd}-1}_{x_i} - F_i$, using that the step is non-boosting (and thus $|\pp^{i,R-I_{odd}-1}_{x_i}|\geq D^{R-I_{odd}-1}$) and weak collective fault tolerance:
    \begin{align}
        |\pp^{i,R-I_{odd}-1}_{x_i} - F_i|
        &\geq |\pp_{x_i}^{i,R-I_{odd}-1}| - \sum_{v\in F_i}|\pp_{x_i}^{i,R-I_{odd}-1}[v]|
            && \text{(union bound)} \notag\\
        &\geq D^{R-I_{odd}-1} - f \cdot 3(R-I_{odd}-1) \cdot D^{R-I_{odd}-2}
            && \text{(weak collective FT)} \notag\\
        &\geq \frac{1}{2}D^{R-I_{odd}-1}
            && \text{(since $D \geq 6Rf$).}
    \end{align}

    Multiplying these bounds:
    \[
        |\Gg^{i+1}| - |\Gg^i|
        = |\pp^{i,R-I_{odd}-1}_{x_i} - F_i| \cdot |\Ll^{i}_{y_i} - F_i|
        \geq \frac{1}{2}D^{R-I_{odd}-1} \cdot \frac{1}{2}D^R
        = \frac{1}{4}D^{k-1}.
    \]

    To prove \Cref{lem:G-is-park} take arbitrary $J\subseteq V(G)$. Suppose that for some vertices $v$ and $u$ the $J$-link of $((\pp^{i, R-I_{odd}-1}_{x_i} - F_i)^{rev}\circ (x_i, m_i, y_i)\circ (\Ll^i_{y_i} - F_i))_{u,v}$ is non-empty, which means that there are some non-faulty paths $P_x\in \pp^{i, R-I_{odd}-1}_{x_i} - F_i$ from $x_i$ to $u$ and $P_y \in \Ll^i_{y_i} - F_i$ from $y_i$ to $v$, such that $J\subseteq V(P_x)\cup V(P_y)\cup \{m_i\}$. Observe that this implies $J\cap F_i = \emptyset$. Suppose now that $\Gg_{u,v}^i$ is $J$-full. Then by \Cref{lemma:union_bound_scores} we have a non-faulty path $P\in \Gg_{u,v}^i[J]$. This implies by the triangle inequality that $\dist_{H_i-F_i}(x_i, y_i) \leq L(P_x) + L(P)+L(P_y) = R-1-I_{odd} + k+1 + R = 2k$, contradiction with the definition of $F_i$. Thus we get that for any $u,v \in V(G)$ either $((\pp^{i, R-I_{odd}-1}_{x_i} - F_i)^{rev}\circ (x_i, m_i, y_i)\circ (\Ll^i_{y_i} - F_i))_{u,v}[J] = \emptyset$ or $\Gg^i_{u,v}$ is not $J$-full. In the first case clearly $\lsco_J(\Gg^{i+1}_{u,v}) = \lsco_J(\Gg_{u,v}^i) \leq 1$ by the induction hypothesis. In the second case $\lsco_J(\Gg_{u,v}^i)\leq \frac12$ as it is not $J$-full. It remains to upper bound the score of $(\pp^{i, R-I_{odd}-1}_{x_i} - F_i)^{rev}\circ (x_i, m_i, y_i)\circ (\Ll^i_{y_i} - F_i)$. We have the following
    
    \begin{claim}
        For each $J\subseteq V(G)$ 
        $$\lsco_J^{}(((\pp^{i, R-I_{odd}-1}_{x_i} - F_i)^{rev}\circ (x_i, m_i, y_i)\circ (\Ll^i_{y_i} - F_i))_{u,v}) \leq \frac{1}{2}.$$
    \end{claim}
    \begin{proof}

    Note that since all paths in the collection have length $k+1$, the claim holds trivially for any $J$ with $|J| > k+2$. Thus we can assume $|J|\leq k+2$. 
    
    Take arbitrary paths $P_x\in \pp^{i, R-I_{odd}-1}_{x_i,u} - F_i$ and $P_y \in \Ll^i_{y_i,v} - F_i$. Note that $V(P_x) \cap V(P_y) = \emptyset$, because otherwise $\dist_{H_i - F_i}(x_i, y_i) \leq L(P_x) + L(P_y) = k-1$. Thus for each $J$ either  $\lsco_J(P_x^{rev}\circ (x_i, m_i, y_i) \circ P_y) = 0$ or there exists a partition $J = J_1\sqcup J_2$ with $J_1\subseteq V(P_x)$ and $J_2\subseteq V(P_y)$ or a partition $J = J_1\sqcup J_2\sqcup \{m_i\}$ with $J_1\subseteq V(P_x)$ and $J_2\subseteq V(P_y)$. In the second case, we have 

    \begin{multline*}
        \lsco^{}_J(P_x^{rev}\circ (x_i, m_i, y_i) \circ P_y) = \beta(2 f)^{-(k+1-|J|)-1} = \beta(2 f)^{-(L(P_x) + L(P_y) + 2 -|J_1|-|J_2|)-1} = \\=\frac{\beta}{\beta_{R-I_{odd}-1}\beta_R}\frac{1}{2 f}\lsco^{R-I_{odd}-1}_{J_1}(P_x)\lsco^R_{J_2}(P_y) \leq \beta\cdot 10^{k-1}\sum_{J'\subseteq J}\lsco^{R-I_{odd}-1}_{J'}(P_x) \cdot \sum_{J''\subseteq J}\lsco^R_{J''}(P_y).
    \end{multline*}
    
   In the third case 
   
   \begin{multline*}
       \lsco^{}_J(P_x^{rev}\circ (x_i, m_i, y_i) \circ P_y) = \beta(2 f)^{-(k+1-|J|)-1} = \beta(2 f)^{-(L(P_x) + L(P_y) + 2 -|J_1|-|J_2|-1)-1} = \\=\frac{\beta}{\beta_{R-I_{odd}-1}\beta_R}\lsco^{R-I_{odd}-1}_{J_1}(P_x)\lsco^R_{J_2}(P_y) \leq \beta \cdot 10^{k-1}\sum_{J'\subseteq J}\lsco^{R-I_{odd}-1}_{J'}(P_x) \cdot \sum_{J''\subseteq J}\lsco^R_{J''}(P_y).
   \end{multline*}
    Thus anyway 
    \begin{equation*}
        \lsco_J(P_x^{rev}\circ (x_i, m_i, y_i) \circ P_y) \leq \beta\cdot 10^{k-1}\sum_{J'\subseteq J}\lsco^{R-I_{odd}-1}_{J'}(P_x) \cdot \sum_{J''\subseteq J}\lsco^R_{J''}(P_y).
    \end{equation*}
    
    Summing over $P_x$ and $P_y$ we get 
    \begin{multline*}
        \lsco_J^{}(((\pp^{i, R-I_{odd}-1}_{x_i} - F_i)^{rev}\circ (x_i, m_i, y_i)\circ (\Ll^i_{y_i} - F_i))_{u,v}) \leq\\\leq \beta \cdot 10^{k-1}\sum_{P_x, P_y}\sum_{J'\subseteq J}\lsco^{R-I_{odd}-1}_{J'}(P_x) \cdot \sum_{J''\subseteq J}\lsco^R_{J''}(P_y) =\\= \beta \cdot 10^{k-1}\sum_{J', J''\subseteq J}\lsco^{R-I_{odd}-1}_{J'}(\pp^{i, R-I_{odd}-1}_{x_i,u} - F_i)\lsco^{R}_{J''}(\Ll^i_{y_i,v}-F_i) \leq\\ \leq \beta \cdot 10^{k-1}\sum_{J', J''\subseteq J}\lsco^{R-I_{odd}-1}_{J'}(\pp^{i, R-I_{odd}-1}_{x_i,u})\cdot\lsco^{R}_{J''}(\Ll^i_{y_i,v}) \leq \beta \cdot 10^{k-1}2^{2(k+2)} = \frac{1}{2},
    \end{multline*}
    where the last inequality holds since $\pp^{i, R-I_{odd}-1}_{x_i,u}$ and $\Ll^i_{y_i,v}$ are parks w.r.t. $\lsco^{R-I_{odd}-1}$ and $\lsco^R$ (by the pairwise park property \textbf{I3} of $\pp$ and the induction hypothesis) and the assumption $|J|\leq k+2$.
    \end{proof}

    Having the preceding claim then $\Gg^i_{u,v}$ is not $J$-full we obtain that $\lsco_J(\Gg^{i+1}_{u,v}) = \lsco_J(\Gg^{i}_{u,v}) + \lsco_J^{}(((\pp^{i, R-I_{odd}-1}_{x_i} - F_i)^{rev}\circ (x_i, m_i, y_i)\circ (\Ll^i_{y_i} - F_i))_{u,v}) \leq \frac12+ \frac12 = 1$. \Cref{lem:G-is-park} follows. 

    \paragraph{Case 2: $\Ll^i_{y_i}$ is not big enough, but grows substantially ($|\Ll^i_{y_i}| < D^{R}$ and $|\pp\Ll^i|\geq \frac{1}{4}D^{R-1}$).}
    In this case only $\Ll$ changes, so we must verify \Cref{lem:path-collections-growth}(c) and \Cref{lem:L-properties}.

    The growth guarantee is immediate: $|\Ll^{i+1}| - |\Ll^{i}| = |\pp\Ll^i| \geq \frac{1}{4}D^{R-1}$. Note that all the paths that we have added to $\Ll$ were previously not in $\Ll$ since they contain a new edge $\{m_i, y_i\}$. 

    For \Cref{lem:L-properties} first notice that we only modify $\Ll_{y_i}$. We have $|\Ll^{i+1}_{y_i}| = |\Ll^{i}_{y_i}| + |\pp\Ll^i| \leq |\Ll^{i}_{y_i}| + |\pp^{i, R-1}_{m_i}| \leq D^{R} + \frac{D^{R}-1}{D-1} = \frac{D^{R+1}-1}{D-1} \leq 2D^R$, where the second inequality is from \textbf{I5} for $\pp$. This proves \Cref{lem:L-properties}(1).

    For \Cref{lem:L-properties}(2) notice that by construction $\Ll^{i+1}_{y_i} = \bigcup_{j=1}^{p}(y_i, m_{i_j})\circ \pp\Ll^{i_j}$, where $i_1 < i_2 < ... < i_p = i$. The construction also gives that $|\pp\Ll^{i_j}| \geq \frac{1}{4}D^{R-1}$. Consequently, using \Cref{lem:L-properties}(1), $p \leq \frac{2D^R}{\frac{1}{4}D^{R-1}} = 8D$. On the other hand, since $\pp\Ll^j \subseteq \pp^{j,R-1}_{m_j}$, for each $u\neq y_i$ we have $\Ll^{i+1}_{y_i}[u] \subseteq ((y_i,u)\circ \pp^{i, R-1}_{u})\cup \bigcup_{m_{i_j} \neq u}((y_i,m_{i_j})\circ (\pp^{i_j, R-1}_{m_{i_j}}[u]))$. So, by the upper bound on the sizes of $\pp^{i, \ell}_v$ and the weak collective fault tolerance of $\pp$, $|\Ll^{i+1}_{y_i}[u]| \leq \frac{D^R-1}{D-1} + p\cdot 3(R-1)D^{R-2} \leq 24R\cdot D^{R-1}$.

    Now we turn to prove \Cref{lem:L-properties}(3). We claim that for each $J\subseteq V(G)$ and each $v \in V(G)$ either $((y_i, m_i)\circ \pp\Ll^{i})_{y_i, v}[J]$ is empty or $\Ll^{i}_{y_i,v}$ is not $J$-full w.r.t. $\lsco^R$. Indeed, assume both these conditions fail. Then there is a path $P \in \pp\Ll^{i}$ ending in $v$ and $J\subseteq V(P) \cup \{y_i\}$ such that $\Ll^{i}_{y_i,v}$ is $J$-full. But since $\Ll^{i}_{y_i,v}$ is $J$-full, $P$ should have been put in $\pp\Ggh^i$, and not $\pp\Ll^i$, contradiction. If the first condition holds, we have $\lsco_J^R(\Ll^{i+1}_{y_i,v}) = \lsco_J^R(\Ll^{i}_{y_i,v}) \leq 1$ by induction hypothesis. In the second case 

    \begin{multline*}
        \lsco_J^R(\Ll^{i+1}_{y_i,v}) = \lsco_J^R(\Ll^{i}_{y_i,v}) + \lsco^R_J((y_i, m_i)\circ \pp\Ll^i_{m_i, v}) \leq \frac{1}{2} + \frac{1}{10}(\lsco^{R-1}_J(\pp\Ll^i_{m_i, v}) + \lsco^{R-1}_{J-y_i}(\pp\Ll^i_{m_i, v})) \leq\\\leq \frac{1}{2} + \frac{1}{10}(\lsco^{R-1}_J(\pp^{i,R-1}_{m_i, v}) + \lsco^{R-1}_{J-y_i}(\pp^{i,R-1}_{m_i, v})) \leq \frac12 + \frac{2}{10} \leq 1,
    \end{multline*}
    where the first inequality follows from \Cref{cl:concat} and the third one from the pairwise park property \textbf{I3} of $\pp$. This proves \Cref{lem:L-properties}(3).

    \paragraph{Case 3: the FT connections witnessed by $\pp \Ggh^i$, force $\Ggh$ to grow.}
    In this case only $\Ggh$ changes, so we must verify \Cref{lem:path-collections-growth}(b) and \Cref{lem:Gtilde-is-park}.

    \emph{Verifying the growth guarantee.} First, we show $|\pp\Ggh^i| \geq \frac14 D^{R-1}$. Since $\pp\Ll^i$ and $\pp\Ggh^i$ partition the surviving paths from $\pp^{i, R-1}_{m_i}$:
    \begin{align}
        |\pp\Ll^{i}| + |\pp\Ggh^i|
        &= |\pp^{i,R-1}_{m_i} - F_i| \notag\\
        &\geq |\pp^{i,R-1}_{m_i}| - \sum_{v\in F_i}|\pp^{i,R-1}_{m_i}[v]|
            && \text{(union bound)} \notag\\
        &\geq D^{R-1} - f \cdot 3(R-1) \cdot D^{R-2}
            && \text{(weak collective FT)} \notag\\
        &\geq \frac{1}{2}D^{R-1}
            && \text{(since $D \geq 6Rf$).}
    \end{align}
    Since we're not in Case 2, we have $|\pp\Ll^{i}| < \frac{1}{4}D^{R-1}$, hence $|\pp\Ggh^i| \geq \frac{1}{4}D^{R-1}$.

    Now we bound the growth of $\Ggh$:
    \begin{align}
        |\Ggh^{i+1}| - |\Ggh^i|
        &= |\pp^{i, R-I_{odd}}_{x_i} - F_i| \cdot |\pp\Ggh^i| \notag\\
        &\geq \left(D^{R-I_{odd}} - 3f(R-I_{odd})D^{R-I_{odd}-1}\right) \cdot \frac{1}{4}D^{R-1} \notag\\
        &\geq \frac{1}{2}D^{R-I_{odd}} \cdot \frac{1}{4}D^{R-1}
            && \text{(since $D \geq 6Rf$)} \notag\\
        &= \frac{1}{8} D^{k-1}.
    \end{align}

    Now we shall prove \Cref{lem:Gtilde-is-park}. First we prove the following 
    \begin{claim}
        For any $u,v\in V(G)$ and $J\subseteq V(G)$ either $((\pp^{i, R-I_{odd}}_{x_i} - F_i)^{rev}\circ(x_i, m_i)\circ \pp\Ggh^i)_{u,v}[J] = \emptyset$ or $\Ggh^i_{u,v}$ is not $J$-full.
    \end{claim}
    \begin{proof}
        Suppose that both options fail for some $u,v$ and $J$. It means that there are paths $P_x \in (\pp^{i, R-I_{odd}}_{x_i} - F_i)_{x_i,u}$ and $P_m \in \pp\Ggh^i_{m_i, v}$ such that $J\subseteq V(P_x)\cup V(P_m)$, and so $J\cap F_i = \emptyset$. Since $\Ggh^i_{u,v}$ is $J$-full, by \Cref{lemma:union_bound_scores}, there is a non-faulty path $P' \in \Ggh^i_{u,v}[J]$. Since $P_m \in \pp\Ggh^i$, there is some $J' \subseteq V(P_m) \cup \{y_i\}$ (and so $J'\cap F_i=\emptyset$) such that $\Ll^i_{y_i,v}$ is $J'$-full. Thus, again by \Cref{lemma:union_bound_scores}, there is a non-faulty path $P''\in \Ll^i_{y_i,v}[J']$. Overall, we obtain that, by triangle inequality, $\dist_{H_i - F_i}(x_i, y_i) \leq L(P_x) + L(P') + L(P'') = R - I_{odd} + k + R = 2k$,  contradiction.
    \end{proof}

    If for vertices $u,v$ and a vertex subset $J$ the first option in the above claim is in place, we have $\hat{\lsco}_J(\Ggh^{i+1}_{u,v}) = \hat{\lsco}_J(\Ggh^{i}_{u,v}) \leq 1$ by the induction hypothesis. In the second case we have $\hat{\lsco}_J(\Ggh^{i+1}_{u,v}) = \hat{\lsco}_J(\Ggh^{i}_{u,v}) + \hat\lsco_J(((\pp^{i, R-I_{odd}}_{x_i} - F_i)^{rev}\circ(x_i, m_i)\circ \pp\Ggh^i)_{u,v}) \leq \frac{1}{2} + \hat\lsco_J(((\pp^{i, R-I_{odd}}_{x_i} - F_i)^{rev}\circ(x_i, m_i)\circ \pp\Ggh^i)_{u,v})$. Consequently, it remains to prove 

    \begin{claim}
        $$\hat\lsco_J(((\pp^{i, R-I_{odd}}_{x_i} - F_i)^{rev}\circ(x_i, m_i)\circ \pp\Ggh^i)_{u,v})\leq \frac{1}{2}.$$
    \end{claim}
    
    \begin{proof}
        Since all paths in the collection are of length $k$, we may assume $|J| \leq k+1$, otherwise the claim holds trivially.
        
        Take arbitrary $P_x \in \pp^{i, R-I_{odd}}_{x_i,u} - F_i$ and $P_m \in \pp\Ggh^i_{m_i,v}$. Note that $V(P_x)\cap V(P_m) = \emptyset$, since otherwise $\dist_{H_i-F_i}(x_i, m_i) \leq L(P_x) + L(P_m) = R-I_{odd} + R-1 = k-1$. Thus either $J\not\subseteq V(P_x)\cup V(P_m)$ and $\hat\lsco_J(P_x^{rev}\circ (x_i, m_i)\circ P_m) = 0$, or there is a partition $J = J_1 \sqcup J_2$ with $J_1 \subseteq V(P_x)$ and $J_2 \subseteq V(P_m)$. In the latter case 

        \begin{multline*}
            \hat\lsco_J(P_x^{rev}\circ (x_i, m_i)\circ P_m) = \hat\beta(2f)^{-(L(P_x) + L(P_m) + 1 - |J_1| - |J_2|)-1} =\\= \hat{\beta}\cdot 10^{k-1}\lsco_{J_1}^{R-I_{odd}}(P_x)\lsco_{J_2}^{R-1}(P_m) \leq \hat{\beta}\cdot 10^{k-1}\sum_{J'\subseteq J}\lsco_{J'}^{R-I_{odd}}(P_x)\sum_{J''\subseteq J}\lsco_{J''}^{R-1}(P_m).
        \end{multline*}
    
        In any case, $\hat\lsco_J(P_x^{rev}\circ (x_i, m_i)\circ P_m)\leq \hat{\beta}\cdot 10^{k-1}\sum_{J'\subseteq J}\lsco_{J'}^{R-I_{odd}}(P_x)\sum_{J''\subseteq J}\lsco_{J''}^{R-1}(P_m)$. Summing over $P_x$ and $P_m$ we get
        \begin{multline*}
            \hat\lsco_J(((\pp^{i, R-I_{odd}}_{x_i} - F_i)^{rev}\circ(x_i, m_i)\circ \pp\Ggh^i)_{u,v})\leq \hat{\beta}\cdot 10^{k-1}\sum_{P_x,P_m, J',J''}\lsco_{J'}^{R-I_{odd}}(P_x)\lsco_{J''}^{R-1}(P_m) =\\= \hat{\beta}\cdot 10^{k-1}\sum_{J',J''}\lsco_{J'}^{R-I_{odd}}((\pp^{i, R-I_{odd}}_{x_i} - F_i)_{x_i,u}) \lsco_{J''}^{R-1}(\pp\Ggh^i_{m_i, v}) \leq\\\leq \hat{\beta}\cdot 10^{k-1}\sum_{J',J''\subseteq J}\lsco_{J'}^{R-I_{odd}}(\pp^{i, R-I_{odd}}_{x_i,u}) \lsco_{J''}^{R-1}(\pp^{i,R-1}{m_i, v}) \leq \hat{\beta}\cdot 10^{k-1}\cdot 2^{2(k+1)} = \frac{1}{2},
        \end{multline*}
        where the last inequality holds because of the pairwise park property \textbf{I3} of $\pp$ and the assumption $|J| \leq k+1$. 
    \end{proof}
\end{proof}

\begin{proof}[Proof of \Cref{thm:2-2k-ft-greedy}]
    Set $D = 128k^2fn^{\frac{1}{k}}$.

    By \Cref{lem:G-is-park} and \Cref{obs:size-bounds-park-links}  we have $|\Gg^{m+1}_{u,v}| = |\Gg^{m+1}_{u,v}[\{u,v\}]| \leq \frac{1}{\beta}(2f)^{k + 1 - 2 + 1} \leq 2^{7k + 1}f^{k}$. Consequently, $|\Gg^{m+1}|\leq n^2\cdot 2^{7k + 1}f^{k}$.
    Thus, only $\frac{2^{7k + 3}f^{k}n^2}{D^{k-1}} = O(f n^{1+\frac{1}{k}})$ steps may satisfy \Cref{lem:path-collections-growth}(a). 

    By \Cref{lem:Gtilde-is-park} and \Cref{obs:size-bounds-park-links} 
    we have $|\Ggh^{m+1}_{u,v}| = |\Ggh^{m+1}_{u,v}[\{u,v\}]| \leq \frac{1}{\hat\beta}(2f)^{k - 2 + 1} \leq 2^{7k -2}f^{k-1}$
    and thus $|\Ggh^{m+1}|\leq 2^{7k -2}f^{k-1}n^2$. Consequently, only $\frac{2^{7k + 1}f^{k-1}n^2}{D^{k-1}} = O(n^{1+\frac{1}{k}})$ steps may satisfy \Cref{lem:path-collections-growth}(b).

    By \Cref{lem:L-properties}(1), $|\Ll^{m+1}| \leq 2D^Rn$, and so only $8Dn = O(k^2f n^{1+\frac{1}{k}})$ steps may satisfy \Cref{lem:path-collections-growth}(c).

    Summing up we get $|\mathcal{G}| = O(k^2f n^{1+\frac{1}{k}})$, and by \Cref{thm:greedy-path-system-bounds} $|\mathcal{L}| = O(kDn) = O(k^3f n^{1+\frac{1}{k}})$, which concludes the proof.
\end{proof}

\subsection{A Tight Bound for $k=2$}

The theorem proven above gives an additional $f^{1/k}$ factor, compared to the lower bound of \cite{BodwinDPW18}. The following claim shows that for $k = 2$ this lower bound is actually achievable, and suggests there may be room to improve \Cref{thm:2-2k-ft-greedy} further.

\begin{claim}
    The VFT greedy $(2\to 4)$ algorithm with fault parameter $f$, produces a spanner with $O(f^{\frac{1}{2}}n^{\frac{3}{2}})$ edges.
\end{claim}
\begin{proof}
    First note that every addition of a path increases the degree of its middle vertex in $H$  by at least 1. Consequently, there can be only $O(f^{\frac{1}{2}}n^{\frac{3}{2}})$ added paths whose middle vertex had degree $< 4f^{\frac{1}{2}}n^{\frac{1}{2}}$ in $H$ at the moment just before addition. Thus we need to bound only the number of paths for which the converse holds. Let, as usual, $H_i$ be the spanner before addition of the $i$-th path. For any two distinct vertices $u$ and $v$ let $\hat X^i(u,v)$ be the number of distinct $2$-paths connecting $u$ and $v$ in $H^i$. Note that distinct $2$-paths have no common vertices except the ends. Let $X_i(u, v) = \min\{\hat X_i(u,v), f+1\}$. Denote $X_i := \sum_{u\neq v \in V(G)} X_i(u,v)$. Obviously, $X_{i+1} \geq X_i$ for all $i$. Clearly, $X_i \leq (f+1)n^2$ for all $i$, so it is enough to prove the following claim.
    \begin{claim}
        If for the $i$-th added path $P_i = (x_i, m_i, y_i)$ it holds that $\deg_{H_i}(m_i) \geq 4f^{\frac{1}{2}}n^{\frac{1}{2}}$, then $$X_{i+1} - X_i \geq f^{\frac{1}{2}}n^{\frac{1}{2}}.$$
    \end{claim}
    \begin{proof}
        Let $F_i$ be the set of faults that justifies the addition of $P_i$. Let $w \in N_{H_i}(m_i)\setminus (F_i\cup \{x_i, y_i\})$, where $N_{H_i}(m_i)$ is the standard notation for the set of neighbors of $m_i$ in $H_i$. Note that $(x_i, m_i, w)$ is a path in $H_{i+1}$ from $x_i$ to $w$. Thus, $X_i(x_i, w) = X_{i+1}(x_i, w)$ could have happened only if one of the following holds:
        \begin{itemize}
            \item $(x_i, m_i, w)$ is not a new path, i.e. $(x_i, m_i) \in H_i$;
            \item $X_i(x_i, w) = f+1$, meaning that there are at least $f+1$ distinct $2$-paths from $x_i$ to $w$ in $H_i$.
        \end{itemize}
        In either of these cases, we have a non-faulty $2$-path from $x_i$ to $w$ in $H_i$, i.e. $\dist_{H_i-F_i}(x_i, w) \leq 2$. By exactly the same argument, $X_i(y_i, w) = X_{i+1}(y_i, w)$ could have happened only if $\dist_{H_i-F_i}(y_i, w) \leq 2$. But $\dist_{H_i-F_i}(x_i, w) \leq 2$ and $\dist_{H_i-F_i}(y_i, w) \leq 2$ could not be true simultaneously, as otherwise $\dist_{H_i - F_i}(x_i,y_i) \leq 4$ by the triangle inequality and so $F_i$ doesn't justify the addition of $P_i$. Hence for every $w$ either $X_{i+1}(x_i, w) \geq X_{i}(x_i, w) + 1$ or $X_{i+1}(y_i, w) \geq  X_{i}(y_i, w) + 1$. 
        We conclude by noticing that $|N_{H_i}(m_i)\setminus (F_i\cup \{x_i, y_i\})| \geq 4f^{\frac{1}{2}}n^{\frac{1}{2}} - f - 2 \geq f^{\frac{1}{2}}n^{\frac{1}{2}}$, where we use trivial inequalities $1\leq f \leq n$.
    \end{proof}

    This concludes the proof of the claim.
\end{proof}

\paragraph{Disclosure of AI}
We used ChatGPT 5.4 to assist in correcting typos, checking computations and internal consistency, suggesting revisions to improve exposition, and reorganizing mathematical expressions to make proofs easier to read. The tool played no role in the research methodology, analysis, or results presented in this paper. ChatGPT 6 was used for proof auditing for the final version.

%% file: Sampling-lemma.tex
\section{Sampling Lemma and Proof of \Cref{prop:vft-d-to-r-structure}}\label{sec:sampling-appendix}

\begin{lemma}\label{lem-sampling}
Let $C'\ge 2$ be a constant, $r := \lceil \frac{s}{2}\rceil$ and assume $D \ge 12C's^3 f$.
Let $S\subseteq V$ be formed by including each vertex independently with probability
\[
p \;=\; \min\bigg\{11\Bigl(\frac{10\cdot C'sf}{D}\Bigr)^r \ln\!\bigl(100\cdot 2^{r+2}D^rn\bigr), 1\bigg\},
\]
and call vertices in $S$ \emph{centers}.
Let $\pp_v$ be a collection of paths of length $r$ beginning at $v$ that satisfy \textbf{I2-I5} for this $v$ and $\ell = r$ and of size $\geq D^r$.
Let $\prepp_v\subseteq \pp_v$ be the subcollection of $\pp_v$ consisting of paths whose endpoint is a center.
Then with probability at least $1-\frac{1}{100n}$:
\begin{itemize}
    \item $|\prepp_v|\ge \frac{pD^r}{2}$;
    \item for every $J\subseteq V\setminus\{v\}$,
    \[
    |\prepp_v[J]| \;\le\; (C'sf)^{-|J|}\,|\prepp_v|.
    \]
\end{itemize}
\end{lemma}

\begin{proof}[Proof of \Cref{lem-sampling}]

If $p = 1$, $\prepp_v=\pp_v $ deterministically and thus satisfies the desired properties (even stronger properties of $\pp_v$). Thus from now on we assume the other value for $p$.

Fix $J\subseteq V\setminus\{v\}$ that is contained in some path of $\pp_v$.
Write
\[
|\prepp_v[J]| \;=\; \sum_{w\neq v} I(w\in S)\cdot |\pp_{v,w}[J]|.
\]
Define centered summands
\[
X_w^J \;:=\; \bigl(I(w\in S)-p\bigr)\cdot |\pp_{v,w}[J]|.
\]
Then $\mathbb{E}[X_w^J]=0$ and $\sum_{w\neq v}X_w^J = |\prepp_v[J]|-p|\pp_v[J]|$.

\medskip
\noindent\textbf{Step 1: Bernstein parameters.}
By \textbf{I3}, for every $w$,
\[
|X_w^J|\;\le\; |\pp_{v,w}[J]| = |\pp_{v,w}[J\sqcup \{v\}]|\;\le\; \beta_r^{-1}(2f)^{\,r-|J|}\,.
\]
Moreover,
\begin{align*}
\sum_{w\neq v}\mathbb{E}\bigl[(X_w^J)^2\bigr]
&= \sum_{w\neq v} p(1-p)\,|\pp_{v,w}[J]|^2  = \sum_{w\neq v} p(1-p)\,|\pp_{v,w}[J]|\cdot |\pp_{v,w}[J\sqcup \{v\}]|\\
&\le p\sum_{w\neq v} |\pp_{v,w}[J]|\cdot \beta_r^{-1}(2f)^{\,r-|J|}
= p|\pp_v[J]|\cdot \beta_r^{-1}(2f)^{\,r-|J|}.
\end{align*}
Applying Bernstein's inequality \cite{VershyninHigh} yields, for every $t\ge 0$,
\begin{equation}\label{eq:bernstein-sampling}
\Pr\Bigl[\Bigl|\sum_{w\neq v} X_w^J\Bigr|\ge t\Bigr]
\le
2\exp\!\left(
-\frac{t^2}{2\cdot 10^r(2f)^{\,r-|J|}\bigl(p|\pp_v[J]|+t/3\bigr)}
\right).
\end{equation}

\medskip
\noindent\textbf{Step 2: nonempty $J$ (small links).}
Assume $1\le |J|\le r$ (for $|J|>r$ the link is empty).
Set
\[
t_J \;:=\; \frac{pD^r}{4(C'sf)^{|J|}}.
\]
By strong collective fault tolerance of $\pp_v$, for nonempty $J$,
\[
p|\pp_v[J]|\;\le\; 3p\,s^{2|J|}D^{r-|J|}\;\le\; t_J,
\]
where the last inequality follows from the choice of $p$ and $D\ge 12C's^3f$.
Plugging $t=t_J$ into \eqref{eq:bernstein-sampling} and using $p|\pp_v[J]|\le t_J$ gives
\[
\Pr\Bigl[\Bigl|\sum_{w\neq v} X_w^J\Bigr| \ge t_J\Bigr]
\le
2\exp\!\left(
-\frac{3t_J}{8\cdot 10^r(2f)^{\,r-|J|}}
\right)
=
2\exp\!\left(
-\frac{3pD^r}{32\cdot 10^r(2f)^{r}(C's/2)^{|J|}}
\right)
\le \frac{1}{100\cdot 2^{r+1}D^r n}.
\]
Consequently,
\[
\Pr\bigl[\,|\prepp_v[J]|\ge p|\pp_v[J]|+t_J\,\bigr]
\le \frac{1}{100\cdot 2^{r+1}D^r n},
\]
and since $p|\pp_v[J]|\le t_J$ we obtain
\begin{equation}\label{eq:Jlink-bound}
\Pr\bigl[\,|\prepp_v[J]|\ge 2t_J\,\bigr]
\le \frac{1}{100\cdot 2^{r+1}D^r n}.
\end{equation}

\medskip
\noindent\textbf{Step 3: $J=\emptyset$ (lower bound on $|\prepp_v|$).}
Let $t_{\emptyset}:=\frac{pD^r}{2}$.
Since $|\pp_v|\le 2D^r$ by \textbf{I5}, we have $p|\pp_v|\le 2pD^r=4t_{\emptyset}$.
Applying \eqref{eq:bernstein-sampling} with $J=\emptyset$ and $t=t_{\emptyset}$ yields
\[
\Pr\Bigl[\Bigl|\sum_{w\neq v} X_w^\emptyset\Bigr|\ge t_{\emptyset}\Bigr]
\le
2\exp\!\left(
-\frac{t_{\emptyset}^2}{2\cdot 10^r(2f)^{r}\bigl(p|\pp_v|+t_{\emptyset}/3\bigr)}
\right)
\le
2\exp\!\left(
-\frac{3pD^r}{26\cdot 10^r(2f)^{r}}
\right)
\le \frac{1}{100\cdot 2^{r+1}D^r n}.
\]
Using $|\pp_v|\ge D^r$, we conclude
\begin{equation}\label{eq:size-lb}
\Pr\Bigl[\,|\prepp_v|\le \frac{pD^r}{2}\Bigr]
\le \frac{1}{100\cdot 2^{r+1}D^r n}.
\end{equation}

\medskip
\noindent\textbf{Step 4: union bound and conclusion.}
The collection $\pp_v$ has at most $2D^r$ paths, hence the number of distinct nonempty subsets $J\subseteq V(\pp_v)\setminus\{v\}$
with $|J|\le r$ is at most $2^{r+1}D^r-1$.
By \eqref{eq:Jlink-bound} and a union bound over all such $J$, together with \eqref{eq:size-lb}, with probability at least
$1-\frac{1}{100n}$ we have simultaneously:
\[
|\prepp_v|\ge \frac{pD^r}{2}
\qquad\text{and}\qquad
\forall\, J\subseteq V\setminus\{v\},\ 1\le |J|\le r:\
|\prepp_v[J]|\le 2t_J=\frac{pD^r}{2(C'sf)^{|J|}}.
\]
On this event,
\[
|\prepp_v[J]|
\le
\frac{pD^r}{2(C'sf)^{|J|}}
\le
(C'sf)^{-|J|}\,|\prepp_v|,
\]
so $\{V(P)\setminus\{v\}\}_{P\in\prepp_v}$ satisfies the second part of the lemma.
\end{proof}

\begin{proof}[Proof of \Cref{prop:vft-d-to-r-structure}]
    \Cref{alg:greedy-path-system} provides the partition of the added paths into $\mathcal L$ and $\mathcal G$, as well as the sequence of path collections $\pp^i$. The first point then follows directly from \Cref{thm:greedy-path-system-bounds}.

    Towards the second point, by applying \Cref{cl:obvious-prop-g} for  $P_i \in \mathcal{G}$, every $s$-critical edge $e = \{u,v\}$ on $P_i$ satisfies $|\pp_u^{i, R}| \geq D^R$ and $|\pp_v^{i, R}| \geq D^R$. Recall that $P_i=(x_0,\ldots,x_d)$.

    We set $\ell := \min\{j: |\pp_{x_j}^{i, R}| \geq D^R\}$ and $q := \max\{j: |\pp_{x_j}^{i, R}| \geq D^R\}$. Since $P_i$ has at least one critical edge, $\ell$ and $q$ are well defined and distinct. 

    By definition, there can be no critical edges on $P_i$ between $x_0$ and $x_\ell$ and between $x_q$ and $x_{d}$. We can immediately conclude that $\dist_{H_i-F_i}(x_0,x_\ell) \leq \ell \cdot s$ and $\dist_{H_i-F_i}(x_q, x_{d}) \leq (d-q)\cdot s$. 

    For each vertex $v\in V(G)$ let $\tilde{\pp}_v := \pp^{i_v,R}_v$, where $i_v = \min\{j : |\pp^{j,R}_v|\geq D^R\}$. If $i_v$ does not exist, we will never need this definition, and for concreteness we set $\tilde{\pp}_v := \emptyset$. Let $S$ be a $p$-random subset of $V(G)$ with
\[
p = 11\Bigl(\frac{10\cdot C'sf}{D}\Bigr)^R \ln\!\bigl(100\cdot 2^{R+2}D^Rn\bigr).
\]
    For every $v$ let $\prepp_v$ be the paths of $\tilde\pp_v$ that end in the points of $S$. By \Cref{lem-sampling} and the union bound, with probability at least $\frac{99}{100}$ the collection $\prepp_v$ is non-empty and has the property from the second bullet of the second point for all $v\in V(G)$ except those for which $\tilde{\pp}_v = \emptyset$. Also $|S| \leq 2pn$ with probability $\geq \frac{1}{2}$ by Markov's inequality. Thus we can fix one $S$ that satisfies both conditions.
    
    Since $|\pp_{x_\ell}^{i, R}| \geq D^R$ and $|\pp_{x_q}^{i, R}| \geq D^R$, we have $\tilde \pp_{x_\ell} \subseteq \pp_{x_\ell}^{i, R}$ and $\tilde \pp_{x_q} \subseteq \pp_{x_q}^{i, R}$, and thus $\tilde \pp_{x_\ell}$ and $\tilde \pp_{x_q}$ lie in $H_i$. Then, clearly, $\prepp_{x_\ell}$ and $\prepp_{x_q}$ lie in $H_i$ too.
\end{proof}

%% file: Equivalence-d-to-r-alpha-beta.tex
\section{Equivalence of $d\to r$ and $(\alpha,\beta)$-spanners}\label{sec:equiv-d-to-r-alpha-beta}

We show that obtaining a $d\to r$ spanner with good size bounds implies an $(r/d,r)$-spanner with essentially the same bounds. Assume $d\leq r$.

\begin{claim}\label{clm:reduction-spanners-one}
    Suppose that a function $s:\mathbb N \to \mathbb N$  is such that every graph on $n$ vertices has a $d \to r$ spanner with at most $s(n)$ edges. Then for every $\ell \in [d-1]$, every graph on $n$ vertices also has a $(d - \ell) \to (r-\ell)$ spanner with at most $s((\ell+1)n)$ edges.
\end{claim}

\begin{proof}
    Let $G$ be an arbitrary graph on $n$ vertices. Let $G'$ be a graph obtained by attaching to each vertex of $G$ a path of length $\ell$ with all these paths disjoint. For a vertex $v\in V(G)$, denote by $v'$ the last vertex of the path attached to $v$. Graph $G'$ has $(\ell+1) n$ vertices, and so it has a $d\to r$ spanner $H'$ of size $\leq s((\ell+1) n)$. Consider the subgraph $H$ of $G$ with $V(H) = V(G)$ and $E(H) = E(G) \cap E(H')$. Clearly, $H$ has at most $|E(H')| \leq s((\ell+1) n)$ edges. It remains to prove that it is the desired $(d - \ell) \to (r - \ell)$ spanner of $G$. Let $u, v \in V(G)$ be such that $\dist_{G}(u,v) = d - \ell$. Then, evidently, $\dist_{G'}(u,v') = (d - \ell) +\ell = d$, and so, since $H'$ is the $d\to r$ spanner of $G'$, we have that $\dist_{H'}(u,v') \leq r$. Clearly, the shortest path from $u$ to $v'$ consists of some path $P$ in $G$ from $u$ to $v$ and the added $\ell$-path from $v$ to $v'$. Then $\dist_{H'}(u,v') = L(P) + \ell \leq r$, and so $\dist_{H}(u,v) \leq L(P)\leq r - \ell$, as desired. 
\end{proof}

\begin{claim}\label{clm:reduction-spanners-two}
    The union over $\ell \in \{0, 1, ...,d-1\}$ of the $(d - \ell) \to (r - \ell)$ spanners of a given graph $G$ is an $(\frac{r}{d}, (r-d)(1 - \frac{1}{d}))$ spanner of $G$.
\end{claim}

\begin{proof}
    Let $H_\ell$ be a $(d - \ell) \to (r - \ell)$ spanner of $G$ for each $\ell\in\{0,\ldots,d-1\}$, and let $H = \bigcup_{\ell = 0}^{d-1}H_\ell$. Let $u\neq v\in V(G)$ be arbitrary. Then there are integers $q \in \mathbb Z_{\geq 0}$ and $b \in \{0, ..., d-1\}$ such that $\dist_G(u,v) = q\cdot d + b$. Let $P$ be the shortest path in $G$ connecting $u$ and $v$ and let $u = x_0, x_1, ..., x_q$ be the vertices on $P$ such that for all $1\leq i \leq q$ we have $\dist_{G}(x_{i-1}, x_{i}) = d$ and $\dist_G(x_q, v) = b$. Then we immediately have that $\dist_{H_0}(x_{i-1}, x_{i}) \leq r$ and $\dist_{H_{d-b}}(x_q, v) \leq r - (d - b)$ if $b > 0$ (otherwise $x_q = v$). Consequently, by the triangle inequality, $\dist_{H}(u,v)\leq q\cdot r + r - (d - b) = \frac{r}{d}\cdot \dist_G(u,v) + r - d - b(\frac{r}{d} - 1) \leq \frac{r}{d}\cdot \dist_G(u,v) + (r-d)(1 - \frac{1}{d})$ if $b > 0$ and $\dist_{H}(u,v)\leq q\cdot r = \frac{r}{d}\cdot \dist_G(u,v)$ if $b = 0$.
\end{proof}

As $(r-d)(1 - \frac{1}{d}) \leq r$ we have the following corollary:

\begin{corollary}
    Suppose that an increasing function $s:\mathbb N \to \mathbb N$ is such that every graph on $n$ vertices has a $d \to r$ spanner with at most $s(n)$ edges. Then every graph has an $(\frac{r}{d}, r)$-spanner with at most $ \sum_{\ell = 0}^{d-1}s((\ell+1) n) \leq d\cdot s(dn)$ edges. 
\end{corollary}